\documentclass[acmsmall,screen,nonacm]{acmart}

\usepackage{xspace,url,hyperref,doi,wrapfig,stmaryrd, graphicx,xparse,etoolbox}
\usepackage{caption,subcaption}
\usepackage{style/julia}
\usepackage{mathpartir,listings}
\usepackage[inline]{enumitem}
\usepackage{cancel}
\usepackage[xcolor,rightbars]{changebar}

\newcommand{\EM}[1]{\ensuremath{#1}\xspace}
\newcommand{\xt}[1]{{\sf{#1}}}

\newcommand{\EMxt}[1]{\EM{\xt{#1}}}

\newcommand{\any}    {\EMxt{any}}  
\renewcommand{\S}{\EMxt{S}}           
\newcommand{\m}   {\EMxt m}          
\newcommand{\emp}{\EM{\epsilon}}  

\DeclareDocumentCommand\TR{om}{\EM{\IfNoValueTF{#1}{\PackageWarning{}{Undefined Type System}}{#1}\llbracket #2 \rrbracket}}

\DeclareDocumentCommand\TRG{omm}{\EM{\IfNoValueTF{#1}{\PackageWarning{}{Undefined Type System}}{#1}\llbracket #2 \rrbracket_{#3}}}
\DeclareDocumentCommand\TAG{ommm}{\EM{\IfNoValueTF{#1}{\PackageWarning{}{Undefined Type System}}{#1}\llparenthesis #2 \rrparenthesis_{#3}^{#4}}}

\newcommand{\n}   {\EMxt n}

\DeclareDocumentCommand\a{o}{\IfNoValueTF{#1}{\EMxt {a}}{\EM{\xt {a}_{#1}}}}
\DeclareDocumentCommand\ap{o}{\IfNoValueTF{#1}{\EMxt {a'}}{\EM{\xt {a'}_{#1}}}}
\DeclareDocumentCommand\app{o}{\IfNoValueTF{#1}{\EMxt {a''}}{\EM{\xt {a''}_{#1}}}}
\DeclareDocumentCommand\t{o}{\IfNoValueTF{#1}{\EMxt t}{\EM{\xt t_{#1}}}}
\DeclareDocumentCommand{\tp}{o}{\IfNoValueTF{#1}{\EM{ \xt t' }}{\EM{\xt t_{#1}'}}}
\DeclareDocumentCommand{\tpp}{o}{\IfNoValueTF{#1}{\EM{ \xt t'' }}{\EM{\xt t_{#1}''}}}
\DeclareDocumentCommand{\tppp}{o}{\IfNoValueTF{#1}{\EM{ \xt t''' }}{\EM{\xt t_{#1}'''}}}
\DeclareDocumentCommand{\e}{o}{\IfNoValueTF{#1}{\EM{ \xt e }}{\EM{\xt e_{#1}}}}
\DeclareDocumentCommand{\ep}{o}{\IfNoValueTF{#1}{\EM{ \xt e' }}{\EM{\xt e_{#1}'}}}
\DeclareDocumentCommand{\epp}{o}{\IfNoValueTF{#1}{\EM{ \xt e'' }}{\EM{\xt e''_{#1}}}}
\DeclareDocumentCommand{\eppp}{o}{\IfNoValueTF{#1}{\EM{ \xt e''' }}{\EM{\xt e'''_{#1}}}}
\DeclareDocumentCommand{\fd}{o}{\IfNoValueTF{#1}{\EM{ \xt{fd} }}{\EM{\xt{fd}_{#1}}}}
\DeclareDocumentCommand{\fdp}{o}{\IfNoValueTF{#1}{\EM{ \xt{fd}' }}{\EM{\xt{fd}_{#1}'}}}
\DeclareDocumentCommand{\fdpp}{o}{\IfNoValueTF{#1}{\EM{ \xt{fd}'' }}{\EM{\xt{fd}_{#1}''}}}
\DeclareDocumentCommand{\fdppp}{o}{\IfNoValueTF{#1}{\EM{ \xt{fd}''' }}{\EM{\xt{fd}_{#1}'''}}}
\DeclareDocumentCommand{\md}{o}{\IfNoValueTF{#1}{\EM{ \xt{md} }}{\EM{\xt{md}_{#1}}}}
\DeclareDocumentCommand{\f}{o}{\IfNoValueTF{#1}{\EM{ \xt f }}{\EM{\xt f_{#1}}}}
\DeclareDocumentCommand{\mdp}{o}{\IfNoValueTF{#1}{\EM{ \xt{md}' }}{\EM{\xt{md}_{#1}'}}}
\DeclareDocumentCommand{\mdpp}{o}{\IfNoValueTF{#1}{\EM{ \xt{md}'' }}{\EM{\xt{md}_{#1}''}}}
\DeclareDocumentCommand{\mdppp}{o}{\IfNoValueTF{#1}{\EM{ \xt{md}''' }}{\EM{\xt{md}_{#1}'''}}}
\DeclareDocumentCommand{\C}{o}{\IfNoValueTF{#1}{\EM{ \xt{C} }}{\EM{\xt{C}_{#1}}}}
\DeclareDocumentCommand{\mt}{o}{\IfNoValueTF{#1}{\EM{ \xt{mt} }}{\EM{\xt{mt}_{#1}}}}
\DeclareDocumentCommand{\mtp}{o}{\IfNoValueTF{#1}{\EM{ \xt{mt}' }}{\EM{\xt{mt}_{#1}'}}}
\DeclareDocumentCommand{\mtpp}{o}{\IfNoValueTF{#1}{\EM{ \xt{mt}'' }}{\EM{\xt{mt}_{#1}''}}}
\DeclareDocumentCommand{\mtppp}{o}{\IfNoValueTF{#1}{\EM{ \xt{mt}''' }}{\EM{\xt{mt}_{#1}'''}}}
\DeclareDocumentCommand{\D}{o}{\IfNoValueTF{#1}{\EM{ \xt{D} }}{\EM{\xt{D}_{#1}}}}
\DeclareDocumentCommand{\Dp}{o}{\IfNoValueTF{#1}{\EM{ \xt{D'} }}{\EM{\xt{D'}_{#1}}}}
\DeclareDocumentCommand{\Dpp}{o}{\IfNoValueTF{#1}{\EM{ \xt{D''} }}{\EM{\xt{D''}_{#1}}}}

\renewcommand{\k} {\EMxt k}

\newcommand{\E}   {\EMxt {E}}

\newcommand{\Rule}[4][]{\inferrule*{#3}{#4}}

\makeatletter
\newcommand{\efqn}{
\foreach\n in {1,...,\@listdepth}{foo}
}
\makeatother
\newlist{myEnumerate}{enumerate}{10}

\DeclareDocumentCommand\stepp{o}{\IfNoValueTF{#1}{\item}{\item\hypertarget{proofstep:#1}{}\label{proofstep:#1}}}

\newcommand{\ebox}[1]{\fbox{#1}\hfill\vspace{-1em}\centering}

\newcounter{rules}

\newcommand{\ol}[1]{\EM{\overline{#1}}}

\newcounter{theo}

\newcounter{lem}

\newcommand{\stepwidth}{.55\textwidth-\leftskip}
\newcommand{\numwidth}{3.5em}

\newbool{stepspace}
\newcommand{\stepspace}
  {\ifbool{stepspace}{\vspace{0.4em}\boolfalse{stepspace}}{}\par}
\newcommand{\step}[2]
  {\stepspace\par\stepnum\makebox[\stepwidth][l]{#1} by #2}

\newsavebox{\stepsby}

\newcommand{\done}[1]{\step{done}{#1}}

\DeclareFontFamily{OMS}{cmtt}{\skewchar\font48 }
\DeclareFontShape{OMS}{cmtt}{m}{n}%
   {<->ssub*cmsy/m/n}{}
\DeclareFontShape{OMS}{cmtt}{m}{it}%
   {<->ssub*cmsy/m/n}{}
\DeclareFontShape{OMS}{cmtt}{m}{sl}%
   {<->ssub*cmsy/m/n}{}
\DeclareFontShape{OMS}{cmtt}{m}{sc}%
   {<->ssub*cmsy/m/n}{}
\DeclareFontShape{OMS}{cmtt}{bx}{n}%
   {<->ssub*cmsy/b/n}{}
\DeclareFontShape{OMS}{cmtt}{bx}{it}%
   {<->ssub*cmsy/b/n}{}
\DeclareFontShape{OMS}{cmtt}{bx}{sl}%
   {<->ssub*cmsy/b/n}{}
\DeclareFontShape{OMS}{cmtt}{bx}{sc}%
   {<->ssub*cmsy/b/n}{}

\newcommand{\IGNOREUNLESSNEEDED}[1]{}
\newcommand{\figref}[1]{Fig.~\ref{#1}\xspace}
\newcommand{\lemref}[1]{Lem.~\ref{#1}\xspace}
\newcommand{\thmref}[1]{Thm.~\ref{#1}\xspace}

\newcommand{\secref}[1]{Sec.~\ref{#1}\xspace}
\newcommand{\defref}[1]{Def.~\ref{#1}\xspace}
\newcommand{\appref}[1]{Appendix~\ref{#1}\xspace}

\renewcommand{\EM}[1]{\ensuremath{#1}\xspace}   

\newcommand{\SF}[1]{\mathsf{#1}}                
\newcommand{\SC}[1]{\textsc{#1}}
\newcommand{\tinyb}[1]{\scalebox{0.8}{{\normalsize #1}}}  
\renewcommand{\v}[1]{\EM{{\tinyb{\%}}\SF{#1}}}  
\newcommand{\val}{\EM{\SF{v}}}                  
\newcommand{\ass}[2]{\EM{\v{#1} \leftarrow #2}} 
\renewcommand{\int}{\EM{\SF{Int}}}              
\renewcommand{\any}{\EM{\SF{Any}}}              
\newcommand{\ty}{\EM{\SF{ty}}}                   
\newcommand{\aty}{\EM{\SF{A}}}                  
\newcommand{\call}[2]{\EM{\SF{#1}(#2)}}         
\newcommand{\cond}[4]{\EM{\ass{#1}{\v{#2} ~?~ #3 : #4}}}  
\newcommand{\get}[2]{\EM{\v{#1}[\SF{#2}]}}      
\renewcommand{\i}{\EM{\SF{i}}}                  
\renewcommand{\j}{\EM{\SF{j}}}                  
\renewcommand{\m}{\EM{\SF{m}}}                  
\renewcommand{\k}{\EM{\SF{k}}}                  
\newcommand{\p}{\EM{\SF{p}}}                    
\renewcommand{\l}{\EM{\SF{l}}}                  
\newcommand{\st}{\EM{\SF{st}}}                  
\newcommand{\env}{\EM{\SF{E}}}                  
\newcommand{\frm}{\EM{\SF{F}}}                  
\newcommand{\mtbl}{\EM{\SF{M}}}                 
\newcommand{\tytbl}{\EM{\SF{D}}}                
\newcommand{\config}[2]{\EM{#1,\, #2}}          
\newcommand{\configd}{\config{\frm}{\mtbl}}   
\newcommand{\stk}[2]{\EM{#1 \cdot #2}}          
\renewcommand{\step}{\EM{~\rightarrow~}}        
\newcommand{\stepmul}{\EM{~\rightarrow^*~}}   
\newcommand{\stepd}{\EM{~{\rightarrow_{\mathcal{D}}}~}} 
\newcommand{\stepj}{\EM{~{\rightarrow_{\SF{JIT}}}~}} 
\newcommand{\stepdmul}{\EM{~{\rightarrow^*_{\mathcal{D}}}~}}
\newcommand{\stepjmul}{\EM{~{\rightarrow^*_{\SF{JIT}}}~}}
\newcommand{\idx}[2]{\EM{#1[{\SF{#2}}]}}        
\newcommand{\construct}[2]{\EM{#1(#2)}}   
\newcommand{\msig}[2]{\EM{#1!{#2}}}         
\newcommand{\meth}[3]{\EM{#1!{#2}\,}}
\newcommand{\direct}[3]{\EM{\msig{#1}{#2}(#3)}}      
\newcommand{\Ty}{\EM{\SF{T}}}                   
\newcommand{\last}[1]{\EM{\mathit{last}(#1)}}   
\renewcommand{\done}{\EM{\epsilon}}
\newcommand{\jules}{\EM{\SF{Jules}}}
\renewcommand{\c}[1]{\lstinline{#1}\xspace}
\newcommand{\VD}{\vdash}
\newcommand{\VDnd}{\EM{\vdash^{{\mathcal{D}}}}}

\renewcommand{\n}{\EM{\SF n}}

\newcommand{\main}{\EM{\SF{main}}}
\newcommand{\origmtbl}[1]{\EM{\lfloor #1 \rfloor}}
\newcommand{\compst}[7]{\EM{#1\ \VD\ \config{#2}{#3}\ \leadsto\
  \config{#5}{#6}}}
\newcommand{\devirtst}[3]{\EM{#1\ \VDnd_{#2}\ #3}}
\newcommand{\devirtm}[1]{\EM{\VDnd\ #1}}

\newcommand{\eqdirop}{\triangleright}
\newcommand{\eqst}[5]{\EM{#1\ \vdash_{#2,#3}\ #4\ \eqdirop\ #5}}
\newcommand{\eqstd}[3]{\eqst{#1}{\mtbl}{\mtbl'}{#2}{#3}}
\newcommand{\eqmtbl}[2]{\EM{#1 \eqdirop #2}}
\newcommand{\eqmtbld}{\eqmtbl{\mtbl}{\mtbl'}}

\newcommand{\typeinfst}[3]{\EM{\VD^{\typeinfop}_{#1}\,#2\,<:\,#3}}
\newcommand{\typeinfstd}[2]{\typeinfst{\mtbl}{#1}{#2}}
\newcommand{\optst}[5]{\EM{#3\,\VD_{#1,#2}\, #4\ \eqdirop\ #5}}
\newcommand{\optstd}[3]{\optst{\mtbl}{\mtbl'}{#1}{#2}{#3}}

\newcommand{\optconfig}[5]{\EM{\config{#1}{#2}\ \eqdirop\ \config{#3}{#4}\ <:\ #5}}
\DeclareMathOperator{\typeof}{\mathit{typeof}}

\DeclareMathOperator{\dispatchop}{{\mathcal D}}

\DeclareMathOperator{\body}{\mathit{body}}
\DeclareMathOperator{\signature}{\mathit{signtr}}
\DeclareMathOperator{\origsignature}{\mathit{o-signtr}}

\DeclareMathOperator{\typeinfop}{\mathcal{I}}
\DeclareMathOperator{\jitop}{\mathit{jit}}
\newcommand{\dispatch}[3]{\EM{\dispatchop(#1,\SF{#2},{#3})}} 
\newcommand{\typeinf}[3]{\EM{\typeinfop(#1, #2, #3)}} 
\newcommand{\jit}[0]{\EM{\jitop}} 

\newtheorem{requirement}{Requirement}[section]

\newcommand{\goodpkgsnum}{760\xspace}
\newcommand{\juliaversion}{Julia 1.5.4\xspace}

\newcommand{\ExtendedVersion}{1}    
\newcommand{\HightlightChanges}{0}  

\newcommand{\PAPERVERSION}[2]{
\ifnum1=\ExtendedVersion\relax
\begin{changebar}
#2%
\end{changebar}\xspace
\else
#1\xspace
\fi}

\newcommand{\PAPERVERSIONINLINE}[2]{
\ifnum1=\ExtendedVersion\relax
#2%
\else
#1%
\fi}

\newcommand{\ADD}[1]{%
\ifnum1=\HightlightChanges\relax
\cbcolor{green}
\begin{changebar}
#1%
\end{changebar}
\cbcolor{gray}
\else
#1%
\fi}

\newcommand{\MODIFY}[1]{%
\ifnum1=\HightlightChanges\relax
\begin{changebar}
#1
\end{changebar}
\else
#1
\fi}

\definecolor{Gray}{gray}{0.9}
\definecolor{vlightgray}{gray}{0.93}

\newcommand{\HIDEFORLONGVERSION}[1]{}

\lstdefinelanguage{Jules}{
  keywords={struct,is,end},
  keywordstyle=\color{darkgray}\bfseries,
  ndkeywords={struct,is,end},
  ndkeywordstyle=\color{darkgray}\bfseries,
  identifierstyle=\color{black},
  sensitive=false,  comment=[l]{//},  morecomment=[s]{/*}{*/},
  commentstyle=\color{gray}\ttfamily,  stringstyle=\color{gray}\ttfamily,
  morestring=[b]',  morestring=[b]",
  aboveskip=\medskipamount, 
  belowskip=\medskipamount, 
  escapeinside={(*@}{@*)}
}
\setcopyright{rightsretained}
\acmPrice{}
\acmDOI{10.1145/3485527}
\acmYear{2021}
\copyrightyear{2021}
\acmSubmissionID{oopsla21main-p237-p}
\acmJournal{PACMPL}
\acmVolume{5}
\acmNumber{OOPSLA}
\acmArticle{150}
\acmMonth{10}

\begin{document}


\renewcommand{\c}[1]{\lstinline[language=Julia]!#1!\xspace}

\title{Type Stability in Julia\PAPERVERSIONINLINE{}{\ (Extended Version)}}
\subtitle{Avoiding Performance Pathologies in JIT Compilation}

\author{Artem Pelenitsyn}\affiliation{\institution{Northeastern University}\country{USA} }
\email{pelenitsyn.a@northeastern.edu}
\author{Julia Belyakova}\affiliation{\institution{Northeastern University}\country{USA} }
\author{Benjamin Chung}\affiliation{\institution{Northeastern University}\country{USA} }
\author{Ross Tate}\affiliation{\institution{Cornell University}\country{USA} }
\author{Jan Vitek}\affiliation{\institution{Northeastern University}\country{USA}}
\affiliation{\institution{Czech Technical University in Prague}\country{Czech Republic}}

\renewcommand{\shorttitle}{Type Stability in Julia: Avoiding Performance
  Pathologies in JIT Compilation}

\begin{abstract}
  As a scientific programming language, Julia strives for performance but also
  provides high-level productivity features. To avoid performance pathologies,
  Julia users are expected to adhere to a coding discipline that enables
  so-called type stability. Informally, a function is type stable if the type of
  the output depends only on the types of the inputs, not their values.
  This paper provides a formal definition of type stability as well as a
  stronger property of type groundedness, shows that groundedness enables
  compiler optimizations, and proves the compiler correct. We also perform a
  corpus analysis to uncover how these type-related properties manifest in
  practice.
\end{abstract}
\maketitle

\PAPERVERSION{}{\emph{A note on Extended Version. This extended version
    of~\cite{Pelenitsyn21} contains full proofs of the theorems in
    \secref{sec:jules} and an appendix with graphs, as described in
    \secref{sec:empirical}. Most notable extension blocks (usually, the proofs)
    are marked with a vertical bar on the right margin, like
    this paragraph.}}

\section{Introduction}

Performance is serious business for a scientific programming language. Success
in that niche hinges on the availability of a rich ecosystem of high-performance
mathematical and algorithm libraries. Julia is a relative newcomer in this space.
Its approach to scientific computing is predicated on the bet that
users can write efficient numerical code in Julia without needing to resort to C
or Fortran.
The design of the language is a balancing act between productivity features, such
as multiple dispatch and garbage collection, and performance features, such as
limited inheritance and restricted-by-default dynamic code loading.
Julia has been designed to ensure that a
relatively straightforward path exists from source code to machine code,
and its performance results are encouraging. They show, on some simple benchmarks,
speeds in between those of C and Java. In other words, a new dynamic language
written by a small team of language engineers can compete with mature, statically
typed languages and their highly tuned compilers.

Writing efficient Julia code is best viewed as a dialogue between the programmer
and the compiler.
From its earliest days, Julia exposed the compiler's intermediate representation
to users, encouraging them to (1) observe if and how the compiler is able to optimize
their code, and (2) adapt their coding style to warrant optimizations. This came with a
simple execution model: each time a function is called with a different tuple of
concrete argument types, a new \emph{specialization} is generated by Julia's
just-in-time compiler. That specialization leverages the run-time type information
about the arguments, to apply \emph{unboxing} and \emph{devirtualization} transformations
to the code. The former lets the compiler manipulate values without indirection
and stack allocate them; the latter sidesteps the powerful but costly
multiple-dispatch mechanism and enables inlining of callees.

One key to performance in Julia
stems from the compiler's success in determining the \emph{concrete} return
type of any function call it encounters. The intuition is that in such cases,
the compiler is able to propagate precise type information as it traverses the
code, which, in turn, is crucial for unboxing and devirtualization.
More precisely, Julia makes an important distinction between 
concrete and abstract types: a concrete type such as \c{Int64}, is final
in the sense that the compiler knows the size and exact layout of
values of that type; for values of abstract types such as \c{Number}, the compiler has
no information. The property we alluded to is called \emph{type stability}.
Its informal definition states that a method is type stable if the concrete type of its
output is entirely determined by the concrete types of its arguments.\footnote{
\url{https://docs.julialang.org/en/v1/manual/faq/\#man-type-stability}} Folklore
suggests that one should strive to write type-stable methods outright, or, if
performance is an issue, refactor methods so that they become type stable.

Our goal in this paper is three-fold. First, we give a precise definition of
type stability that allows for formal reasoning.
Second, we formalize the relationship between type-stable code and the
ability of a compiler to perform type-based optimizations; to this end, we find
that a stronger property, which we call \emph{type groundedness}, describes
optimizable code more precisely.
Finally, we analyze prevalence of type stability in open-source Julia code
and identify patterns of type-stable code.
We make the following contributions:

\begin{itemize}
  \item An overview of type stability and its role in Julia, as well as the
    intuition behind stability and groundedness (\secref{sec:stability}).
  \item An abstract machine called \jules, which models Julia's intermediate
    representation, dynamic semantics, and just-in-time (JIT) compilation
    (\secref{sec:jules}).
  \item Formal definitions of groundedness and stability, and a proof that
    type-grounded methods are fully devirtualized by the JIT compiler
    (\secref{sec:stability-formal}). \ADD{Additionally, we prove that the JIT
    compilation is correct with respect to the dynamic-dispatch semantics
    (\secref{sec:jit-correct}). \PAPERVERSIONINLINE{Detailed proofs are available
    in the extended version of the paper~\cite{oopsla21jules:arx}.}{}}
  \item A corpus analysis of packages to measure stability and groundedness in
    the wild, find patterns of type-stable code, and identify properties
    that correlate with unstable code (\secref{sec:empirical}).
\end{itemize}
The paper is accompanied by the artifact~\cite{artifact} reproducing results
of~\secref{sec:empirical}.

\section{Julia in a Nutshell}

The Julia language is designed around multiple dispatch~\cite{BezansonEKS17}.
Programs consist of \emph{functions} that are implemented by multiple
\emph{methods} of the same name; each method is distinguished by a distinct type
signature, and all methods are stored in a so-called method table.
At run time, the Julia implementation dispatches a function call to
the \emph{most specific} method by comparing the types of the arguments to the
types of the parameters of all methods of that function. As an example of a
function and its constituent methods, consider \c{+}; as of version 1.5.4 of the
language, there are 190 implementations of \c{+}, each covering a
specific case determined by its type signature. Fig.~\ref{plus} displays custom
implementations for 16-bit floating point numbers, missing values,
big-floats/big-integers, and complex arithmetic.
\ADD{Although at the source-code level, multiple methods look similar to
overloading known from languages like C++ and Java, the key difference is that
those languages resolve overloading statically whereas Julia does that
dynamically using multiple dispatch.}

\begin{figure}
\begin{lstlisting}[language=julia]
# 184 methods for generic function "+":
[1] +(a::Float16, b::Float16) in Base at float.jl:398
[2] +(::Missing, ::Missing) in Base at missing.jl:114
[3] +(::Missing) in Base at missing.jl:100
[4] +(::Missing, ::Number) in Base at missing.jl:115
[5] +(a::BigFloat, b::BigFloat, c::BigFloat, d::BigFloat) in Base.MPFR at mpfr.jl:541
[6] +(a::BigFloat, b::BigFloat, c::BigFloat) in Base.MPFR at mpfr.jl:535
[7] +(x::BigFloat, c::BigInt) in Base.MPFR at mpfr.jl:394
[8] +(x::BigFloat, y::BigFloat) in Base.MPFR at mpfr.jl:363
...
\end{lstlisting}
\caption{Methods from the standard library}\label{plus}
\Description{First 8 out of 184 methods for the plus function in the Julia standard library.}
\end{figure}

Julia supports a rich type language for defining method
signatures. Base types consist of either bits types---types that have a direct
binary representation, like integers---or structs. Both bits types and struct
types, referred to as \emph{concrete types}, can have supertypes, but all
supertypes are \emph{abstract types}. Abstract types can be arranged into a
single-subtyping hierarchy rooted at \c{Any}, and no abstract type can be
instantiated. The type language allows for further composition of these base
types using unions, tuples, and bounded existential constructors; the result of
composition can be abstract or concrete. \citet{oopsla18b} gives a detailed
discussion of the type language and of subtyping.

Any function call in a program, such as \c{x+y}, requires choosing one of the
methods of the target function. \emph{Method dispatch} is a multi-step process.
First, the implementation obtains the concrete types of arguments. Second, it
retrieves applicable methods by checking for subtyping between argument types
and type annotations of the methods. Next, it sorts these methods into subtype
order. Finally, the call is dispatched to the most specific method---a method
such that no other applicable method is its strict subtype. If no such
method exists, an error is produced. As an example, consider the above
definition of \c{+}: a call with two \c{BigFloat}'s dispatches to
definition 8 from \figref{plus}.

\MODIFY{
Function calls are pervasive in Julia, and their efficiency is crucial for
performance. However, the many complex type-level operations involved in dispatch
make the process slow. Moreover, the language implementation, as of this writing,
does not perform inline caching~\cite{DS84}, meaning that dispatch results are
not cached across calls. To attain acceptable performance, the compiler attempts
to remove as many dispatch operations as it can. This optimization leverages
run-time type information whenever a method is compiled, i.e., when it is called
for the first time with a novel set of argument types.  These types are used by
the compiler to infer types in the method body. Then, this type information
frequently allows the compiler to devirtualize and inline the function calls
within a method~\cite{aigner}, thus improving performance. However, this
optimization is not always possible: if type inference cannot produce a
sufficiently specific type, then the call cannot be devirtualized. Consider the
prior example of \c{x+y}: if \c{y} is known to be one of \c{BigFloat} or
\c{BigInt}, the method to call cannot be determined. This problem arises for
various reasons, for example, accessing a struct field of an abstract type, or
the type inferencer losing precision due to a branching statement. A more
detailed description of the compilation strategy and its performance is given in
\cite{oopsla18a}.}

\section{Type Stability: a Key to Performance?}\label{sec:stability}

Designing performant Julia code is best understood as a conversation between the
language designers, compiler writers, and the developers of widely used
packages. The designers engineered Julia so that it is possible to achieve
performance with a moderately optimizing just-in-time (JIT) compiler
and with disciplined use of the abstractions included in the language~\cite{oopsla18a}.
For example, limiting the reach of \c{eval} by default~\cite{oopsla20a} allows
for a simpler compiler implementation,
and constraining struct types to be final enables unboxing.
The compiler's work is split between high-level optimizations, mainly
devirtualization and unboxing, and low-level optimizations that are off-loaded
to LLVM.\@ This split is necessary as LLVM on its own does not have enough
information to optimize multiple dispatch.
Removing dispatch is the key to performance, but to perform the optimization,
the compiler needs precise type information. Thus, while developers
are encouraged to write generic code, the code also needs to be conducive
to type inference and type-based optimizations. In this section,
we give an overview of the appropriate coding discipline, and explain
how it enables optimizations.

\paragraph{Performance.} To illustrate performance implications of careless
coding practices, consider Fig.~\ref{ide}, which displays a method for one of
the Julia microbenchmarks, \c{pisum}. For the purposes of this example, we
have added an identity function \c{id} which was initially implemented to
return its argument in both branches, as well-behaved identities do.
Then, the \c{id} method was changed to return a string in the impossible
branch (\c{rand()} returns a value from 0 to 1). The impact of that change
was about a 40\% increase in the execution time of the benchmark (\juliaversion).

\lstdefinestyle{jterm}{
    basicstyle=\footnotesize\ttfamily,
    moredelim=[is][\bfseries\color{Red}]{a@}{a@},
    moredelim=[is][\color{MidnightBlue}]{b@}{b@},
    moredelim=[is][\bfseries\color{OliveGreen}]{`}{`},
}

\begin{figure}[h!]
\centering
\begin{subfigure}[b]{0.42\textwidth}
\centering
\begin{lstlisting}[language=julia,style=jterm]
function id(x)
  (rand() == 4.2) ? "x" : x
end


function pisum()
  sum = 0.0
  for j = 1:500
    sum = 0.0
    for k = 1:10000
      sum += id(1.0/(k*id(k)))
    end
  end
  sum
end
\end{lstlisting}
\caption{Microbenchmark source code, redacted}
\end{subfigure}
\hspace{5mm}
\begin{subfigure}[b]{0.47\textwidth}
\begin{lstlisting}[style=jterm]
`julia>` @code_warntype id(5)
Variables
  #self#b@::Core.Compiler.Const(id, false)b@
  xb@::Int64b@
Bodya@::Union{Int64, String}a@
1 - %1 = Main.rand()b@::Float64b@
|   %2 = (%1 == 4.2)b@::Boolb@
+--      goto #3 if not %2
2 -      return "x"
3 -      return x

`julia>` @code_warntype pisum()
...
|   %20 = kb@::Int64b@
|   %21 = Main.id(k)a@::Union{Int64, String}a@
\end{lstlisting}
\caption{Julia session}
\end{subfigure}
\caption{\MODIFY{A Julia microbenchmark (a) illustrating performance implications
  of careless coding practices: changing \c{id} function to return
  values of different types leads to longer execution
  because of the \c{Union} type of \c{id(..)}, which propagates to \c{pisum} (b).}}
\label{ide}
\Description{Analyzing a redacted version of the pisum() Julia benchmark with
  the built-in code_warntype macro, we spot type instability. Indeed, the
  redaction amounted to insertion of an id function that does not change
  semantics but makes Julia believe that id and, in turn, pisum have
  instability. Our id function pretends that it can return something completely
  random of the string type in a very unlikely case when the rand function
  returns 4.2.}
\end{figure}

When a performance regression occurs, it is common for developers to study the
intermediate representation produced by the compiler. To facilitate this, the
language provides a macro, \c{code_warntype}, that shows the code along with the
inferred types for a given function invocation. Fig.~\ref{ide} demonstrates the result of
calling \c{code_warntype} on \c{id(5)}. Types that are imprecise, i.e., not
concrete, show up in red: they indicate that concrete type of a value may vary
from run to run. Here, we see that when called with an integer,
\c{id} may return either an
\c{Int64} or a \c{String}.
Moreover, the imprecise return type of \c{id} propagates to the caller,
as can be seen by inspecting \c{pisum} with \c{code_warntype}.
Such type imprecision can impact performance in two ways. First,
the \c{sum} variable has to be boxed, adding a level of indirection to
any operation performed therein. Second, it is harder for
the compiler to devirtualize and inline consecutive calls, thus requiring
dynamic dispatch.

\paragraph{Type Stability}
Julia's compilation model is designed to accommodate source programs
with flexible types, yet to make such programs efficient. The compiler, by
default, creates an \emph{instance} of each source method for each distinct tuple of
argument types. Thus, even if the programmer does not provide any type
annotations, like in the \c{id} example, the compiler will create method
instances for \emph{concrete} input types seen during
an execution. For example, since in \c{pisum}, function \c{id} is called both
with a \c{Float64} and an \c{Int64}, the method table will hold two method instances
in addition to the original, user-defined method.
Because method instances have more precise argument types, the compiler can
leverage them to produce more efficient code and infer more precise return types.

In Julia parlance, a method is called \emph{type stable} if its inferred
return type depends solely on the types of its arguments; in the
example, \c{id} is not type stable, as its return type may change depending on
the input value (in principle). The definition of type stability, though, is somewhat slippery.
The community has multiple, subtly different, informal definitions that capture
the same broad idea, but describe varying properties. The canonical definition
from the Julia documentation describes type stability as

\begin{quote}
  \it ``[...] the type of the output is predictable from the types of the
  inputs. In particular, it means that the type of the output cannot vary
  depending on the values of the inputs.''
\end{quote}
However, elsewhere, the documentation also states that \emph{``An analogous
type-stability problem exists for variables used repeatedly within a
function:''}
\begin{lstlisting}[basicstyle=\footnotesize\tt,language=julia]
function foo()
    x = 1
    for i = 1:10
        x /= rand()
    end
    x
end
\end{lstlisting}
This function will always return a \c{Float64}, which is the type of \c{x}
at the end of the \c{foo} definition, regardless of the (nonexistent)
inputs. However, the manual says that it is a type stability issue nonetheless.
This is because the variable \c{x} is initialized to an \c{Int64} but then
assigned a \c{Float64} in the loop. Some versions of the compiler boxed \c{x} as
it could hold two different types; of course, in this example, one level of loop
unrolling would alleviate the performance issue, but in general,
imprecise types limit compiler optimizations.
Conveniently, the \c{code_warntype} macro mentioned above
will highlight imprecise types for \emph{all} intermediate variables.
Furthermore, the documentation states that
\begin{quote}
  \it ``[t]his serves as a warning of potential type instability.''
\end{quote}

Effectively, there are two competing, type-related properties
of function bodies. To address this confusion, we
define and refer to them using two distinct terms:
\begin{itemize}
  \item \emph{type stability} is when a function's return type depends only on
    its argument types, and
  \item \emph{type groundedness} is when every variable's type depends
    only on the argument types.
\end{itemize}
\MODIFY{Although type stability is strictly weaker than type groundedness,
for the purposes of this paper, we are interested in both properties.
The latter, type groundedness, is useful for performance of the function itself,
as it implies that unboxing, devirtualization, and inlining can occur.
The former, type stability, allows the function to be used efficiently
by other functions: namely, type-grounded functions may call a function that
is only type stable but not grounded.
For brevity, when the context is clear, we will refer to type stability and
type groundedness as stability and groundedness in what follows.
}

\paragraph{Flavors of stability.}

Type stability is an inter-procedural property, and in the worst case, it can depend
on the whole program. Consider the functions of Fig.~\ref{fla}. Function \c{f0}
is trivially type unstable, regardless of the type of its input: if \c{good(x)}
returns true, \c{f0} returns an \c{Int64} value, otherwise \c{f0} returns a
\c{String}. Function \c{f1} is trivially type stable as all control-flow paths
return a constant of \c{Int64} type. Function \c{f2} is type stable as
long as the negation operator is type stable and returns a value of the same type
as its argument. As an example, we show a method \c{-(::Float64)} that causes
\c{f2(::Float64)} to lose type stability. This is a common bug where the function
\c{(-)} either returns a \c{Float64} or \c{Int64} due to the constant \c{0} being
of type \c{Int64}.  The proper, Julia-style implementation for this method is
to replace the constant \c{0} with \c{zero(x)}, which returns the zero value
for the type of \c{x}, in this case \texttt{\small0\!.\!0}.

\begin{figure}[!h]\footnotesize
\begin{minipage}{3cm}\begin{lstlisting}[basicstyle=\footnotesize\tt,language=julia]
function f0(x)
  if (good(x))
    0
  else
    "not 0"
  end
end
\end{lstlisting}\end{minipage}
\hspace{.5cm}
\begin{minipage}{3cm}\begin{lstlisting}[basicstyle=\footnotesize\tt,language=julia]
function f1(x)
  if (good(x))
    0
  else
    1
  end
end
\end{lstlisting}\end{minipage}
\hspace{.5cm}
\begin{minipage}{3cm}\begin{lstlisting}[basicstyle=\footnotesize\tt,language=julia]
function f2(x)
  if (good(x))
    x
  else
    -x
  end
end
\end{lstlisting}\end{minipage}
\hspace{.5cm}
\begin{minipage}{3.7cm}\footnotesize\begin{lstlisting}[basicstyle=\footnotesize\tt,language=julia]
function -(x::Float64)
  if (x==0)
    0
  else
    Base.neg_float(x)
  end
end
\end{lstlisting}\end{minipage}
\caption{\MODIFY{Flavors of stability: \c{f0} is unstable,
\c{f1} is stable; \c{f2} is stable if \c{(-)} is stable; \c{(-)} is unstable.}}\label{fla}
\Description{Several simple variants of a type-(un)stable function.}
\end{figure}

The example of function \c{f2} illustrates the fact that stability is a
whole-program property. Adding a method may cause some, seemingly unrelated,
method to lose type stability.

\paragraph{Stability versus groundedness}
Type stability of a method is important for the groundedness of its callers.
Consider the function \c{h(x::Int64)=g(x)+1}. If \c{g(x)=x}, it follows that
\c{h} is both stable and grounded, as \c{g} will always return an \c{Int64}, and
so will \c{h}. However, if we define \c{g(x) = (x==2) ? "Ho" : 4}, then \c{h}
suddenly loses both properties. To recover stability and groundedness of \c{h},
it is necessary to make \c{g} type stable, yet it does not have to be grounded.
For example, despite the presence of the imprecise variable \c{y}, the following
definition makes \c{h} grounded: \c{g(x) = let y = (x==2 ? "Ho" : 4) in x
end}.

In practice, type stability is sought after when making architectural decisions.
Idiomatic Julia functions are small and have no internal type-directed
branching; instead, branches on types are replaced with multiple dispatch.
Once type ambiguity is lifted into dispatch, small functions with specific
type arguments are relatively easy to make type stable. In turn, this
architecture allows for effective devirtualization in a caller, as in many cases, the
inferred type at a call site will determine its callee at compilation time.

Thus, writing type-stable functions is a good practice, for it provides
callers of those functions with an opportunity to be efficient.
However, stability of callees is not a sufficient condition for the efficiency
of their callers: the callers themselves need to strive for type groundedness,
which requires eliminating type imprecision from control flow.

\ADD{
\paragraph{Patterns of instability.} There are several code patterns
that are inherently type unstable. For one, accessing abstract fields of
a structure is an unstable operation:
the concrete type of the field value depends on the struct value,
not just struct type. In Julia, it is recommended to avoid abstractly typed
fields if performance is important, 
but they are a useful tool for interacting with external data sources
and representing heterogenous data.

Another example is sum types (algebraic data types or open unions), which can be
modeled with subtyping in Julia. Take a hierarchy of an abstract type \c{Expr}
and its two concrete subtypes, \c{Lit} and \c{BinOp}.
Such a hierarchy is convenient, because
it allows for an \c{Expr} evaluator written with multiple dispatch:
}
\begin{lstlisting}[basicstyle=\footnotesize\tt,language=julia]
run(e :: Lit)   = ...                           run(e :: BinOp) = ...
\end{lstlisting}
If we want the evaluator to be called with the result of a function
\c{parse(s :: String)}, the latter cannot be type stable: \c{parse} will return
values of different concrete types, \c{Lit} and \c{BinOp}, depending on the
input string.
If one does want \c{parse} to be stable, they need to always
return the same concrete type, e.g. an S-expression-style struct. Then, \c{run}
has to be written without multiple dispatch, as a big if-expression, which may
be undesirable.

\section{Jules: Formalizing Stability and Groundedness}\label{sec:jules}

\MODIFY{
To simplify reasoning about type stability and groundedness, we first
define \jules, an abstract machine that provides an idealized version of
Julia's intermediate representation (IR) and compilation model.
\jules captures the just-in-time (JIT) compilation process that (1) specializes methods
for concrete argument types as the program executes, and (2) replaces dynamically
dispatched calls with direct method invocations when type inference
is able to get precise information about the argument types.
It is the type inference algorithm that directly affects
type stability and groundedness of code, and thus the ability of the JIT compiler
to optimize it. While Julia's actual type inference algorithm
is quite complex, its implementation is not relevant for understanding
our properties of interest; thus, \jules abstracts over type inference
and uses it as a black box.
}

\begin{figure}[!h]
  \includegraphics[width=1.1\columnwidth]{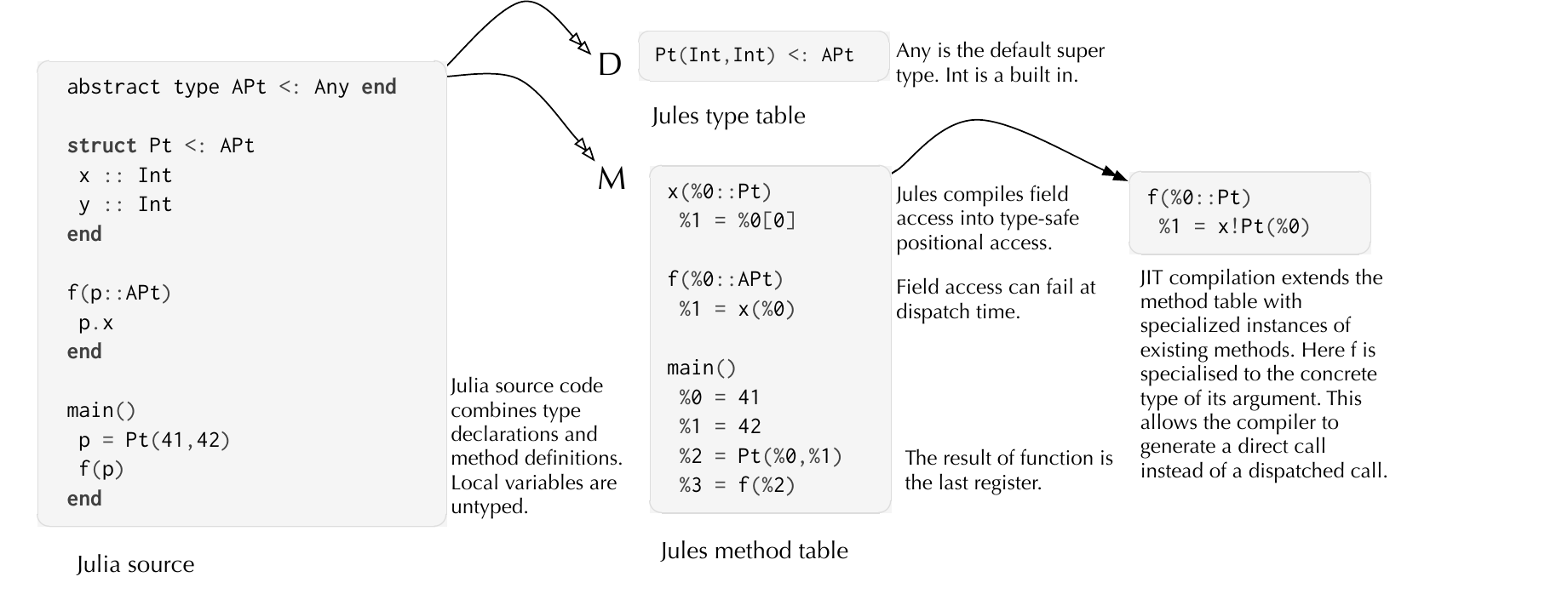}
  \caption{Compilation from Julia to \jules}\label{comp}
  \Description{Julia source code translated to Jules; one function is further compiled
    to a specialized version.}
\end{figure}

Fig.~\ref{comp} illustrates the relationship between Julia source code, the
\jules intermediate representation, and the result of compilation. We do not model
the translation from the source to \jules, and simply assume that the front-end
generates well-formed \jules code.\footnote{The front-end does not devirtualize
function calls, as Julia programmers do not have the ability to write direct
method invocations in the source.} A \jules program consists of an immutable \emph{type
table}~\tytbl and a \emph{method table}~\mtbl; the method table can be incrementally extended
with method instances that are compiled by the just-in-time compiler.

The source program of Fig.~\ref{comp} defines two types, the concrete \c{Pt} and
its parent, the abstract type \c{APt}, as well as two methods, \c{f} and
\c{main}. When translated to \jules, \c{Pt} is added to the type table along
with its supertype \c{APt}. Similarly, the methods \c{main} and \c{f} are added
to the \jules method table, along with accessors for the fields of \c{Pt}, with
bodies translated to the \jules intermediate representation.

The \jules IR is similar to static single assignment form. Each statement can
access values in registers, and saves its result into a new, consecutively
numbered, register. Statements can perform a dispatched call \c{f(\%2)}, direct
call \c{x\!Pt(\%0)}, conditional assignment (not shown), and a number of other
operations. The IR is untyped, but the translation from Julia is type sound. In
particular, type soundness guarantees that only dispatch errors can occur at run
time. For example, compilation will produce only well-formed field accesses such
as the one in \c{x(\%0::Pt)}, but a dispatched call \c{x(\%0)} in \c{f} could
fail if \c{f} was called with a struct that did not have an \c{x} field. In
order to perform this translation, \jules uses type inference to determine the
types of the program's registers. We abstract over this type inference mechanism
and only specify that it produces sound (with respect to our dynamic semantics)
results. 

Execution in \jules occurs between \emph{configurations} consisting of both
a stack of \emph{frames} \frm (representing the current execution state)
and a \emph{method table} \mtbl (consisting of original methods and specialized
instances), denoted $\config{\frm}{\mtbl}$. A configuration
$\config{\frm}{\mtbl}$ evolves to a new configuration
$\config{\frm'}{\mtbl'}$ by stepping $\config\frm\mtbl \step
\config{\frm'}{\mtbl'}$; every step processes the top instruction in $\frm$
and possibly compiles new method instances into $\mtbl'$. Notably, due to the
so-called world-age mechanism~\cite{oopsla20a} which restricts the effect
of \c{eval}, source methods are fixed from the
perspective of compilation; only compiled instances changes.

\subsection{Syntax}

The syntax of \jules methods is defined in~\figref{syntax}. We use two key
notational devices. First, sequences are denoted \ol{\,\cdot\,}; thus, \ol\ty
stands for types $\ty_0\dots\ty_n$, \ol{\v\k} for registers $\v\k_0\dots\v\k_n$,
and \ol\st for instructions $\st_0\dots\st_n$. An empty sequence is written
\emp. Second, indexing is denoted $[\cdot]$; \idx{\ol\ty}\k is the $k$-th type
in \ol\ty (starting from 0), \get\j\k is the $k$-th field of \v\j, and
\idx\mtbl{\msig\m{\ol\ty}} indexes \mtbl by method signature where \m denotes
method name and \ol\ty denotes argument types.

\begin{figure}[!h]\small
\begin{tabular}{lcl}
\begin{tabular}{llll}
\ty &::=& & type \\
    & | & \Ty  & \it concrete type \\
    & | & \aty & \it abstract type \\
\\
\tytbl &::=&  & type table \\
&&      $\left(\ \ \msig\Ty{\ol\ty} ~<:~ \aty \ \ \right)\!*$&\\
\mtbl&::=&  & {method table}\\
&&
  $\left(\ \ \langle \meth\m{\ol{\ty'}}{\ol{\v\j}}\ol\st,\ \ol{\ty}\rangle \ \ \right)\!*$\\
\end{tabular}
&\quad&
\begin{tabular}{ll@{~}ll}
\st &::=& & instruction\\
    & | & $\ass\i\p$ &\it int. assignment\\
    & | & $\ass\i{\v\j}$ &\it reg. transfer\\
    & | & $\ass\i{\construct\Ty{\ol{\v\k}}}$ & \it allocation\\
    & | & $\ass\i{\get\j\k}$ & \it field access\\
    & | & $\cond\i\j{\call\m{\ol{\v\k}}}{\v\l}$&\it dispatched call\\
    & | & $\cond\i\j{\direct\m{\ol{\Ty}}{\ol{\v\k}}}{\v\l}$&\it direct call\\
\multicolumn{3}{l}{$i \in \mathbb{N},\ \p \in \mathbb{Z}$} & \\
\end{tabular}
\end{tabular}\caption{Syntax of \jules}\label{syntax}
\Description{Syntax of Jules}
\end{figure}

Types \ty live in the immutable type table \tytbl, which contains both concrete
(\Ty) and abstract (\aty) types. Each type table entry is of the form
$\msig{\Ty}{\ol\ty} <: \aty$, introducing concrete type $\Ty$, with fields of
types $\ol\ty$, along with the single supertype $\aty$. Two predefined types are
the concrete integer type \int, and the universal abstract supertype \any.

Method tables \mtbl contain method definitions of two sorts:
first, original methods that come from source code;
second, method instances compiled from original methods.
To distinguish between the two sorts, we store the type signature of the
original method in every table entry.
Thus, table entry
$\langle\meth\m{\ol{\ty'}}{\ol{\v\j}}\ol\st,\ \ol{\ty}\rangle$ describes a
method \m with parameter types $\ol{\ty'}$ and the body comprised of
instructions $\ol\st$; type signature \ol{\ty} points to the original method.
If \ol\ty is equal to \ol{\ty'}, the entry defines an original method,
and \ol{\st} cannot contain direct calls.%
\footnote{All function calls in Julia source code are dispatched calls.}
Otherwise, \ol{\ty'}
denotes some concrete types~\ol\Ty, and the entry defines a
method instance compiled from \msig\m{\ol\ty}, specialized for concrete argument
types $\ol\Ty <: \ol\ty$.
Method table may contain multiple method definitions with the same name,
but they have to have distinct type signatures.

Method bodies consist of instructions \ol{\st}. An instruction $\v{\i}
\leftarrow \text{\c{op}}$ consists of an operation \c{op} whose result is
assigned to register $\v\i$. An instruction can assign from a primitive integer
\p, another register \v\j, a newly created struct $\construct\Ty{\ol{\v\k}}$
of type \Ty with field values $\ol{\v\k}$, or the result of looking up a struct
field as $\get\j\k$. Finally, the instruction may perform a function call. Calls
can be dispatched, $\call\m{\ol{\v\k}}$, where the target method is dynamically
looked up, or they can be direct, \direct\m{\ol\Ty}{\ol{\v\k}}, where the target
method is specified. All calls are conditional: \EM{\v\j ~?~ \text{\c{call}} :
  \v\l}, to allow for recursive functions.
If the register \v\j is non-zero, then \c{call} is performed. Otherwise,
the result is the value of register \v\l. Conditional calls can be trivially
transformed into unconditional calls; in examples, this transformation is
performed implicitly.

\subsection{Dynamic Semantics}

\jules is parameterized over three components: method dispatch $\mathcal D$,
type inference $\mathcal I$, and just-in-time compilation \jit. We do not
specify how the first two work, and merely provide their interface and a set of
criteria that they must meet (in sections \ref{sec:disp} and \ref{sec:infer},
respectively). The compiler, \jit, is instantiated with either the identity
function, which gives a regular non-optimizing semantics, or an optimizing
compiler, which is defined in section~\ref{sec:comp}. The optimizing compiler
relies on type inference $\mathcal I$ to devirtualize method calls. Type
inference also ensures well-formedness of method tables. Method dispatch
$\mathcal D$ is used in operational semantics.

\begin{figure}[!t]\small
\begin{tabular}{rclcrclcrcl}
\val &::=& \p\ |\ \construct\Ty{\ol\val}
& \qquad &
\env &::=& \ol\val
& \qquad &
\frm &::=& \done\ |\ \stk{\env\ \ol\st}{\frm} \\
\end{tabular}
\begin{mathpar}
\PAPERVERSIONINLINE{}{\framebox{\configd \step \config{\frm'}{\mtbl}}\\}

\inferrule[Prim]{
  \st = \ass\i\p\\\\
  \env' = \env + \p
}{
  \config{ \stk{ \env \ \st\,\ol\st}{\frm} }{\mtbl} \step
  \config{ \stk{ \env'\ \ol\st     }{\frm} }{\mtbl}
}

\inferrule[Reg]{
  \st = \ass\i{\v\j}\\\\
   \env'=\env + \idx\env\j
}{
  \config{ \stk{ \env \ \st\,\ol\st}{\frm} }{\mtbl} \step
  \config{ \stk{ \env'\ \ol\st     }{\frm} }{\mtbl}
}

\inferrule[New]{
  \st= \ass\i{\construct\Ty{\ol{\v\k}}}\\\\
  \env'  = \env + \construct\Ty{\idx\env{\ol{k}}}
}{
  \config{ \stk{ \env \ \st\,\ol\st }{\frm} }{\mtbl} \step
  \config{ \stk{ \env'\ \ol\st      }{\frm} }{\mtbl}
}
\\

\inferrule[Field]{
  \st= \ass\i{\get\j\k}\\\\
  \val = \idx\env\j \and
  \env' = \env + \idx\val\k
}{
  \config{ \stk{ \env \ \st\,\ol\st }{\frm} }{\mtbl} \step
  \config{ \stk{ \env'\ \ol\st      }{\frm} }{\mtbl}
}

\inferrule[False1]{
  \st= \cond\i\j{\call\m{\ol{\v k}}}{\v\l}\\\\
  0 = \idx\env\j \and
  \env'=\env +  \idx\env\l
}{
  \config{ \stk{ \env \ \st\,\ol\st }{\frm} }{\mtbl} \step
  \config{ \stk{ \env'\ \ol\st      }{\frm} }{\mtbl}
}

\inferrule[False2]{
  \st= \cond\i\j{\direct\m{\ol\Ty}{\ol{\v\k}}}{\v\l}\\\\
  0 = \idx\env\j \and
  \env'=\env + \idx\env\l
}{
  \config{ \stk{ \env \ \st\,\ol\st }{\frm} }{\mtbl} \step
  \config{ \stk{ \env'\ \ol\st      }{\frm} }{\mtbl}
}
\\

\inferrule[Disp]{
  \st= \cond\i\j{\call\m{\ol{\v\k}}}{\v\l}\\\\
  0 \not=\idx\env\j \and
  \ol\Ty = \typeof(\idx\env{\ol{k}}) \and
  \mtbl' = \jit(\mtbl, \m, \ol\Ty) \\\\\
  \ol{\st'} = \body(\dispatch{\mtbl'}\m{\ol\Ty}) \and
  \env' =  \idx\env{\ol{k}}
}{
  \config{                        \stk{\env\ \st\,\ol\st}{\frm}  }{\mtbl} \step
  \config{ \env'\ \stk{\ol{\st'}}{\stk{\env\ \st\,\ol\st}{\frm}} }{\mtbl'}
}

\inferrule[Direct]{
  \st= \cond\i\j{\direct\m{\ol\Ty}{\ol{\v\k}}}{\v\l}\\\\
  0 \not = \idx\env\j \\\\
  \ol{\st'} = \body(\mtbl[\msig\m{\ol\Ty}]) \and
  \env' = \idx\env{\ol{k}}
}{
  \config{                        \stk{\env\ \st\,\ol\st}{\frm}  }{\mtbl} \step
  \config{ \env'\ \stk{\ol{\st'}}{\stk{\env\ \st\,\ol\st}{\frm}} }{\mtbl}
}

\inferrule[Ret]{
  \env''=\env + \last{\env'}
}{
  \config{ \stk{\env'\ \done}{\stk{ \env  \ \st\,\ol\st }{\frm}}}{\mtbl} \step
  \config{                    \stk{ \env''\ \ol\st      }{\frm} }{\mtbl}
}
\end{mathpar}
\caption{Dynamic semantics of Jules}\label{sems}
\Description{Dynamic semantics of Jules}
\end{figure}

\subsubsection{Operational Semantics}\label{sec:opsem}

Fig.~\ref{sems} gives rules for the dynamic semantics. Given a type table \tytbl
as context, \jules can step a configuration $\config{\frm}{\mtbl}$ to
$\config{\frm'}{\mtbl'}$, written as $\config\frm\mtbl \step
\config{\frm'}{\mtbl'}$.
Stack frames \frm consist of a sequence of environment-instruction list pairs.
Thus, $\stk{\env\ \ol\st}\frm$ denotes a stack with environment \env and
instructions \ol\st on top, followed by a sequence of environment-instruction
pairs. Each environment is a list of values $\env=\ol\val$, representing
contents of the sequentially numbered registers. Environments can then be
extended as $\env + \val$, indexed as $\idx\env\k$, and their last value is
\last\env if \env is not empty.

The small-step dynamic semantics is largely straightforward. The first four
rules deal with register assignment: updating the environment with a constant
value (\SC{Prim}), the value in another register (\SC{Reg}),
a newly constructed struct (\SC{New}), or the value in a field (\SC{Field}).
The remaining five rules deal with function calls, either dispatched
\call\m{\ol{\v\k}} or direct \direct\m{\ol\Ty}{\ol{\v\k}}.
Call instructions are combined with conditioning: a call can only be made after
testing the register \v\j, called a guard register.
If the register value is zero, then the value of the alternate register
\v\l is returned (\SC{False1/False2}).
Otherwise, the call can proceed, by \SC{Disp} for dispatched calls and \SC{Direct}
for direct ones.
A dispatched call starts by prompting the JIT compiler to
specialize method \m from the method table~\mtbl with the argument types \ol\Ty
and produce a new method table $\mtbl'$.
Next, using the new table $\mtbl'$, the dispatch mechanism $\mathcal D$
determines the method to invoke. Finally, the body \ol{\st'} of the method and
call arguments \idx\env{\ol{k}} form a new stack frame for the callee,
and the program steps with the extended stack and the new table.
Direct calls are simpler because
a direct call $\msig\m{\ol\Ty}$ uniquely identifies the
method to invoke. Thus, the method's instructions are looked up in \mtbl
by the method signature, a new stack frame is created,
and the program steps with the new stack and the same method table.
The top frame without instructions to execute indicates the
end of a function call (\SC{Ret}): the last value of the top-frame environment
becomes the return value of the call, and the top frame is popped from the stack.

Program execution begins with the frame $\emp\ \main()$%
\footnote{Recall that unconditional calls are implicitly expanded into conditional ones.},
i.e. a call to the \main function with an empty environment;
the execution either diverges, finishes with a final configuration $\env\ \done$,
or runs into an error.
We define two notions of error. An \emph{err} occurs only in the \SC{Disp} rule,
when the dispatch function $\mathcal D$ is undefined for the call; the \emph{err}
corresponds to a dynamic-dispatch error in Julia. A configuration is
\emph{wrong} if it cannot make a step for any other reason.

\begin{definition}[Errors] A non-empty configuration \frm, \mtbl
 that cannot step $\frm, \mtbl \step \frm', \mtbl'$ has \emph{erred} if its
 top-most frame, \E\ \ol\st, starts with $\cond\i\j{\call\m{\ol{\v\k}}}{\v\l}$,
 where \ol\Ty is the types of \ol{\v\k} in \E, $\m\in\mtbl$, and
 \dispatch\mtbl\m{\ol\Ty} is undefined. Otherwise, \frm, \mtbl is
 \emph{wrong}.
\end{definition}

\subsubsection{Dispatch}\label{sec:disp}

\jules is parametric over method dispatch: any mechanism $\mathcal D$ that
satisfies the \emph{Dispatch Contract} (\defref{disprel}) can be used. Julia's
method dispatch mechanism
is designed to, given method table, method name, and argument types,
return the \emph{most specific method applicable} to the given arguments
if such a method exists and is unique. First, applicable
methods are those whose declared type signature is a supertype of the argument
type. Then, the most specific method is the one whose type signature is the
most precise. Finally, only one most specific applicable method may exist, or
else an error is produced. Each of these components appears in our dispatch
definition. As in Julia, dispatch is only defined for tuples of concrete types.

\begin{definition}[Dispatch Contract] The \emph{dispatch} function
  $\dispatch\mtbl\m{\ol\Ty}$ takes method table \mtbl, method name \m,
  and concrete argument types \ol\Ty, and returns a method
  $\meth\m{\ol{\ty}}{\ol{\v\j}}\ol\st \in \mtbl$ such that
  the following holds (we write $\ol\ty <: \ol{\ty'}$ as a shorthand for
  $\forall i.\,\ty_i<:\ty'_i$):
  \begin{enumerate}
    \item $\ol\Ty <: \ol\ty$, meaning that $\msig\m{\ol\ty}$
      is applicable to the given arguments;
    \item $\forall \meth\m{\ol{\ty'}}{\ol{\v\j}}\ol\st \in \mtbl.\ \
      \ol\Ty <: \ol{\ty'}\ \implies\ \ol{\ty} <: \ol{\ty'},$
      meaning that $\msig\m{\ol\ty}$ is the most specific applicable method.
  \end{enumerate}
  \label{disprel}
\end{definition}

\subsubsection{Inference}\label{sec:infer}

The Julia compiler infers types of variables by forward data-flow analysis. Like
dispatch, inference is complex, so we parameterize over it. For our purposes, an
inference algorithm~$\mathcal I$ returns a sound typing for a sequence of
instructions in a given method table,
$\typeinf{\origmtbl{\mtbl}}{\ol\ty}{\ol\st}=\ol{\ty'}$, where \origmtbl{\mtbl}
denotes the table containing only methods without direct calls. Inference
returns types $\ol\ty'$ such that each $\ty'_i$ is the type of register of
$\st_i$. Any inference algorithm that satisfies the \SC{Soundness} and
\SC{Monotonicity} requirements is acceptable.

\begin{requirement}[Soundness] If
  $\typeinf\mtbl{\ol\ty}{\ol\st} = \ol{\ty'}$, then
  for any environment \env compatible with \ol\ty, that is, $\typeof(E_{i}) <: \ty_{i}$,
  and for any stack $\frm$, the following holds:
  \begin{enumerate}
  \item If
  $\config{\stk{\env\ \ol\st}\frm}{\mtbl} \stepmul \stk{\frm'}{\frm},\mtbl'$
  and 
  cannot make another step,
  then $\stk{\frm'}{\frm},\mtbl'$ has \emph{erred}.
  \item If $\config{\stk{\env\ \ol\st}{\frm}}{\mtbl} \stepmul
    \config{\stk{\env\,\env'\ \ol{\st''}}{\frm}}{\mtbl'}$,
  then $\ol\st = \ol{\st'}\,\ol{\st''}$ and $\typeof(\env'_{i}) <: \ty'_{i}$.
  \end{enumerate}
 \end{requirement}

The soundness requirement guarantees that if type inference succeeds on a method,
then, when the method is called with compatible arguments,
it will not enter a \emph{wrong} state but may \emph{err} at a dynamic call.
Furthermore, if the method terminates, all its instructions evaluate
to values compatible with the results of type inference.
That is, when $\config{\stk{\env\ \ol\st}{\frm}}{\mtbl} \stepmul
\config{\stk{\env\,\env'\ \done}{\frm}}{\mtbl'}$, we have
$\typeof(\env') <: \ol{\ty'}$.

The second requirement of type inference, monotonicity, is important to specialization:
it guarantees that using more precise argument types for original method
bodies succeeds and does not break assumptions of the caller about
the callee's return type. If inference was not
monotonic, then given more precise argument types, it could return
a method specialization with a less precise return type. As a result,
translating a dynamically dispatched call into a direct call may be unsound.

\begin{requirement}[Monotonicity]\label{prop:ti-monot}
  For all $\mtbl, \ol{\ty},
  \ol\st, \ol{\ty'},$ such that $\typeinf\mtbl{\ol{\ty}}{\ol\st} = \ol{\ty'}$,
  \[
    \forall\, \ol{\ty''}.\quad \ol{\ty''} <: \ol{\ty}
    \quad \implies \quad
    \exists \ol{\ty'''}.\ \ 
    \typeinf\mtbl{\ol{\ty''}}{\ol\st} = \ol{\ty'''} \ \ \land\ \
    \ol{\ty'''} <: \ol{\ty'}.
  \]
\end{requirement}

\subsubsection{Well-Formedness}\label{wellform}

Initial \jules configuration \config{\emp\ \main()}{\mtbl} is well-formed
if the method table \mtbl is well-formed according to \defref{ref:def-well-formed}.
Such a configuration will either successfully terminate, \emph{err},
or run forever, but it will never reach a \emph{wrong} state.

\begin{definition}[Well-Formedness of Method Table]\label{ref:def-well-formed}
A method table is \emph{well-formed} $\mathit{WF}(\mtbl)$ if the following
holds:
\begin{enumerate}
\item Entry method $\msig{\SF{main}}\emp$ belongs to $\mtbl$.
\item Every type \ty in \mtbl, except \int and \any, is declared in \tytbl.
\item Registers are numbered consecutively from 0, increasing by one for each
  parameter and instruction. An instruction assigning to \v\k only refers to
  registers \v\i such that $i<k$.
\item For any original method $\langle
  \meth\m{\ol{\ty}}{\ol{\v\j}}\ol\st,\ \ol\ty \rangle \in \mtbl$, the body is
  not empty and does not contain direct calls, and type
  inference succeeds
  $\typeinf{\origmtbl{\mtbl}}{\ol\ty}{\ol\st}=\ol{\ty'}$
  on original methods \origmtbl\mtbl.
\item Any two methods in \mtbl with the same name,
  $\msig{\m}{\ol\ty}$ and $\msig{\m}{\ol{\ty'}}$, have distinct type
  signatures, i.e. $\ol\ty \neq \ol{\ty'}$.
  \item For any method specialization $\langle \msig\m{\ol{\ty'}}, \ol{\ty}\rangle\in\mtbl$,
  i.\,e. $\ol{\ty'} \neq \ol{\ty}$, the following holds:
  $\ol{\ty'} = \ol\Ty$, and
  $\ol{\Ty} <: \ol{\ty}$, and
  $\langle \msig\m{\ol{\ty}}, \ol{\ty}\rangle\in\mtbl$.
  Moreover,
   $\forall \msig\m{\ol\ty''} \in \mtbl.\
      \ol\Ty <: \ol{\ty''} \implies \ol{\ty} <: \ol{\ty''}$.
\end{enumerate}
\end{definition}

The last requirement ensures that only the most specific original methods have
specializations, which precludes compilation from modifying program behavior.
For example, consider type hierarchy
\c{Int<:Number<:Any} and function \c{f} with original methods for \c{Any} and
\c{Number}. If there are no compiled method instances, the call \c{f(5)}
dispatches to \c{f(::Number)}. But if the method table contained a specialized
instance \c{f(::Int)} of the original method \c{f(::Any)}, the call would dispatch
to that instance, which is not related to the originally used \c{f(::Number)}.
Thus, program behavior would be modified by compilation, which is undesired.

\begin{figure}[!t] \small\begin{mathpar}
\Rule{CompileTrue}{
  \msig{\m}{\ol\Ty} \not\in \mtbl \and
  \ol\st = \body(\dispatch\mtbl\m{\ol\Ty}) \and
  \ol{\ty'} = \typeinf{\origmtbl{\mtbl}}{\ol\Ty}{\ol\st}
  \\
  \ol{\ty} = \mathit{signature}(\dispatch\mtbl\m{\ol\Ty}) \and
  \mtbl_0 = \mtbl + \langle \msig\m{\ol{\Ty}}\done, \ol\ty \rangle
  \\
  \ol\Ty\,\ol{\ty'} \vdash \st_0, \mtbl_0 \leadsto \st'_0,\mtbl_1\and
  \dots\and
  \ol\Ty\,\ol{\ty'} \vdash \st_n, \mtbl_n \leadsto  \st'_n,\mtbl_{n+1} \\
  \mtbl' = \mtbl_{n+1} [\msig\m{\ol{\Ty}} \mapsto \langle \meth\m{\ol\Ty}{\ol{\v k}} \ol{\st'},\ol{\ty}\rangle]
}{
   \jit(\mtbl,\m,\ol\Ty) = \mtbl'
}

\Rule{CompileFalse}{
  \msig{\m}{\ol\Ty} \in \mtbl
}{
  \jit(\mtbl,\m,\ol\Ty) = \mtbl
}

\Rule{Dispatch}{
  \st = \cond\i\j{\call\m{\ol{\v\k}}}{\v\l} \and
  \ol{\Ty} = \ol\ty[\ol\k] \\\\
  \mtbl' = \jit(\mtbl,\m,\ol{\Ty}) \and
  \st' = \cond\i\j{\direct\m{\ol{\Ty}}{\ol{\v\k}}}{\v\l}
}{
  \ol\ty\ \vdash\ \st,\mtbl \leadsto \st',\mtbl'
}

\Rule{NoDispatch}{
  \st \neq \cond\i\j{\call\m{\ol{\v\k}}}{\v\l}\\\\
  \text{or} \and
  \ol{\Ty} \neq \ol\ty[\ol\k]
}{
  \ol\ty\ \vdash\ \st,\mtbl \leadsto \st,\mtbl
}
\end{mathpar}

  \caption{\MODIFY{Compilation: extending method table with a specialized method
  instance (top rule) and replacing a~dynamically dispatched call with
  a direct method invocation in the extended table (bottom-right)}}\label{compile}
\Description{Compilation: extending method table with a specialized method
  instance and replacing a~dynamically dispatched call with
  a direct method invocation in the extended table}
\end{figure}

\subsubsection{Compilation}\label{sec:comp}

\jules implements devirtualization through the $\jit(\mtbl,\m,\ol\Ty)$
operation, shown in \figref{compile}. The compiler specializes methods according to the
inferred type information, replacing non-\emph{err} dispatched calls with direct calls where
possible. Compilation begins with some bookkeeping. First, it ensures that there
is no pre-existing instance in the method table before compiling;
otherwise, the table is returned immediately without modification, by
the bottom-left rule. Next, using dispatch,
it fetches the most applicable original method to compile. Then, using concrete argument types,
the compiler runs the type inferencer on the method's body, producing an
instruction typing \ol{\ty'}. Because the original method table is well-formed,
monotonicity of $\mathcal I$ and the definition of $\mathcal D$ guarantee that
type inference succeeds for $\ol\Ty <: \ol\ty$. Lastly, the compiler can begin
translating instructions. Each instruction $\st_i$ is translated into an
optimized instruction $\st'_i$. This translation respects the existing type
environment containing the argument types \ol\Ty and instruction typing
$\ol{\ty'}$. The translation $ \ol\ty\ \vdash\ \st,\mtbl \leadsto \st',\mtbl'$
leaves all instructions unchanged except dispatched calls. Dispatched calls
cause a recursive JIT invocation, followed by rewriting to a direct call. To
avoid recursive compilation, a stub method $\langle \msig\m{\ol{\Ty}}\done,
\ol\ty \rangle$ is added at the beginning of compilation,
and over-written when compilation is done. As the source program is finite and
new types are not added during compilation, it terminates. Note that if
the original method has concrete argument types, the compiler does nothing.

\subsection{Type Groundedness and Stability}\label{sec:stability-formal}

We now formally define the properties of interest, type stability and type
groundedness. Recall the informal definition which stated that a function is
type stable if its return type depends only on its argument types, and type grounded
if every variable's type depends only on the argument types. In this
definition, ``types'' really mean concrete types, as concrete types allow for
optimizations. The following defines what it means for an original method to be
type stable and type grounded.

\begin{definition}[Stable and Grounded]\label{def:ground}
  Let \meth{\m}{\ol\ty}{}{\ol\st} be an original method in a well-formed
  method table \mtbl. Given concrete argument types $\ol\Ty <: \ol\ty$,
  if $\typeinf{\origmtbl{\mtbl}}{\m}{\ol\Ty} = \ol{\ty'}$,
  and
  \begin{enumerate}
  \item  $\last{\ol{\ty'}} = \Ty'$, i.e. the return type is concrete, then
    the method is \emph{type stable for \ol\Ty},
  \item $\ol{\ty'} = \ol{\Ty'}$, i.e. all register types are concrete, then the
    method is \emph{type-grounded for \ol\Ty}.
  \end{enumerate}

  Furthermore, a method is called \emph{type stable (grounded) for a set $W$
  of concrete argument types} if it is type stable (grounded) for every
  $\ol\Ty \in W$.
\end{definition}

As all method instances are compiled from some original method definitions,
type stabiltiy and groundedness for instances is defined in terms of their
originals.

\begin{definition}\label{def:ground-inst}
  A method instance $\langle \meth{\m}{\ol\Ty}{}{\ol{\st'}},\ \ol\ty \rangle$
  is called \emph{type stable (grounded)},
  if its original method \msig\m{\ol\ty} is type stable (grounded) for \ol\Ty.
\end{definition}

\subsubsection{Full Devirtualization}\label{sec:ground-inst-prop}

The key property of type groundedness is that
compiling a type-grounded method results in a fully devirtualized instance.
We say that a method instance \meth\m{\ol\Ty}{\ol{\v k}}\ol{\st} is \emph{fully devirtualized}
if \ol{\st} does not contain any dispatched calls. To show that \jit
indeed has the above property, we will need an additional
notion, \emph{maximal devirtualization}, which is defined
in~\figref{fig:max-devirt}. Intuitively, the predicate \devirtst{\ol\ty}{\mtbl}{\st}
states that an instruction \st does not contain a dispatched call
that can be resolved in table \mtbl for argument types found in \ol\ty.
Then, a method instance is maximally devirtualized if this predicate
holds for every instruction using \ol\ty that combines argument types with
the results of type inference.

\begin{figure}[!t]\small
\begin{mathpar}
\PAPERVERSIONINLINE{\\}{\framebox{\devirtst{\ol\ty}{\mtbl}{\st}}\\}%

\inferrule[D-NoCall]{
  \st \neq \cond\i\j{\call\m{\ol{\v\k}}}{\v\l} \and
  \st \neq \cond\i\j{\call{\msig\m{\ol{\Ty}}}{\ol{\v\k}}}{\v\l}
}{
  \devirtst{\ol\ty}{\mtbl}{\st}
}
\\

\inferrule[D-Disp]{
  \st = \cond\i\j{\call\m{\ol{\v\k}}}{\v\l} \and
  \ol\Ty \neq \idx{\ol\ty}{\ol\k}
}{
  \devirtst{\ol\ty}{\mtbl}{\st}
}

\inferrule[D-Direct]{
  \st = \cond\i\j{\call{\msig\m{\ol\Ty}}{\ol{\v\k}}}{\v\l} \and
  \ol\Ty = \idx{\ol\ty}{\ol\k} \and
  \msig{\m}{\ol{\Ty}} \in \mtbl
}{
  \devirtst{\ol\ty}{\mtbl}{\st}
}
\\

\PAPERVERSIONINLINE{}{\framebox{\devirtst{\ol\ty}{\mtbl}{\ol\st}}\\}
\inferrule[D-Seq]{
  \devirtst{\ol\ty}{\mtbl}{\st_0} \and \ldots \and
  \devirtst{\ol\ty}{\mtbl}{\st_n}
}{
  \devirtst{\ol\ty}{\mtbl}{\ol\st}
}
\\

\PAPERVERSIONINLINE{}{\framebox{\devirtm{\mtbl}}\\}%
\inferrule[D-Table]{
  \forall \langle \meth\m{\ol\Ty}{}\ol{\st'},\ol{\ty} \rangle \in \mtbl. \ \ \
  \ol\Ty \neq \ol\ty \ \land\ \ol{\st'} \neq \done \ \ \ \implies \ \ \ 
  \meth\m{\ol\ty}{}\ol{\st} \in \mtbl \ \ \land \ \
  \ol{\ty'} = \typeinf{\origmtbl{\mtbl}}{\ol\Ty}{\ol\st} \ \ \land \ \
  \devirtst{\ol\Ty\,\ol{\ty'}}{\mtbl}{\ol{\st'}}
}{
  \devirtm{\mtbl}
}
\end{mathpar}
\caption{\MODIFY{Maximal devirtualization of instructions
  and method tables}}\label{fig:max-devirt}
\Description{Maximal devirtualization of instructions
  and method tables}
\end{figure}

Next, we review the definition from~\figref{fig:max-devirt} in more details.
\SC{D-NoCall} states that an instruction without a call
is maximally devirtualized. \SC{D-Disp} requires that for a dispatched call,
\ol\ty does not have precise enough type information to resolve the call with
$\mathcal{D}$. Finally, \SC{D-Direct} allows a direct call to a concrete method
with the right type signature: as concrete types do not have subtypes, the
\msig\m{\ol\Ty} with concrete argument types is exactly the definition that a call
\call\m{\ol{\v\k}} would dispatch to. \SC{D-Seq} simply checks that all
instructions in a sequence \ol\st are devirtualized in the same register typing
\ol\ty. Here, \ol\ty has to contain typing for all instructions $\st_i$, because
later instructions refer to the registers of the previous ones. The last rule
\devirtm{\mtbl} covers the entire table \mtbl, requiring all methods to be maximally
devirtualized. Namely, \SC{D-Table} says that for all method instances
(condition $\ol\Ty \neq \ol\ty$ implies that \msig\m{\ol\Ty}
is not an original method) that are not
stubs ($\ol{\st'} \neq \done$), the body \ol{\st'} is maximally devirtualized
according to the typing of the original method with respect to the instance's
argument types.

Using the notion of maximal devirtualization, we can connect type groundedness
and full devirtualization.

\begin{lemma}[Full Devirtualization]\label{lem:full-virt}
  If \mtbl is well-formed and maximally devirtualized, then any type-grounded
  method instance $\meth\m{\ol\Ty}{\ol{\v k}}\ol{\st'} \in \mtbl$ is fully
  devirtualized.
\end{lemma}
\begin{proof}
\PAPERVERSION{
  Follows from the definitions of type groundedness and maximal
  devirtualization; details can be found in the extended
  version~\cite{oopsla21jules:arx}.
}{
  Recall that a method instance is type-grounded if type inference
  produces concrete typing \ol{\Ty'} for the original method body \ol\st, i.e.
  $\typeinf{\origmtbl{\mtbl}}{\ol\Ty}{\ol\st} = \ol{\Ty'}$.
  By the definition of a maximally devirtualized table, we know that
  \devirtst{\ol\Ty\,\ol{\Ty'}}{\mtbl}{\ol{\st'}}.
  Since all types in the register typing $\ol\Ty\,\ol{\Ty'}$ are concrete,
  by analyzing maximal devirtualization for instructions, we can see
  that the only applicable rules are \SC{D-NoCall} and \SC{D-Direct}.
  Therefore, \ol{\st'} does not contain any dispatched calls.
}
\end{proof}

\begin{figure}\small
\begin{mathpar}
\PAPERVERSIONINLINE{}{\framebox{\compst{\ol\ty}{\st}{\mtbl}{\S}{\st'}{\mtbl'}{\S'}}\\}

\inferrule[C-NoDisp]{
  \st \neq \cond\i\j{\call\m{\ol{\v\k}}}{\v\l}
}{
  \compst{\ol\ty}{\st}{\mtbl}{\S}{\st}{\mtbl}{\S}
}

\inferrule[C-Disp]{
  \st = \cond\i\j{\call\m{\ol{\v\k}}}{\v\l} \and
  \ol\Ty \neq \idx{\ol\ty}{\ol\k}
}{
  \compst{\ol\ty}{\st}{\mtbl}{\S}{\st}{\mtbl}{\S}
}

\inferrule[C-Direct]{
  \st = \cond\i\j{\call\m{\ol{\v\k}}}{\v\l} \and
  \ol\Ty = \idx{\ol\ty}{\ol\k} \and
  \msig{\m}{\ol\Ty} \in \mtbl \and
  \st' = \cond\i\j{\call{\msig\m{\ol\Ty}}{\ol{\v\k}}}{\v\l}
}{
  \compst{\ol\ty}{\st}{\mtbl}{\S}{\st'}{\mtbl}{\S}
}

\inferrule[C-Instance]{
  \st = \cond\i\j{\call\m{\ol{\v\k}}}{\v\l} \and
  \ol\Ty = \idx{\ol\ty}{\ol\k} \and
  \meth{\m}{\ol{\ty'}}{}{\ol\st} = \dispatch\mtbl\m{\ol\Ty} \and
  \ol\Ty \neq \ol{\ty'} \\\\
  \ol{\ty''} = \typeinf{\origmtbl{\mtbl}}{\ol\Ty}{\ol\st} \and
  \mtbl_0 = \mtbl + \langle \msig\m{\ol{\Ty}}\done, \ol{\ty'} \rangle \\\\
  \compst{\ol\Ty\,\ol{\ty''}}{\st_0}{\mtbl_0}{\S_0}{\st'_0}{\mtbl_1}{\S_1} \and
  \ldots \and
  \compst{\ol\Ty\,\ol{\ty''}}{\st_n}{\mtbl_n}{\S_n}{\st'_n}{\mtbl_{n+1}}{\S_{n+1}} \\\\
  \st' = \cond\i\j{\call{\msig\m{\ol\Ty}}{\ol{\v\k}}}{\v\l} \and
  \mtbl' = \mtbl_{n+1} + \langle \meth\m{\ol\Ty}{\ol{\v k}} \ol{\st'},\ol{\ty'} \rangle
}{
  \compst{\ol\ty}{\st}{\mtbl}{\S}{\st'}{\mtbl'}{\S'}
}
\end{mathpar}
\caption{Reformulated compilation}\label{fig:compile-wd}
\Description{Reformulated compilation}
\end{figure}

The final step is to show that compilation defined in \figref{compile} preserves
maximal devirtualization. To simplify the proof, \figref{fig:compile-wd}
reformulates \figref{compile} by inlining \jit into the compilation
relation. The relation does not have a rule for processing direct calls because
we compile only original methods. Since the set of method instances is finite,
the relation is well-defined: every recursive call to compilation happens
for a method table that contains at least one more instance.
Every compilation step produces a maximally devirtualized instruction and
potentially extends the method table with
new maximally devirtualized instances.

\MODIFY{

\begin{lemma}[Preserving Well Formedness]\label{lem:wf}
  For any method table \mtbl that is well-formed $\mathit{WF}(\mtbl)$,
  register typing \ol\ty, and instruction \st,
  if the instruction \st is compiled,
  $\compst{\ol\ty}{\st}{\mtbl}{\S}{\st'}{\mtbl'}{\S'},$
  then the new table is well-formed $\mathit{WF}(\mtbl')$.
\end{lemma}
\begin{proof}
  By induction on the derivation of
  \compst{\ol\ty}{\st}{\mtbl}{\S}{\st'}{\mtbl'}{\S'}.
  The only case that modifies the method table is \SC{C-Instance}.
\PAPERVERSION{
  There, by analyzing $\mtbl_0$, we can see that it is well-formed,
  and thus the induction hypothesis can be applied to
  \compst{\ol\Ty\,\ol{\ty''}}{\st_0}{\mtbl_0}{\S_0}{\st'_0}{\mtbl_1}{\S_1},
  and then to all
  \compst{\ol\Ty\,\ol{\ty''}}{\st_i}{\mtbl_i}{\S_0}{\st'_i}{\mtbl_{i+1}}{\S_{i+1}}.
  Since $\mtbl_{n+1}$ is well-formed, so is $\mtbl'$.
  More details are available in the extended version~\cite{oopsla21jules:arx}.
}{
  Let us analyze $\mtbl_0$ first.
  Since $\ol\Ty \neq \ol{\ty'}$ and \mtbl is well-formed,
  we know that \msig\m{\ol{\ty'}} is an original method and $\ol\Ty \neq \ol{\ty'}$.
  Furthermore, as dispatch is known to return the most applicable method,
  $\msig\m{\ol\Ty} \notin \mtbl$ and properties (5) and (6) of
  \defref{ref:def-well-formed} for $\mtbl_0$ are satisfied.
  Other properties are trivially satisfied because $\mtbl_0$
  does not add or modify any original methods, so:
  \[\mathit{WF}(\mtbl_0).\]
  Therefore, we can apply the induction hypothesis to
  \compst{\ol\Ty\,\ol{\ty''}}{\st_0}{\mtbl_0}{\S_0}{\st'_0}{\mtbl_1}{\S_1}
  to get $\mathit{WF}(\mtbl_1).$
  Proceeding in this manner, we arrive to:
  \[\mathit{WF}(\mtbl_{n+1}).\]
  As $\mtbl'$ only changes the body of the compiled instance \msig\m{\ol\Ty}
  compared to the well-formed $\mtbl_{n+1}$, we get the desired conclusion:
  \[\mathit{WF}(\mtbl').\]
}
\end{proof}

\begin{lemma}[Preserving Maximal Devirtualization]\label{lem:max-devirt}
  For any well-formed method table \mtbl, register typing \ol\ty,
  and instruction \st, if the method table is maximally devirtualized,
  \devirtm{\mtbl},  and the instruction \st is compiled,
  $\compst{\ol\ty}{\st}{\mtbl}{\S}{\st'}{\mtbl'}{\S'},$
  then the following holds:
  \begin{enumerate}
    \item the resulting instruction is maximally devirtualized in the new table,
      \devirtst{\ol\ty}{\mtbl'}{\ol{\st'}},
    \item the new table is maximally devirtualized, \devirtm{\mtbl'},
    \item any maximally devirtualized instruction stays maximally devirtualized
      in the new table,
      $\forall\, \ol{\ty^x}, \ol{\st^x}.\quad
      \devirtst{\ol{\ty^x}}{\mtbl }{\ol{\st^x}}\ \implies\
      \devirtst{\ol{\ty^x}}{\mtbl'}{\ol{\st^x}}$.
  \end{enumerate}
\end{lemma}

\begin{proof}
\PAPERVERSION{
  By induction on the derivation of
  \compst{\ol\ty}{\st}{\mtbl}{\S}{\st'}{\mtbl'}{\S'}.
  The most interesting case is \SC{C-Instance} where
  compilation of additional instances happens.
  The key step of the proof is to observe that the induction hypothesis
  is applicable to
  \compst{\ol\Ty\,\ol{\ty''}}{\st_0}{\mtbl_0}{\S_0}{\st'_0}{\mtbl_1}{\S_1},
  and then to all consecutive
  \compst{\ol\Ty\,\ol{\ty''}}{\st_i}{\mtbl_i}{\S_0}{\st'_i}{\mtbl_{i+1}}{\S_{i+1}},
  which is possible due to \lemref{lem:wf} and facts (2).
  Using facts (3), we propagate the information about maximal devirtualization
  of $\st_i$ in their respective tables to the final $\mtbl'$.
  A complete proof is available in~\cite{oopsla21jules:arx}.
}
{
  By induction on the derivation of
  \compst{\ol\ty}{\st}{\mtbl}{\S}{\st'}{\mtbl'}{\S'}.
  \begin{itemize}
    \item Cases \SC{C-NoDisp} and \SC{C-Disp} are straightforward:
      to show (1), we use rules \SC{D-NoCall} and \SC{D-Disp}, respectively.
      Since the resulting method table \mtbl is the same, (2) and (3)
      hold trivially.

    \item Case \SC{C-Direct}. By assumption, we know that
      $\ol\Ty = \idx{\ol\ty}{\ol\k}$ and that $\msig{\m}{\ol\Ty} \in \mtbl$.
      Therefore, we have (1) by \SC{D-Direct}:
      \[\devirtst{\ol\ty}{\mtbl}{\cond\i\j{\call{\msig\m{\ol\Ty}}{\ol{\v\k}}}{\v\l}}.\]
      The resulting method table \mtbl is the same, so (2) and (3) hold.

    \item Case \SC{C-Instance}. This is the interesting case where
      compilation of additional instances happens.
      First, note that $\mtbl_0$ is the same as \mtbl except for a new instance
      stub $\langle \msig\m{\ol{\Ty}}\done, \ol{\ty'} \rangle$.
      By assumption, \devirtm{\mtbl} holds, and \SC{D-Table} does not impose
      any constraints on stubs, so \devirtm{\mtbl_0} also holds.
      As shown in the proof of \lemref{lem:wf}, $\mtbl_0$ is also well-formed.
      Therefore, we can apply the induction hypothesis to
      \[\compst{\ol\Ty\,\ol{\ty''}}{\st_0}{\mtbl_0}{\S_0}{\st'_0}{\mtbl_1}{\S_1},\]
      and we get that:
      \begin{itemize}
        \item \devirtst{\ol\Ty\,\ol{\ty''}}{\mtbl_1}{\st'_0},
        \item \devirtm{\mtbl_1}, and
        \item $\forall\, \ol{\ty^x}, \ol{\st^x}.\quad
        \devirtst{\ol{\ty^x}}{\mtbl_0}{\ol{\st^x}}\ \implies\
        \devirtst{\ol{\ty^x}}{\mtbl_1}{\ol{\st^x}}$.
      \end{itemize}
      The fact that $\mtbl_1$ is maximally devirtualized by (2)
      and well-formed by \lemref{lem:wf}, lets us apply the
      induction hypothesis to the next instruction $\st_1$ of the method
      body \ol\st, and so on.  By repeating the same steps, we get
      that the final table obtained by compiling the body is
      also maximally devirtualized:
        \[\devirtm{\mtbl_{n+1}}.\]
      Now, let us look at the property (3) that we need to prove.
      Assuming we have some \ol{\ty^x}, \ol{\st^x} such that
      \devirtst{\ol{\ty^x}}{\mtbl }{\ol{\st^x}},
      we want to show that
      \devirtst{\ol{\ty^x}}{\mtbl'}{\ol{\st^x}}.
      First, observe that by case analysis on the derivation of
      \devirtst{\ol{\ty^x}}{\mtbl }{\ol{\st^x}},
      we can show that the property holds in $\mtbl_0$, i.e.
        \devirtst{\ol{\ty^x}}{\mtbl_0}{\ol{\st^x}}.
      For an instruction that is not a direct call, the presence of an additional
      method instance is irrelevant. If an instruction is a direct call to an
      existing method instance, it stays maximally devirtualized because
      $\mtbl_0$ does not remove or alter any existing instances of \mtbl.\\
      Next, we can apply the fact (3) from the induction on $\st_0$
      to conclude that \devirtst{\ol{\ty^x}}{\mtbl_1}{\ol{\st^x}},
      and so on, until we get that
        \[\devirtst{\ol{\ty^x}}{\mtbl_{n+1}}{\ol{\st^x}}.\]
      Since $\mtbl'$ only replaces the body of $\msig\m{\ol\Ty}$,
      by case analysis similar to the above, we conclude
        \[\devirtst{\ol{\ty^x}}{\mtbl'}{\ol{\st^x}},\]
      which \emph{proves (3)} for \SC{C-Instance}.
      Proceeding with similar reasoning, we can chain facts (3) from all
      intermediate compilations and apply them to facts (1) of the induction,
      so we get that
        \[\forall i.\quad \devirtst{\ol\Ty\,\ol{\ty''}}{\mtbl'}{\st'_i},\]
      and thus by \SC{D-Seq},
        \[\devirtst{\ol\Ty\,\ol{\ty''}}{\mtbl'}{\ol{\st'}}.\]
      Since the compilation does not remove or alter any original methods,
      $\origmtbl{\mtbl'} = \origmtbl{\mtbl}$, and thus
      type inference \typeinf{\origmtbl{\mtbl'}}{\ol\Ty}{\ol\st} produces
      the same results as \typeinf{\origmtbl{\mtbl}}{\ol\Ty}{\ol\st}.
      Therefore, the requirement of \SC{D-Table} for the method instance
      $\langle \meth\m{\ol\Ty}{\ol{\v k}}\ol{\st'},\ \ol{\ty'} \rangle$
      in $\mtbl'$ is satisfied to \emph{prove (2)}.\\
      Finally, we can \emph{prove (1)} for $\mtbl'$ by \SC{D-Direct}
      because $\msig\m{\ol\Ty} \in \mtbl'$.
  \end{itemize}
}
\end{proof}


Putting it all together, we have shown that type-grounded methods
compile to method instances without dynamically dispatched calls.
\PAPERVERSION{
Namely, since a well-formed original table
is trivially maximally devirtualized, and compilation preserves this property
by \lemref{lem:max-devirt},
executing a well-formed program results in a method table where
all method instances are maximally devirtualized. Furthermore, type-grounded
instances are fully devirtualized by \lemref{lem:full-virt},
and thus do not contain dispatched calls.
\begin{theorem}[Grounded Methods Never Dispatch]
  For any initial well-formed \mtbl, if
  \[ \config{\done\ \main()}{\mtbl} \stepmul
    \config{\stk{\env\ (\cond\i\j{\call{\msig\m{\ol\Ty}}{\ol{\v\k}}}{\v\l})\,\ol\st}{\frm}}{\mtbl'}
    \quad
    \wedge
    \quad
    \text{\msig{\m}{\ol\Ty} is type-grounded in $\mtbl'$,}
  \]
  then \msig{\m}{\ol\Ty} is fully devirtualized in $\mtbl'$.
\end{theorem}
}{
\begin{proof}
  First, observe that the initial table \mtbl is maximally devirtualized because
  it does not contain any compiled method instances. Then, by induction on
  \stepmul, we can show that the dynamic semantics preserves maximal
  devirtualization of the method table.
  Namely, by case analysis on one step of the dynamic semantics, we can see that the
  only rule that affects the method table is \SC{Disp}; all other rules simply
  propagate the same maximally devirtualized table.

  In the case of \SC{Disp}, by \lemref{lem:max-devirt},
  the new method table $\mtbl'$ is also maximally devirtualized.
  Since every step of the dynamic semantics preserves maximal devirtualization
  of the method table, so does the multi-step relation \stepmul.

  If the execution gets to a direct call, the call must have been produced
  by the JIT compiler.  As all compiled code is maximally devirtualized,
  by \SC{D-Direct}, the method instance \msig\m{\ol{\Ty}} exists,
  and, by \SC{D-Table}, it is known to be maximally devirtualized.
  Finally, as we know by \lemref{lem:full-virt}, any maximally devirtualized
  method instance that is type-grounded, is also fully devirtualized,
  and thus, it does not contain dispatched calls.
\end{proof}
}

}

\subsubsection{Relationship between Stability and Groundedness}\label{sec:stable}

\MODIFY{
While it is type groundedness that enables devirtualization,
the weaker property, type stability, is also important in practice.
Namely, stability of calles is crucial for groundedness
of the caller if the type inference algorithm
analyzes function calls using nothing but type information about the
arguments. \PAPERVERSIONINLINE{(A formal definition and a proof of this property
can be found in the extended version~\cite{oopsla21jules:arx}.)}{Since
types are the object of the analysis, we call such type inference
\emph{context-insensitive}: no information about the calling context
other than types is available to the analysis of a function call.}
A more powerful type inference algorithm might be able to work around
unstable methods in some cases, but even then, stability would be needed in general.

As an example, consider two type-unstable callees that return either an integer
or a string depending on the argument value, \c{f(x) = (x>0) ? 1 : "1"}
and \c{g(x) = (rand()>x) ? 1 : "1"}, and two calls, \c{f(y)} and \c{g(y)}, where
\c{y} is equal to $0.5$. If only type information about \c{y} is available to
type inference, both \c{f(y)} and \c{g(y)} are deemed to have abstract return
types. Therefore, the results of such function calls cannot be stored in
concretely typed registers, which immediately makes the calling code ungrounded.
If type inference were to analyze the value of \c{y} (not just its type),
the result of \c{f(y)} could be stored in a concrete, integer-valued register,
as only the first branch of \c{f} is known to execute. However, the other call,
\c{g(y)}, would still have to be assigned to an abstract register.
Thus, to enable groundedness and optimization of the client code regardless of
the value of argument \c{x}, \c{f(x)} and \c{g(x)} need to be type-stable.


Note that stability is a necessary but not sufficient
condition for groundedness of the client,
as conditional branches may lead to imprecise types,
like in the \c{f0} example from \figref{fla}.

}

\PAPERVERSION{}{
In what follows, we give a formal definition of
context-insensitive type inference, and
show that in that case, \emph{reachable} callees of a grounded method need to be
type stable.\ADD{
The intuition behind the definition is that inferring the type of
a function call in the context of a program gives us the same amount of
information as inferring the return type of the callee independently.
Technical complexities of the definition come from the fact that
call instructions are conditional: if the callee is not reachable,
its return type does not need to be compatible with the register
of the calling instruction.
}

\MODIFY{
\begin{definition}[Context-Insensitive Type Inference]
  A type inference algorithm $\typeinfop$ is
  \emph{context-insensitive}, if, 
  given any method table \mtbl, register typing \ol\ty, and instructions \ol\st
  such that (1)~type inference succeeds,
  \[
    \typeinf{\mtbl}{\ol\ty}{\ol\st} = \ol{\ty'},
  \]
  (2) instructions contain a function call \call\m{\ol{\v\k}},
  \[
    \ol\st = \ldots \st_i \ldots \quad \land \quad
    \st_i = \cond\i\j{\call\m{\ol{\v\k}}}{\v\l},
  \]
  and (3) the call is reachable by running \ol\st in a compatible environment,
  \[
    \exists \env.\ \typeof(\env) <: \ol\ty \quad \land \quad
    \config{\env\ \ol\st}{\mtbl} \stepmul
      \config{\stk{(\env'\ \ol{\st'})}{(\env''\ \st_i\ldots)}}{\mtbl'},
  \]
  where \ol{\st'} is the body of the callee, that is
  $\ol{\st'} = \body(\dispatch{\mtbl}{\m}{\typeof(\env')})$,
  the inferred type of the calling instruction
  is no more precise than
  the inferred return type of the callee:
  \[
    \last{\typeinf{\mtbl}{\typeof(\env')}{\ol{\st'}}} <: \ty'_i.
  \]
\end{definition}

\begin{theorem}[Type Stability of Callees Reachable from Type-Grouned Code]
  If type inference is context-insensitive,
  then for all type-grounded sequences of instructions \ol\st where
  \[
    \typeinf{\mtbl}{\ol\Ty}{\ol\st} = \ol{\Ty'},
  \]
  any reachable callee in \ol\st is type stable for the types of its arguments.
\end{theorem}
\begin{proof}
  The property follows from the definition of context-insensitivity.
  Let $\st_i$ be a reachable call \cond\i\j{\call\m{\ol{\v\k}}}{\v\l}.
  Since \ol\st is type-grounded, the inferred type of $\st_i$
  is some concrete $\Ty'_i$.
  Because by the definition of context-insensitivity, $\Ty'_i$ is no more precise
  than the inferred return type $\ty^m$ of the body,
  \[
    \last{\typeinf{\mtbl}{\typeof(\env')}{\ol{\st'}}} = \ty^m <: \Ty'_i,
  \]
  and concrete types do not have subtypes other than themselves,
  $\ty^m = \Ty'_i$. Thus, according to the definition of type stability,
  the method \m is type stable for $\typeof(\env')$.
\end{proof}
}
}

\subsection{Correctness of Compilation}\label{sec:jit-correct}

\ADD{
In this section, we prove that evaluating a program with just-in-time
compilation yields the same result as if no compilation occurred. We use two
instantiations of the \jules semantics, one (written~\stepd) where
\jit is the identity, and the other (written \stepj) where
\jit is defined as before. Our proof strategy is to
\begin{enumerate}
  \item define the optimization relation $\eqdirop$ which relates original and
    optimized code (\figref{fig:opt-config}),
  \item show that compilation from~\figref{fig:compile-wd} does produce
    optimized code (\thmref{thm:compilation-correct}),
  \item show that evaluating a program with \stepj or \stepd gives the same
    result (\thmref{thm:disp-jit-equiv}).
\end{enumerate}
The main \thmref{thm:disp-jit-equiv} is a corollary of bisimulation between
\stepd and \stepj (\lemref{lem:bisim-disp-jit}). The bisimulation lemma uses
\thmref{thm:compilation-correct} but otherwise, it has a proof similar to the
proof of bisimulation between original and optimized code both running with
\stepd (\lemref{lem:bisim-disp}). A corollary of the latter bisimulation lemma,
\thmref{thm:disp-equiv}, shows the correctness of \jit as an
ahead-of-time compilation strategy.
}

\begin{figure}
\small\begin{mathpar}
\PAPERVERSIONINLINE{}{
\framebox{\optstd{\ol\ty}{\st}{\st'}}

\qquad\qquad

\framebox{\optstd{\ol{\ty}}{\ol\st}{\ol{\st'}}}\\}

\inferrule[OI-Refl]{ }{
  \optstd{\ol\ty}{\st}{\st}
}

\inferrule[OI-Direct]{
  \st = \cond\i\j{\call\m{\ol{\v k}}}{\v\l} \and
  \ol{\Ty} = \ol\ty[\ol\k] \and \\\\
  \signature(\dispatch\mtbl\m{\ol\Ty})
    = \origsignature(\idx{\mtbl'}{\msig\m{\ol\Ty}}) \\\\
  \st'= \cond\i\j{\call{\msig{\m}{\ol\Ty}}{\ol{\v k}}}{\v\l}
}{
  \optstd{\ol\ty}{ \st }{ \st' }
}

\inferrule[OI-Seq]{
  \eqstd{ \ol{\ty} }{\st_0}{\st'_0} \\\\ \ldots \\\\
  \eqstd{ \ol{\ty} }{\st_n}{\st'_n}
}{
  \optstd{\ol{\ty}}{\ol\st}{\ol{\st'}}
}
\\

\PAPERVERSIONINLINE{}{\framebox{\eqmtbld}\\}

\inferrule[O-Table]{
  \mtbl = \origmtbl{\mtbl'} \and
  \forall \langle \meth\m{\ol{\Ty}}{}\ol{\st'},\ \ol{\ty} \rangle,
    \langle \meth\m{\ol{\ty}}{}\ol{\st},\ \ol{\ty} \rangle \in \mtbl',
    \text{ s.t. }
    \ol\Ty \neq \ol\ty,\,\ol{\st'} \neq \done. \\\\
  \qquad\quad \and
  \ol{\ty'} = \typeinf{\mtbl}{\ol\Ty}{\ol\st} \and
  \optstd{\ol{\Ty}\,\ol{\ty'}}{\ol\st}{\ol{\st'}}
}{
  \eqmtbld
}
\end{mathpar}
\begin{tabular}{rcll}
  \\
  $\Delta$ & ::= & \done\ |\ \stk{\ol\ty}{\Delta'} & stack typing \\
  \\
\end{tabular}
\begin{mathpar}
\PAPERVERSIONINLINE{}{\framebox{\eqstd{\Delta}{\frm}{\frm'}}\\}

\inferrule[O-StackEmpty]{ }{
  \eqstd{\done}{\done}{\done}
}

\inferrule[O-Stack]{
  \eqstd{\ol\ty}{\ol\st}{\ol{\st'}} \and
  \eqstd{\Delta}{\frm}{\frm'}
}{
  \eqstd{\stk{\ol\ty}{\Delta}}
    {\stk{(\env\ \ol{\st})}{\frm}}
    {\stk{(\env\ \ol{\st'})}{\frm'}}
}
\\

\PAPERVERSIONINLINE{}{\framebox{\typeinfstd{\frm}{\Delta}}\\}

\inferrule[I-StackEmpty]{ }{
  \typeinfstd{\done}{\done}
}

\inferrule[I-Stack]{
  \exists \meth\m{\ol{\ty'}}{}\ol{\st_b} \in \mtbl. \and
  \env = \env' \env'' \and
  \frm \neq \done \ \implies\
    \config{\frm}{\mtbl} \stepd
    \config{\stk{\env'\ \ol{\st_b}}{\frm}}{\mtbl} \\\\
  \config{\stk{\env'\ \ol{\st_b}}{\frm}}{\mtbl} \stepdmul
    \config{\stk{\env\ \ol{\st}}{\frm}}{\mtbl} \and
  \ol\Ty = \typeof(\env') \and
  \ol{\ty''} = \typeinf{\mtbl}{\ol\Ty}{\ol{\st_b}} \\\\
  \ol\Ty <: \ol{\ty'} \and
  \ol\Ty\,\ol{\ty''} <: \ol\ty \and \and
  \typeinfstd{\frm}{\Delta}
}{
  \typeinfstd{\stk{(\env\ \ol{\st})}{\frm}}{\stk{\ol\ty}{\Delta}}
}
\\

\PAPERVERSIONINLINE{}{\framebox{\optconfig{\frm}{\mtbl}{\frm'}{\mtbl'}{\Delta}}\\}

\inferrule[O-Config]{
  \typeinfstd{\frm}{\Delta} \and
  \eqstd{\Delta}{\frm}{\frm'} \and
  \eqmtbld
}{
  \optconfig{\frm}{\mtbl}{\frm'}{\mtbl'}{\Delta}
}
\end{mathpar}
\caption{Optimization relation for instructions, method tables,
  stacks, and configurations}
\label{fig:opt-config}
\Description{Optimization relation for instructions, method tables,
  stacks, and configurations}
\end{figure}

\ADD{
First, \figref{fig:opt-config} defines optimization relations for instructions,
method tables, stacks, and configurations. According to the definition of
\eqmtbld, method table $\mtbl'$ optimizes \mtbl if (1) it has all the same
original methods\footnote{\origmtbl{\mtbl'} denotes the method table containing
all original methods of $\mtbl'$ without compiled instances.}, and (2) bodies of
compiled method instances optimize the original methods using more specific type
information available to the instance. These requirements guarantee that
dispatching in an original and optimized method tables will always resolve to
related methods. Optimization of instructions only allows for replacing
dynamically dispatched calls with direct calls in the optimized table.
Optimization of stacks ensures that for all frames, instructions are optimized
accordingly, and requires all value environments to coincide. The first premise
of configuration optimization \SC{O-Config} guarantees that the original
configuration \config{\frm}{\mtbl} is obtained by calling methods from the
original method table \mtbl, and bodies of those methods are amenable to type
inference. Based on the results of type inference, stacks $\frm, \frm'$ need
to be related in method tables $\mtbl, \mtbl'$, and the method tables themselves
also need to be related.

As we show below, when run with the dispatch semantics, related configurations
are guaranteed to run in lock-step and produce the same result.

\begin{lemma}[Bisimulation of Related Configurations]\label{lem:bisim-disp}
  For any well-formed method tables \mtbl and $\mtbl'$
  (i.e. $\mathit{WF}(\mtbl), \mathit{WF}(\mtbl')$ according to
  \defref{ref:def-well-formed}) where $\mtbl'$ does not have stubs
  (i.e. all method bodies $\neq \done$),
  any stacks $\frm_1$, $\frm'_1$, and stack typing $\Delta_1$
  that relates the configurations,
$\optconfig{\frm_1}{\mtbl}{\frm'_1}{\mtbl'}{\Delta_1}$,
  the following holds:
  \begin{enumerate}
    \item Forward direction:
      \[
      \begin{array}{l}
        \config{\frm_1}{\mtbl} \stepd \config{\frm_2}{\mtbl}
        \implies
        \exists \frm'_2, \Delta'_1.\quad
          \config{\frm'_1}{\mtbl'} \stepd \config{\frm'_2}{\mtbl'}\ \ \land\ \
          \optconfig{\frm_2}{\mtbl}{\frm'_2}{\mtbl'}{\Delta'_1}.
      \end{array}
      \]
    \item Backward direction:
      \[
      \begin{array}{l}
        \config{\frm'_1}{\mtbl'} \stepd \config{\frm'_2}{\mtbl'}
        \implies
        \exists \frm_2, \Delta'_1.\quad
          \config{\frm_1}{\mtbl} \stepd \config{\frm_2}{\mtbl}\ \ \land\ \
          \optconfig{\frm_2}{\mtbl}{\frm'_2}{\mtbl'}{\Delta'_1}.
      \end{array}
      \]
  \end{enumerate}
\end{lemma}

\begin{proof}
\PAPERVERSION{
  For both directions, the proof goes by case analysis on the optimization
  relation for configurations and case analysis on the step of execution.
  The most interesting cases are the ones where \config{\frm_1}{\mtbl}
  steps by \SC{Disp}, with \config{\frm'_1}{\mtbl'} stepping by
  \SC{Disp} or \SC{Direct}. The key observation is that because of
  well-formedness of method tables and \eqmtbld, the optimized configuration
  \config{\frm'_1}{\mtbl'} correctly dispatches or directly invokes
  a specialized method instance of the same original method
  that \config{\frm_1}{\mtbl} dispatches to.
  A complete proof is available
  in the extended version of the paper~\cite{oopsla21jules:arx}.
}{
  For both directions, the proof goes by case analysis on the optimization
  relation for configurations and case analysis on the step of execution.
  By analyzing the derivation of the optimization relation
  \optconfig{\frm_1}{\mtbl}{\frm_2}{\mtbl'}{\Delta_1},
  we have three assumptions, \SC{HS}, \SC{HM}, and \SC{H$\Delta$}:
  \begin{mathpar}
    \inferrule*[right=O-Config]{
      \inferrule[HI]{}{ \typeinfstd{\frm_1}{\Delta_1} } \and
      \inferrule[HS]{}{ \eqstd{\Delta_1}{\frm_1}{\frm'_1} } \and
      \inferrule[HM]{}{ \eqmtbld }
    }{
      \optconfig{\frm_1}{\mtbl}{\frm'_1}{\mtbl'}{\Delta_1}
    }
  \end{mathpar}
  The assumptions will be referenced in the proof below.\\
  \textbf{(1) Forward direction.}\\
  To make a step, $\frm_1$ has to have at least one stack frame,
  so \SC{HS} has the following form:
  \[
    \eqstd
        {\stk{\ol\ty}{\Delta}}
        {\stk{(\env\ \ol{\st_p})}{\frm }}
        {\stk{(\env\ \ol{\st'_p})}{\frm'}},
  \]
  where $\frm_1 = \stk{(\env\ \ol{\st_p})}{\frm}$,
  $\frm'_1 = \stk{(\env\ \ol{\st'_p})}{\frm'}$,
  and $\Delta_1 = \stk{\ol\ty}{\Delta}$.
  Let us analyze the top frame.

  \emph{(a) If the sequence of instructions is empty},
  i.e. $\ol{\st_p} = \done$, the only possible step for \config{\frm_1}{\mtbl}
  is by \SC{Ret}. Therefore, $\frm = \stk{(\env'\ \st\,\ol\st)}{\frm_r}$
  and $\frm_2 = \stk{(\env'\!+\!\last{\env}\ \ol\st)}{\frm_r}$.
  As \eqstd{\ol\ty}{\ol{\st_p}}{\ol{\st'_p}}, it has to be that
  $\ol{\st'_p} = \done$. By analyzing the last premise of \SC{HS},
  \eqstd{\Delta}{\frm}{\frm'}, we know that
  $\frm' = \stk{(\env'\ \st'\,{\ol\st'})}{\frm'_r}$, and thus
  $\config{\frm'_1}{\mtbl'}$ can step by \SC{Ret} accordingly:
  \[
    \config{\stk{(\env\ \done)}{\stk{(\env'\!+\!\last{\env}\ \ol{\st'})}{\frm'_r}}}{\mtbl'}
    \stepd
    \config{\stk{(\env'\!+\!\last{\env}\ \ol{\st'})}{\frm'_r}}{\mtbl'}.
  \]
  By \SC{HI} and the last premise of \SC{HI}, i.e.
  \typeinfstd{\stk{(\env'\ \st\,\ol\st)}{\frm_r}}{\stk{\ol{\ty'}}{\Delta_r}}
  where $\Delta = \stk{\ol{\ty'}}{\Delta_r}$,
  combined with $\config{\frm}{\mtbl} \stepdmul \config{\frm_1}{\mtbl}
    \stepd \config{\stk{(\env'\!+\!\last{\env}\ \ol\st)}{\frm_r}}{\mtbl}$,
  we can conclude $\SC{HI}'$:
  \[
    \typeinfstd{\stk{(\env'\!+\!\last{\env}\ \ol\st)}{\frm_r}}{\stk{\ol{\ty'}}{\Delta_r}}.
  \]
  By analyzing the derivation \SC{O-Stack} of \eqstd{\Delta}{\frm}{\frm'},
  we know that \eqstd{\Delta_r}{\frm_r}{\frm'_r} and
  \eqstd{\ol{\ty'}}{\env'\ \st\,\ol\st}{\env'\ \st'\,\ol{\st'}}.
  It is easy to see that
  \eqstd{\ol{\ty'}}{\env'\!+\!\last{\env}\ \ol\st}{\env'\!+\!\last{\env}\ \ol{\st'}}
  also holds by \SC{OI-Seq}, and thus we can conclude $\SC{HS}'$:
  \[
    \eqstd
      {\stk{\ol{\ty'}}{\Delta_r}}
      {\stk{(\env'\!+\!\last{\env}\ \ol\st)}{\frm_r}}
      {\stk{(\env'\!+\!\last{\env}\ \ol{\st'})}{\frm'_r}}.
  \]
  Putting it all together, by $\SC{HI}', \SC{HS}',$ and \SC{HM},
  we get \optconfig{\frm_2}{\mtbl}{\frm'_2}{\mtbl'}{\Delta}, i.e.:
  \[
    \optconfig
      {\stk{(\env'\!+\!\last{\env}\ \ol\st)}{\frm_r}}{\mtbl}
      {\stk{(\env'\!+\!\last{\env}\ \ol{\st'})}{\frm'_r}}{\mtbl'}
      {\stk{\ol{\ty'}}{\Delta_r}}.
  \]

  \emph{(b) If the top frame contains a non-empty sequence of instructions},
  i.e. $\ol{\st_p} = \st\,\ol\st$,
  \SC{HS} has the following form:
  \[
    \eqstd
        {\stk{(\ol\ty^\env \ty\,\ol\ty)}{\Delta}}
        {\stk{(\env\ \st \,\ol{\st })}{\frm }}
        {\stk{(\env\ \st'\,\ol{\st'})}{\frm'}},
  \]
  where $\frm_1 = \stk{(\env\ \st \,\ol{\st })}{\frm }$,
  $\frm'_1 = \stk{(\env\ \st'\,\ol{\st'})}{\frm'}$,
  and $\Delta_1 = \stk{(\ol\ty^\env \ty\,\ol\ty)}{\Delta}$.
  By analyzing the derivation \SC{O-Stack} of \SC{HS},
  we get two assumptions, \SC{HOpt} for optimization of the top-frame
  instructions,
  \eqstd{\ol\ty^\env \ty\,\ol\ty}{\st \,\ol{\st }}{\st'\,\ol{\st'}},
  and $\SC{HS}'$ for optimization of the residual stack,
  \eqstd{\Delta}{\frm}{\frm'}.
  The first premise of \SC{HOpt}, $\SC{HOpt}_0$, indicates optimization
  for the first instruction, \eqstd{\ol\ty^\env\ty\,\ol\ty}{\st}{\st'}.
  By case analysis on the step
  \config{\frm_1}{\mtbl} \stepd \config{\frm_2}{\mtbl},
  we can always find a corresponding step for \config{\frm'_1}{\mtbl'}:
  \begin{itemize}
    \item Case \SC{Prim}. Here, $\st = \ass\i\p$,
      and the configuration steps as:
      \[
        \config{\stk{(\env\ \st\,\ol\st)}{\frm}}{\mtbl} \stepd
        \config{\stk{(\env\!+\!\p\ \ol\st)}{\frm}}{\mtbl}.
      \]
      By analyzing $\SC{HOpt}_0$, we can see that the first instruction
      $\st' = \st$ by \SC{OI-Refl}, and thus the optimized configuration
      $\config{\frm'_1}{\mtbl'}$ has to step by \SC{Prim}:
      \[
        \config{\stk{(\env\ \st\,\ol{\st'})}{\frm'}}{\mtbl'} \stepd
        \config{\stk{(\env\!+\!\p\ \ol{\st'})}{\frm'}}{\mtbl'}.
      \]
      By analyzing \SC{HI} and recombining its premises with
      the \SC{Prim} step, we get:
      \[
        \typeinfstd
          {\stk{(\env\!+\!\p\ \ol{\st })}{\frm}}
          {\stk{(\ol\ty^{\env} \ty\,\ol\ty)}{\Delta}}.
      \]
      Since the rest of the instructions $\ol\st, \ol{\st'}$
      and stacks $\frm, \frm'$ did not change,
      putting it all together, by \SC{O-Stack} we get:
      \[
        \eqstd
            {\stk{(\ol\ty^{\env} \ty\,\ol\ty)}{\Delta}}
            {\stk{(\env\!+\!\p\ \ol{\st })}{\frm }}
            {\stk{(\env\!+\!\p\ \ol{\st'})}{\frm'}},
      \]
      and thus by \SC{O-Config}, we have the desired conclusion:
      \[
        \optconfig
          {\stk{(\env\!+\!\p\ \ol{\st})}{\frm }}{\mtbl}
          {\stk{(\env\!+\!\p\ \ol{\st'})}{\frm'}}{\mtbl'}
          {\stk{(\ol\ty^{\env} \ty\,\ol\ty)}{\Delta}}.
      \]

    \item Cases \SC{Reg}, \SC{New}, \SC{Field}, \SC{False1}, and \SC{False2}
      proceed similarly to \SC{Prim}.

    \item Case \SC{Disp}. Here, $\st = \cond\i\j{\call\m{\ol{\v\k}}}{\v\l}$,
      and the configuration steps as:
      \[
        \config{\stk{(\env\ \st\,\ol\st)}{\frm}}{\mtbl} \stepd
        \config{\stk{(\env'\ \ol{\st_b})}{\stk{(\env\ \st\,\ol\st)}{\frm}}}{\mtbl},
      \]
      where $\env' = \idx\env{\ol\k}$, $\ol\Ty = \typeof(\env')$, and
      $\ol{\st_b} = \body(\dispatch{\mtbl}{\m}{\ol\Ty})$.
      Note that because \config{\frm_1}{\mtbl} does make a step by assumption,
      we know that dispatch is defined for the given arguments.
      Let's denote the dispatch target in \mtbl as
      $\langle \meth\m{\ol{\ty_o}}{}\ol{\st_b},\ \ol{\ty_o} \rangle.$
      By the definition of dispatch, we know that $\ol\Ty <: \ol{\ty_o}$.
      Therefore, by well-formedness of \mtbl and monotonicity of $\typeinfop$,
      type inference succeeds
      $\typeinf{\mtbl}{\ol\Ty}{\ol{\st_b}} = \ol{\ty'}$
      (\mtbl contains only original methods, so $\origmtbl{\mtbl} = \mtbl$).
      Thus, we have:
      \[
        \typeinfstd
          {\stk{(\env'\ \ol{\st_b})}{\stk{(\env\ \st\,\ol\st)}{\frm}}}
          {\stk{(\ol\Ty\,\ol{\ty'})}{\stk{(\ol\ty^\env \ty\,\ol\ty)}{\Delta}}},
      \]
      where $\Delta'_1 =
      \stk{(\ol\Ty\,\ol{\ty'})}{\stk{(\ol\ty^\env \ty\,\ol\ty)}{\Delta}}$.
      Next, let us consider $\frm'_1$.
      There are two possibilities for $\st'$.
      \begin{enumerate}
        \item If $\SC{HOpt}_0$ is built with \SC{OI-Refl}, then
          $\frm'_1 = \stk{(\env\ \st\,\ol{\st'})}{\frm'}$
          where the first instruction is exactly the same dispatch call, i.e.
          $\st' = \st = \cond\i\j{\call\m{\ol{\v\k}}}{\v\l}$.
          The configuration
          \config{\stk{(\env\ \st\,\ol{\st'})}{\frm'}}{\mtbl'}
          could either step by \SC{Disp} or err if
          \dispatch{\mtbl'}{\m}{\ol\Ty} is undefined.
          Let us inspect the possibility of the latter.
          By \SC{HM}, $\mtbl'$ optimizes \mtbl, so $\mtbl'$ contains the same
          original method \msig{\m}{\ol{\ty_o}}. Thus, dispatch cannot fail
          due to the lack of applicable methods. As $\mtbl'$ is well-formed,
          we also know that if there is a specialized method instance,
          it is for the best original method and ambiguity is not possible.
          Thus, \dispatch{\mtbl'}{\m}{\ol\Ty} succeeds,
          and \config{\frm'_1}{\mtbl'} steps by \SC{Disp}:
          \[
            \config{\stk{(\env\ \st\,\ol{\st'})}{\frm'}}{\mtbl'} \stepd
            \config{\stk{(\env'\ \ol{\st'_b})}{\stk{(\env\ \st\,\ol{\st'})}{\frm'}}}{\mtbl'}.
          \]
          There are two possibilities for \ol{\st'_b}.
          (a) If $\mtbl'$ does not contain a specialization, then
          $\ol{\st'_b} = \ol{\st_b}$. By reflexivity of the optimization relation,
          we then have:
          \[
            \eqstd
              {\ol\Ty\,\ol{\ty'}}
              {\ol{\st_b}}
              {\ol{\st_b}}.
          \]
          (b) If $\mtbl'$ contains a specialized method instance
          $\langle \meth\m{\ol{\Ty}}{}\ol{\st'_b},\ \ol{\ty_o} \rangle$,
          then $\ol{\st'_b} \neq \done$ by the assumption on $\mtbl'$,
          and thus the desired relation is guaranteed by \SC{HM}:
          \[
            \eqstd
              {\ol\Ty\,\ol{\ty'}}
              {\ol{\st_b}}
              {\ol{\st'_b}}.
          \]

        \item If $\SC{HOpt}_0$ is built with \SC{OI-Direct}, then
          $\st'$ is a direct call, i.e.
          $\st' = \cond\i\j{\call{\msig{\m}{\ol{\Ty'}}}{\ol{\v k}}}{\v\l}$
          where $\ol{\Ty'} = \idx{\ol\ty^\env}{\ol\k}$,
          and method instance \msig{\m}{\ol{\Ty'}} is in $\mtbl'$.
          Note that by \SC{HI} and the soundness of type inference,
          we know that
          $\typeof(\idx\env{\ol\k}) <: \typeof(\idx{\ol\ty^\env}{\ol\k})$,
          i.e. $\ol{\Ty} <: \ol{\Ty'}$.
          Since both are concrete types, and concrete types do not have subtypes
          other than themselves, it has to be the case that $\ol{\Ty'} = \ol{\Ty}$
          and $\msig{\m}{\ol\Ty'} = \msig{\m}{\ol\Ty}$.
          Thus, \config{\frm'_1}{\mtbl'} calls \msig{\m}{\ol\Ty} and
          steps by \SC{Direct}:
          \[
            \config{\stk{(\env\ \st'\,\ol{\st'})}{\frm'}}{\mtbl'} \stepd
            \config{(\env'\ \stk{\ol{\st'_b})}{\stk{(\env\ \st'\,\ol{\st'})}{\frm'}}}{\mtbl'}.
          \]
          By \SC{OI-Direct}, \meth\m{\ol{\Ty}}{}\ol{\st'_b} specializes
          \msig{\m}{\ol{\ty_o}}. Since $\mtbl'$ does not have stubs,
          we have by \SC{HM}:
          \[
            \eqstd
              {\ol\Ty\,\ol{\ty'}}
              {\ol{\st_b}}
              {\ol{\st'_b}}.
          \]
      \end{enumerate}
      Because the new top frames $\env'\ \ol{\st_b}$ and
      $\env'\ \ol{\st'_b}$ are related in both cases, and the rest of
      the stacks did not change, we get that the entire stacks
      $\frm_2, \frm'_2$ are related:
      \[
        \eqstd
            {\stk{(\ol\Ty\,\ol{\ty'})}{\stk{(\ol\ty^\env \ty\,\ol\ty)}{\Delta}}}
            {\stk{(\env'\ \ol{\st_b })}{\stk{(\env\ \st \,\ol{\st })}{\frm }}}
            {\stk{(\env'\ \ol{\st'_b})}{\stk{(\env\ \st'\,\ol{\st'})}{\frm'}}}.
      \]
      And thus we get the desired:
      \[
        \optconfig{\frm_2}{\mtbl}{\frm'_2}{\mtbl'}{\Delta'_1}.
      \]

    \item Case \SC{Direct}. This case is not possible because
      \frm is obtained by running methods of \mtbl, and \mtbl consists of
      only original methods, which do not contain direct calls.
  \end{itemize}

  \textbf{(2) Backward direction.}
  This direction is similar to the forward direction. Because the structure
  of $\frm'_1$ matches the structure of $\frm_1$, we can always find
  the corresponding step for $\config{\frm_1}{\mtbl}$. The most interesting
  cases of
  \[
    \config{\frm'_1}{\mtbl'} \stepd \config{\frm'_2}{\mtbl'}
  \]
  are \SC{Disp} and \SC{Direct}, and some details on these are provided below.\\
  In both cases, we have
  $\frm_1 = \stk{(\env\ \st \,\ol{\st })}{\frm }$,
  $\frm'_1 = \stk{(\env\ \st'\,\ol{\st'})}{\frm'}$,
  and $\Delta_1 = \stk{(\ol\ty^\env \ty\,\ol\ty)}{\Delta}$,
  and we know that configurations step by making a function call.
  As a result, we get:
  \[
    \optconfig
      {\stk{(\env'\ \ol{\st_b })}{\stk{(\env\ \st \,\ol{\st })}{\frm }}}
      {\mtbl}
      {\stk{(\env'\ \ol{\st'_b})}{\stk{(\env\ \st'\,\ol{\st'})}{\frm'}}}
      {\mtbl'}
      {\stk{(\ol\Ty\,\ol{\ty'})}{\stk{(\ol\ty^\env \ty\,\ol\ty)}{\Delta}}},
  \]
  where $\ol{\ty'} = \typeinf{\mtbl}{\ol\Ty}{\ol{\st_b}}$ just like
  in the forward direction.

  \begin{itemize}
    \item Case \SC{Disp}. Here, $\st' = \cond\i\j{\call\m{\ol{\v\k}}}{\v\l}$,
      and the configuration steps as:
      \[
        \config{\stk{(\env\ \st'\,\ol{\st'})}{\frm'}}{\mtbl'} \stepd
        \config{\stk{(\env'\ \ol{\st'_b})}{\stk{(\env\ \st'\,\ol{\st'})}{\frm'}}}{\mtbl'},
      \]
      where $\env' = \idx\env{\ol\k}$, $\ol\Ty = \typeof(\env')$, and
      $\ol{\st'_b} = \body(\dispatch{\mtbl'}{\m}{\ol\Ty})$.
      There are two options for \ol{\st'_b}: it is the body of either the
      most applicable original method or its specialization in $\mtbl'$.
      Note that $\SC{HOpt}_0$ can be built only with \SC{OI-Refl},
      so $\st = \st'$. Because by \SC{HM}, $\mtbl = \origmtbl{\mtbl'}$,
      \mtbl has the same most applicable original method as $\mtbl'$,
      and thus \config{\frm_1}{\mtbl} steps by \SC{Disp}.
      By reasoning similar to the forward direction, we get that
      the new top frames (as well as entire configurations) are related
      by the optimization relation.

    \item Case \SC{Direct}. Here,
      $\st' = \cond\i\j{\call{\msig{\m}{\ol\Ty}}{\ol{\v k}}}{\v\l}$,
      and the argument about \SC{HI} and the soundness of type inference
      applies similarly to the case (2) of \SC{Disp} of the forward direction.
      As $\SC{HOpt}_0$ can be built only by \SC{OI-Direct}, we know
      $\st = \cond\i\j{\call\m{\ol{\v\k}}}{\v\l}$. Furthermore, dispatch is
      defined in \mtbl by \SC{OI-Direct}, and thus \config{\frm_1}{\mtbl}
      steps by \SC{Disp}. By reasoning similar to the forward direction,
      by \SC{HM}, we know that the instructions in the new top frames are
      related, and thus the entire configurations are also related.
  \end{itemize}
}
\end{proof}

\begin{theorem}[Correctness of Optimized Method Table]\label{thm:disp-equiv}
  For any well-formed method tables \mtbl and $\mtbl'$ where (1) \eqmtbld,
  (2) table $\mtbl'$ does not have stubs,
  and (3) $\langle \meth{\main}{\done}{}{\ol\st}, \done \rangle \in \mtbl$,
  \[
    \config{\done\ \ol\st}{\mtbl} \stepdmul \config{\env\ \done}{\mtbl}
    \ \ \iff\ \
    \config{\done\ \ol\st}{\mtbl'} \stepdmul \config{\env\ \done}{\mtbl'},
  \]
  i.e. program \ol\st runs to the same final value environment in both tables.
\end{theorem}
\begin{proof}
\PAPERVERSION{
  This is a corollary of \lemref{lem:bisim-disp};
  detailes are in~\cite{oopsla21jules:arx}.
}{
  This is a corollary of \lemref{lem:bisim-disp}.
  By reflexivity of the optimization relation and well-formedness of \mtbl, we know:
  \[
    \optconfig{\ol\st}{\mtbl}{\ol\st}{\mtbl'}{\ol\ty},
  \]
  where $\ol\ty = \typeinf{\mtbl}{\done}{\ol\st}$.
  Thus, the bisimulation lemma is applicable: if one of the configurations
  can make a step, so does the other, and the step leads to a pair of
  related configurations. Since method tables did not change,
  the lemma can be applied again to these configurations, and so on.
  If the program terminates and does not err,
  both sides arrive to final configurations where
  environments are guaranteed to coincide by \SC{O-Stack}.
}
\end{proof}

Next, we show that compilation as defined in \figref{fig:compile-wd}
produces optimized code in an optimized method table
according to \figref{fig:opt-config}.

\begin{theorem}[Compilation Satisfies Optimization Relation]\label{thm:compilation-correct}
  For any well-formed method tables \mtbl and $\mtbl'$, typing \ol\ty, and
  instruction \st, such that
  \[
    \eqmtbld \ \ \land\ \ \compst{\ol\ty}{\st}{\mtbl'}{}{\st''}{\mtbl''}{},
  \]
  it holds that:
  \begin{enumerate}
    \item $\eqst{\ol\ty}{\mtbl}{\mtbl''}{\st}{\st''},$
    \item $\eqmtbl{\mtbl}{\mtbl''},$
    \item and the optimization relation on $\mtbl, \mtbl'$
      is preserved on $\mtbl, \mtbl''$:\\
      $\forall \ol{\ty^x}, \st^x, \st^y.\quad
      \eqst{\ol{\ty^x}}{\mtbl}{\mtbl' }{\st^x}{\st^y} \ \ \implies\ \
      \eqst{\ol{\ty^x}}{\mtbl}{\mtbl''}{\st^x}{\st^y}$.
  \end{enumerate}
\end{theorem}

\begin{proof}
  By induction on the derivation of
  \compst{\ol\ty}{\st}{\mtbl'}{}{\st''}{\mtbl''}{}.
  Similar to the proof of \lemref{lem:max-devirt} on maximal devirtualization,
  the most interesting case is \SC{C-Instance} where
  compilation of additional instances happens.
\PAPERVERSION{
  The key step of the proof is to observe that the induction hypothesis
  is applicable to
  \compst{\ol\Ty\,\ol{\ty''}}{\st_0}{\mtbl_0}{\S_0}{\st'_0}{\mtbl_1}{\S_1},
  and then, using facts (2) and \lemref{lem:wf}, to all
  \compst{\ol\Ty\,\ol{\ty''}}{\st_i}{\mtbl_i}{\S_0}{\st'_i}{\mtbl_{i+1}}{\S_{i+1}}.
  Using facts (3), we can propagate the information about optimization
  of $\st_i, \st'_i$ and respective tables $\mtbl_i$ to the final $\mtbl''$.
  A complete proof is available in the extended version~\cite{oopsla21jules:arx}.
}{
  \begin{itemize}
    \item Cases \SC{C-NoDisp} and \SC{C-Disp} are straightforward:
      by reflexivity of the optimization relation, we immediately get
      $\eqst{\ol\ty}{\mtbl}{\mtbl''}{\st}{\st}$;
      by assumption \eqmtbld and $\mtbl'' = \mtbl'$,
      we also have \eqmtbl{\mtbl}{\mtbl''};
      property (3) also holds trivially because of $\mtbl'' = \mtbl'$.

    \item Case \SC{C-Direct}. Here, $\st = \cond\i\j{\call\m{\ol{\v\k}}}{\v\l}$
      and $\st'' = \cond\i\j{\call{\msig{\m}{\ol\Ty}}{\ol{\v k}}}{\v\l}$.
      Since $\mtbl'' = \mtbl'$, we have \eqmtbl{\mtbl}{\mtbl''} by assumption,
      and (3) holds trivially.\\
      If \msig\m{\ol\Ty} is an original method, then by
      $\origmtbl{\mtbl'} = \mtbl$ (because of \eqmtbld) and the properties
      of dispatch, \msig\m{\ol\Ty} has to be the method returned by
      \dispatch{\mtbl}{\m}{\ol\Ty}. Therefore,
      $\signature(\dispatch{\mtbl}{\m}{\ol\Ty}) = \ol\Ty =
        \signature(\idx{\mtbl''}{\msig\m{\ol\Ty}}) =
        \origsignature(\idx{\mtbl''}{\msig\m{\ol\Ty}}).$\\
      If \msig\m{\ol\Ty} is a compiled instance in $\mtbl'$, then by
      well-formedness of $\mtbl'$, we know that its original method is the
      most applicable method in $\origmtbl{\mtbl'}$.
      Since $\origmtbl{\mtbl'} = \mtbl$, we get
      $\signature(\dispatch{\mtbl}{\m}{\ol\Ty}) =
        \origsignature(\idx{\mtbl''}{\msig\m{\ol\Ty}})$.\\
      Thus, for both original and compiled \msig\m{\ol\Ty}, we have
      $\eqst{\ol\ty}{\mtbl}{\mtbl''}{\st}{\st''}$ by \SC{OI-Direct}.

    \item Case \SC{C-Instance}. Here, $\st = \cond\i\j{\call\m{\ol{\v\k}}}{\v\l}$
      and $\st'' = \cond\i\j{\call{\msig{\m}{\ol\Ty}}{\ol{\v k}}}{\v\l}$
      like in previous case, but $\mtbl''$ is obtained by compiling the body
      of method \meth{\m}{\ol{\ty'}}{}{\ol{\st}} that is a dispatch target
      of \dispatch{\mtbl'}{\m}{\ol\Ty} where $\ol{\ty'} \neq \ol\Ty$.
      Because of the latter condition, we know that \msig\m{\ol{\ty'}} has to
      be an original method in $\mtbl'$.
      Since $\origmtbl{\mtbl'} = \mtbl$, \msig\m{\ol{\ty'}}
      is also the most applicable method in \mtbl, so it has to be
      returned by \dispatch\mtbl{\m}{\ol\Ty}.\\
      Now, let us consider
      $\mtbl_0 = \mtbl' + \langle \msig\m{\ol{\Ty}}\done, \ol{\ty'} \rangle$.
      Since $\msig\m{\ol\Ty} \notin \mtbl'$ and $\mtbl_0$ only adds a new stub
      of a compiled instance to the well-formed $\mtbl'$, we know
      $\mathit{WF}(\mtbl_0)$. Furthermore, because \eqmtbld and $\mtbl_0$
      adds the stub without modifying anything else in $\mtbl'$,
      it is easy to show that all instruction optimizations
      $\eqst{\ol{\ty^x}}{\mtbl}{\mtbl' }{\st^x}{\st^y}$ are preserved
      by $\mtbl_0$, that is $\eqst{\ol{\ty^x}}{\mtbl}{\mtbl_0}{\st^x}{\st^y}$,
      and thus \eqmtbl{\mtbl}{\mtbl_0}.
      Using the result of type inference on the original body in \origmtbl{\mtbl'},
      i.e. $\ol{\ty''} = \typeinf{\mtbl}{\ol\Ty}{\ol{\st}}$,
      we can apply the induction hypothesis to
      \compst{\ol\Ty\,\ol{\ty''}}{\st_0}{\mtbl_0}{}{\st'_0}{\mtbl_1}{},
      which gives us:
      \begin{itemize}
        \item $\eqst{\ol\Ty\,\ol{\ty''}}{\mtbl}{\mtbl_1}{\st_0}{\st'_0},$
        \item $\eqmtbl{\mtbl}{\mtbl_1},$
        \item $\forall \ol{\ty^x}, \st^x, \st^y.\quad
          \eqst{\ol{\ty^x}}{\mtbl}{\mtbl_0}{\st^x}{\st^y} \ \ \implies\ \
          \eqst{\ol{\ty^x}}{\mtbl}{\mtbl_1}{\st^x}{\st^y}$.
      \end{itemize}
      The fact that \eqmtbl{\mtbl}{\mtbl_1} and $\mathit{WF}(\mtbl_1)$
      by \lemref{lem:wf}, lets us apply the induction
      hypothesis to the next instruction $\st_1$ of the method body $\ol\st$,
      which produces:
      \begin{itemize}
        \item $\eqst{\ol\Ty\,\ol{\ty''}}{\mtbl}{\mtbl_2}{\st_1}{\st'_1},$
        \item $\eqmtbl{\mtbl}{\mtbl_2},$
        \item $\forall \ol{\ty^x}, \st^x, \st^y.\quad
          \eqst{\ol{\ty^x}}{\mtbl}{\mtbl_1}{\st^x}{\st^y} \ \ \implies\ \
          \eqst{\ol{\ty^x}}{\mtbl}{\mtbl_2}{\st^x}{\st^y}$.
      \end{itemize}
      Now we can apply the latter to
      $\eqst{\ol\Ty\,\ol{\ty''}}{\mtbl}{\mtbl_1}{\st_0}{\st'_0}$,
      which gives us $\eqst{\ol\Ty\,\ol{\ty''}}{\mtbl}{\mtbl_2}{\st_0}{\st'_0}$.\\
      Proceeding in this manner, we get:
      \[
        \eqst{\ol\Ty\,\ol{\ty''}}{\mtbl}{\mtbl_{n+1}}{\ol\st}{\ol{\st'}}
      \]
      and
      \[
        \forall \ol{\ty^x}, \st^x, \st^y.\quad
          \eqst{\ol{\ty^x}}{\mtbl}{\mtbl'}{\st^x}{\st^y} \ \ \implies\ \
          \eqst{\ol{\ty^x}}{\mtbl}{\mtbl_{n+1}}{\st^x}{\st^y}.
      \]
      Finally, let us look at $\mtbl''$. Its only difference from $\mtbl_{n+1}$
      is the non-stub body \ol{\st'} for \msig\m{\ol\Ty}.
      Thus, all \eqst{\ol{\ty^x}}{\mtbl}{\mtbl_{n+1}}{\st^x}{\st^y}
      are trivially preserved by $\mtbl''$, which gives us
      \[
        \forall \ol{\ty^x}, \st^x, \st^y.\quad
          \eqst{\ol{\ty^x}}{\mtbl}{\mtbl' }{\st^x}{\st^y} \ \ \implies\ \
          \eqst{\ol{\ty^x}}{\mtbl}{\mtbl''}{\st^x}{\st^y}
      \]
      and
      \[
        \eqst{\ol\Ty\,\ol{\ty''}}{\mtbl}{\mtbl''}{\ol\st}{\ol{\st'}}.
      \]
      The latter lets us conclude \eqmtbl{\mtbl}{\mtbl''}.\\
      Finally, because $\signature(\dispatch\mtbl{\m}{\ol\Ty}) = \ol{\ty'}$
      and \msig\m{\ol\Ty} optimizes \msig\m{\ol{\ty'}} in $\mtbl''$,
      we get \[\signature(\dispatch{\mtbl}{\m}{\ol\Ty}) =
      \origsignature(\idx{\mtbl''}{\msig\m{\ol\Ty}})\]
      and conclude $\eqst{\ol\ty}{\mtbl}{\mtbl''}{\st}{\st''}$ by \SC{OI-Direct}.
  \end{itemize}
}
\end{proof}

Finally, using an auxiliary lemma about preserving stub methods
during compilation, we can show that the JIT-compilation semantics
is equivalent to the dispatch semantics.
\begin{lemma}[Preserving Stubs]\label{lem:comp-stubs}
  For any well-formed method tables $\mtbl'$ and $\mtbl''$, typing \ol\ty, and
  instruction \st, such that
  \[
    \compst{\ol\ty}{\st}{\mtbl'}{}{\st''}{\mtbl''}{},
  \]
  it holds that
  \[
    \{ \ol\Ty\ |\ \meth{\m}{\ol\Ty}{}{\done} \in \mtbl'  \} =
    \{ \ol\Ty\ |\ \meth{\m}{\ol\Ty}{}{\done} \in \mtbl'' \},
  \]
  i.e. the set of stubbed method instances is preserved by a compilation step.
\end{lemma}
\begin{proof}
  By induction on the derivation of
  \compst{\ol\ty}{\st}{\mtbl'}{}{\st''}{\mtbl''}{},
  similar to the proof of \lemref{lem:wf} on well formedness.
\PAPERVERSION{
See \cite{oopsla21jules:arx}.
}{

  All cases except for \SC{C-Instance} are trivial because the method table
  does not change. For \SC{C-Instance}, let's denote the set of stubs in
  $\mtbl'$ by $\S'$, that is:
  \[
    \S' = \{ \ol\Ty\ |\ \meth{\m}{\ol\Ty}{}{\done} \in \mtbl'  \}.
  \]
  Since $\msig\m{\ol\Ty} \notin \mtbl'$ and
  $\mtbl_0 = \mtbl' + \langle \msig\m{\ol{\Ty}}\done, \ol{\ty'} \rangle$,
  we have $\S_0 = \S' \cup \{\ol\Ty\}$.
  By applying the induction hypothesis to
  \compst{\ol\Ty\,\ol{\ty''}}{\st_0}{\mtbl_0}{}{\st'_0}{\mtbl_1}{},
  we know that the set $\S_1$ of stubs of $\mtbl_1$ is the same as~$\S_0$,
  and $\mtbl_1$ is well-formed by \lemref{lem:wf}.
  Proceeding by applying the induction hypothesis to compilation of
  all $\st_i$, we get that:
  \[\S_{n+1} = \S_0 = \S' \cup \{\ol\Ty\}.\]
  Finally, the only difference between $\mtbl''$ and $\mtbl_{n+1}$
  is that the stub for \msig\m{\ol\Ty} is replaced by an actual method body.
  Therefore, we get the desired property:
  \[ \S'' = \S_{n+1} \setminus \{\ol\Ty\} = \S'. \]
}
\end{proof}

\begin{lemma}[Bisimulation of Related Configurations with Dispatch and JIT Semantics]\label{lem:bisim-disp-jit}
  For any well-formed method tables \mtbl and $\mtbl'$
  where $\mtbl'$ does not have stubs,
  any frame stacks $\frm_1$, $\frm'_1$, and stack typing $\Delta_1$,
  such that
$  \optconfig{\frm_1}{\mtbl}{\frm'_1}{\mtbl'}{\Delta_1}$,
  the following holds:
  \begin{enumerate}
    \item Forward direction:
      \[
      \begin{array}{l}
        \config{\frm_1}{\mtbl} \stepd \config{\frm_2}{\mtbl}
        \implies
        \exists \frm'_2, \mtbl'', \Delta'_1.\quad
          \config{\frm'_1}{\mtbl'} \stepj \config{\frm'_2}{\mtbl''}\ \ \land\ \
          \optconfig{\frm_2}{\mtbl}{\frm'_2}{\mtbl''}{\Delta'_1}.
      \end{array}
      \]
    \item Backward direction:
      \[
      \begin{array}{l}
        \config{\frm'_1}{\mtbl'} \stepj \config{\frm'_2}{\mtbl''}
        \implies
        \exists \frm_2, \Delta'_1.\quad
          \config{\frm_1}{\mtbl} \stepd \config{\frm_2}{\mtbl}\ \ \land\ \
          \optconfig{\frm_2}{\mtbl}{\frm'_2}{\mtbl''}{\Delta'_1}.\qquad
      \end{array}
      \]
  \end{enumerate}
  Furthermore, $\mtbl''$ is well-formed and does not have stubs.
\end{lemma}

\begin{proof}
  By case analysis on the derivation of optimization
  \optconfig{\frm_1}{\mtbl}{\frm'_1}{\mtbl'}{\Delta_1}
  and case analysis on the step (\stepd for the forward and \stepj
  for the backward direction), similarly to the proof of
  bisimulation for the dispatch semantics (\lemref{lem:bisim-disp}).
  The only difference appears in cases where \config{\frm'_1}{\mtbl'}
  steps by \SC{Disp}: these are the only places where JIT compilation
  fires and $\mtbl''$ might be different from $\mtbl'$.
\PAPERVERSION{
  The key observation is that compilation produces $\mtbl''$ which
  (1) optimizes \mtbl by \thmref{thm:compilation-correct},
  \eqmtbl{\mtbl}{\mtbl''}, (2) is well-formed by \lemref{lem:wf},
  and (3) does not contain stubs by \lemref{lem:comp-stubs}.
  More details are available in the extended version
  of the paper~\cite{oopsla21jules:arx}.
}{
  As an example, we consider only the case of the forward direction
  where $\config{\frm_1}{\mtbl} \stepd \config{\frm_2}{\mtbl}$ steps
  by \SC{Disp}, and $\SC{HOpt}_0$ is built by \SC{OI-Refl},
  reusing all the notation from \lemref{lem:bisim-disp}.
  Thus, we have $\st = \st' = \cond\i\j{\call\m{\ol{\v\k}}}{\v\l}$,
  $\frm_1 = \stk{(\env\ \st\,\ol\st)}{\frm}$,
  $\frm'_1 = \stk{(\env\ \st\,\ol{\st'})}{\frm'}$,
  $\Delta_1 = \stk{(\ol\ty^\env \ty\,\ol\ty)}{\Delta}$,
  and \config{\frm_1}{\mtbl} steps as:
  \[
    \config{\stk{(\env\ \st\,\ol\st)}{\frm}}{\mtbl} \stepd
    \config{\stk{(\env'\ \ol{\st_b})}{\stk{(\env\ \st\,\ol\st)}{\frm}}}{\mtbl},
  \]
  where $\env' = \idx\env{\ol\k}$, $\ol\Ty = \typeof(\env')$, and
  $\ol{\st_b} = \body(\dispatch{\mtbl}{\m}{\ol\Ty})$.
  \config{\frm'_1}{\mtbl'} can step only by \SC{Disp}, which triggers
  JIT compilation.
  According to the definition of \jit from \figref{compile},
  there are two possibilities: either \msig\m{\ol\Ty} is already in $\mtbl'$,
  in which case $\mtbl'' = \mtbl'$, or \msig\m{\ol\Ty} is not in $\mtbl'$,
  in which case the new method instance gets compiled and added to $\mtbl''$.
  \begin{itemize}
    \item In the former case, the proof proceeds exactly as
      in~\lemref{lem:bisim-disp}.
    \item In the latter case, by \thmref{thm:compilation-correct}, we know that
      \eqmtbl{\mtbl}{\mtbl''} and all instruction optimizations on $\mtbl, \mtbl'$
      are preserved for $\mtbl, \mtbl''$. Therefore, we know:
      \[
        \eqst
            {\stk{(\ol\ty^\env \ty\,\ol\ty)}{\Delta}}
            {\mtbl}{\mtbl''}
            {\stk{(\env\ \st\,\ol{\st })}{\frm }}
            {\stk{(\env\ \st\,\ol{\st'})}{\frm'}}.
      \]
      Furthermore, as $\mtbl'$ is well-formed and does not have stubs
      by assumption, $\mtbl''$ is also well-formed by \lemref{lem:wf}
      and does not have stubs by \lemref{lem:comp-stubs}.
      Reasoning similarly to \lemref{lem:bisim-disp}, we can see that
      the body \ol{\st_b} of the original method returned by
      \dispatch{\mtbl}{\m}{\ol\Ty} optimizes to the body \ol{\st'_b}
      of the compiled instance \msig\m{\ol\Ty}. Thus, we get that
      \config{\frm'_1}{\mtbl'} steps by \SC{Disp},
      \[
        \config{\stk{(\env\ \st\,\ol{\st'})}{\frm'}}{\mtbl'} \stepj
        \config{\stk{(\env'\ \ol{\st'_b})}{\stk{(\env\ \st\,\ol{\st'})}{\frm'}}}{\mtbl''},
      \]
      the resulting configurations are related,
      \[
        \optconfig
          {\stk{(\env'\ \ol{\st_b })}{\stk{(\env\ \st \,\ol{\st })}{\frm }}}
          {\mtbl}
          {\stk{(\env'\ \ol{\st'_b})}{\stk{(\env\ \st'\,\ol{\st'})}{\frm'}}}
          {\mtbl''}
          {\stk{(\ol\Ty\,\ol{\ty'})}{\stk{(\ol\ty^\env \ty\,\ol\ty)}{\Delta}}},
      \]
      and $\mtbl''$ has no stubs.
  \end{itemize}
}
\end{proof}

\begin{theorem}[Correctness ot JIT]\label{thm:disp-jit-equiv}
  For any original well-formed method table \mtbl the following holds:
  \[
    \config{\done\ \ol\st}{\mtbl} \stepdmul \config{\env\ \done}{\mtbl}
    \ \ \iff\ \
    \config{\done\ \ol\st}{\mtbl} \stepjmul \config{\env\ \done}{\mtbl'},
  \]
  i.e. program \ol\st runs to the final environment \env with the dispatch
  semantics \stepd if and only if it runs to the same environment with
  the JIT-compilation semantics \stepj.
\end{theorem}
\begin{proof}
  This is a corollary of \lemref{lem:bisim-disp-jit}. Be reflexivity of the
  optimization relation and well-formedness of \mtbl, we know:
  $\optconfig{\ol\st}{\mtbl}{\ol\st}{\mtbl}{\ol\ty}$,
  where $\ol\ty = \typeinf{\mtbl}{\done}{\ol\st}$. Since \mtbl does not have
  stubs, the bisimulation lemma is applicable: if one of the configurations can
  make a step, so does the other, and the step leads to a pair of related
  configurations such that the lemma can be applied again. If the program
  terminates and does not err,
  both sides arrive to final configurations where environments are
  guaranteed to coincide by \SC{O-Stack}.
\end{proof}

}

\section{Empirical Study}\label{sec:empirical}

Anecdotal evidence suggests that type stability is discussed in the
Julia community, but does Julia code exhibit the properties of stability
and groundedness in practice? And if so, are there any indicators correlated with
instability and ungroundedness? To find out, we ran a dynamic analysis on a
corpus of Julia packages. All the packages come from the main language registry
and are readily available via Julia's package manager; registered packages have
to pass basic sanity checks and usually have a test suite.

The main questions of this empirical study are:
\begin{enumerate}
\item How uniformly are type stability and groundedness spread over Julia packages?
  How much of a difference can users expect from different packages?
\item Are package developers aware of type stability?
\item Are there any predictors of stability and groundedness in the code and do
  they influence how type-stable code looks?
\end{enumerate}

\subsection{Methodology}

We take as our main corpus the 1000 most starred (using GitHub stars) packages
from the Julia package registry; as of the beginning of 2021, the registry
contained about 5.5K packages. The main corpus is used for an automated,
high-level, aggregate analysis. We also take the 10 most starred packages from
the corpus to perform finer-grained analysis and manual inspection. Out of the
1000 packages in the corpus, tests suits of only \goodpkgsnum succeeded on
\juliaversion, so these \goodpkgsnum comprise our final corpus. The reasons of
failures are diverse, spanning from missing dependencies to the absence of
tests, to timeout.

For every package of interest, the dynamic analysis runs the package test suite,
analyzes compiled method instances, and records information relevant to
type stability and groundedness. Namely, once a test suite runs to completion,
we query Julia's virtual machine for the current set of available method
instances, which represent instances compiled during the tests' execution. To
avoid bias towards the standard library, we remove instances of methods defined in
standard modules, such as \c{Base}, \c{Core}, etc., \ADD{which typically leaves
us with several hundreds to several thousands of instances.}
For these remaining instances,
we analyze their type stability and groundedness. As type information is not
directly available for compiled, optimized code, we retrieve the original method
of an instance and run Julia's type inferencer to obtain each register's type,
similarly to the \jules model. In rare cases, type inference may fail;
on our corpus, this almost never happened, with at most 5 failures per package.
With the inference results at hand, we check the concreteness of the register
typing and record a yes/no answer for both stability and groundedness. In
addition to that, several metrics are recorded for each method: 
method size, the number of gotos and returns in the body, whether the method has
varargs or \c{@nospecialize} arguments\footnote{\texttt{@nospecialize} tells the
compiler to \emph{not} specialize for that argument and leave it abstract.}, and
how polymorphic the method is, i.e. how many instances were compiled for it.
This information is then used to identify possible correlations between the
metrics and stability/groundedness.

To get a better understanding of type stability and performance, we employ
several additional tools to analyze the 10 packages. For example, we look at
their documentation, especially at the stated goals and domain of a package, and
check the Git history to see if and how type stability is mentioned in the
commits.

\begin{table}[t]\small
\caption{Aggregate statistics for stability and groundedness}%
\label{empirical:fig:all}
\begin{tabular}{@{}lrr@{}}
\toprule
          & \multicolumn{1}{c}{Stable} & \multicolumn{1}{c}{Grounded} \\ \midrule
Mean      & 74\%                       & 57\%                         \\
Median    & 80\%                       & 57\%                         \\
Std. Dev. & 22\%                       & 24\%                         \\ \bottomrule
\end{tabular}
\end{table}

\subsection{Package-Level Analysis}

The aggregate results of the dynamic analysis for the \goodpkgsnum packages are
shown in Table~\ref{empirical:fig:all}: 74\% of method instances in a package
are stable and 57\% grounded, on average; median values are close to the means.
The standard deviation is noticeable, so even on small samples of packages, we
expect to see packages with large deflections from the means.

A more detailed analysis of the 10 most starred packages, in alphabetical order,
is shown in Table~\ref{empirical:fig:top}. A majority of these packages have
stability numbers very close to the averages shown above, with the exception of
Knet, which has only 16\% of stable and 8\% of grounded instances.

\begin{table}[h]\small
\caption{Stability, groundedness, and polymorphism in 10 popular packages}%
\label{empirical:fig:top}
\begin{tabular}{@{}lrrrrrr@{}}
\toprule
\multicolumn{1}{c}{Package} & \multicolumn{1}{c}{Methods} & \multicolumn{1}{c}{Instances} & \multicolumn{1}{l}{Inst/Meth} & \multicolumn{1}{c}{Varargs} & \multicolumn{1}{c}{Stable} & \multicolumn{1}{c}{Grounded} \\ \midrule
{\footnotesize DifferentialEquations}       & 1355                        & 7381                          & 5.4                           & 3\%                         & 70\%                       & 44\%                         \\
Flux                        & 381                         & 4288                          & 11.3                          & 13\%                        & 76\%                       & 70\%                         \\
Gadfly                      & 1100                        & 4717                          & 4.3                           & 10\%                         & 81\%                       & 58\%                         \\
Gen                         & 973                         & 2605                          & 2.7                           & 2\%                         & 64\%                       & 43\%                         \\
Genie                       & 532                         & 1401                          & 2.6                           & 12\%                        & 93\%                       & 78\%                         \\
IJulia                      & 39                          & 136                           & 3.5                           & 8\%                         & 84\%                       & 60\%                         \\
JuMP                        & 2377                        & 36406                         & 15.3                          & 7\%                        & 83\%                       & 63\%                         \\
Knet                        & 594                         & 9013                          & 15.2                          & 7\%                         & 16\%                       & 8\%                          \\
Plots                       & 1167                        & 5377                          & 4.6                           & 8\%                         & 74\%                       & 58\%                         \\
Pluto                       & 727                         & 2337                          & 3.2                           & 4\%                         & 80\%                       & 66\%                         \\ \bottomrule
\end{tabular}

\end{table}

The Knet package is a type stability outlier. A quick search over project's
documentation and history shows that the only kind of stability ever mentioned
is numerical stability; furthermore, the word ``performance'' is mostly used to
reference the performance of neural networks or CUDA-related utilities. Indeed,
the package primarily serves as a communication layer for a GPU; most
computations are done by calling the CUDA API for the purpose of building deep
neural networks. Thus, in this specific domain, type stability of Julia code
appears to be irrelevant.

On the other side of the stability spectrum is the 93\% stable (78\% grounded)
Genie package, which provides a web application framework. Inspecting the
package,
we can confirm that its developers were aware of type stability and
intentional about performance.
For example, Genie's Git history contains several commits mentioning (improved)
``type stability.''
The project README file states that the authors build upon
\begin{quote}
\it ``Julia's strengths (high-level, high-performance, dynamic, JIT compiled).''
\end{quote}
Furthermore, the tutorial claims:
\begin{quote}
\it ``Genie's goals: unparalleled developer productivity, excellent run-time
performance.''
\end{quote}

\paragraph{Type stability (non-)correlates}
One parameter that we conjectured may correlate with stability is the average
number of method instances per method (Inst/Meth column of
Table~\ref{empirical:fig:top}), as it expresses the amount of polymorphism
discovered in a package. Most of the packages compile just 2--4 instances per
method on average, but Flux, JuMP, and Knet have this metric 5--6 times higher,
with JuMP and Knet exploiting polymorphism the most, at 15.3 and 15.2
instances per method, respectively. The latter may be related to the very low type stability
index of Knet. However, the other two packages are more stable than the overall
average. Analyzing JuMP and Flux further, we order their methods by the number
of instances. In JuMP, the top 10\% of most instantiated methods are 5\% less
stable and grounded than the package average, whereas in Flux, the top 10\% have
about the same stability characteristics as on average. Overall, we cannot
conclude that method polymorphism is related to type stability.

Another dimension of polymorphism is the variable number of arguments in a
method (Varargs column of Table~\ref{empirical:fig:top}). We looked into three
packages with a higher than average (9\%) number of varargs methods in the 10
packages: Flux, Gadfly and Genie. Relative to the total number of methods, Flux has the most
varargs methods---13\%---and those methods are only 55\% stable and 44\%
grounded, which is a significant drop of 21\% and 26\% below this package's
averages. However, the other two packages have higher-than-package-average
stability rates, 82\% (Gadfly) and 99\% (Genie), with groundedness being high in
Genie, 93\%, and low in Gadfly, 38\%. Again, no general conclusion about the
relationship between varargs methods and their stability can be made.

\subsection{Method-Level Analysis}

\MODIFY{
In this section, we inspect stability of individual methods in its possible
relationship with other code properties like size, control flow (number of goto
and return statements), and polymorphism (number of compiled instances and varargs).
Our analysis consists of two steps: first, we plot histograms
showing the number of methods with particular values of properties, and second,
we manually sample some of the methods with more extreme characteristics.
}

\subsubsection{Graphical Analysis}
\label{sssect:graphs}

We use two-dimensional histograms like those presented in
Fig.~\ref{figs:size:Pluto:main} to discover possible relationships between stability
of code and its other properties. The vertical axis measures stability (on the
left diagram) or groundedness (on the right): $1$ means that all recorded
instances of a method are stable/grounded, and $0$ means that none of them are.
The horizontal axis measures the property of interest; in the case of
Fig.~\ref{figs:size:Pluto:main}, it is method size (actual numbers are not
important here: they are computed from Julia's internal representation of source
code). The color of an area reflects how many methods have characteristics
corresponding to that area's position on the diagram; e.g.\ in
Fig.~\ref{figs:size:Pluto:main}, the lonely yellow areas indicate that there are about
500 (400) small methods that are stable (grounded).

\begin{figure}[h]
     \begin{subfigure}[b]{0.35\textwidth}
       \includegraphics[width=\textwidth]{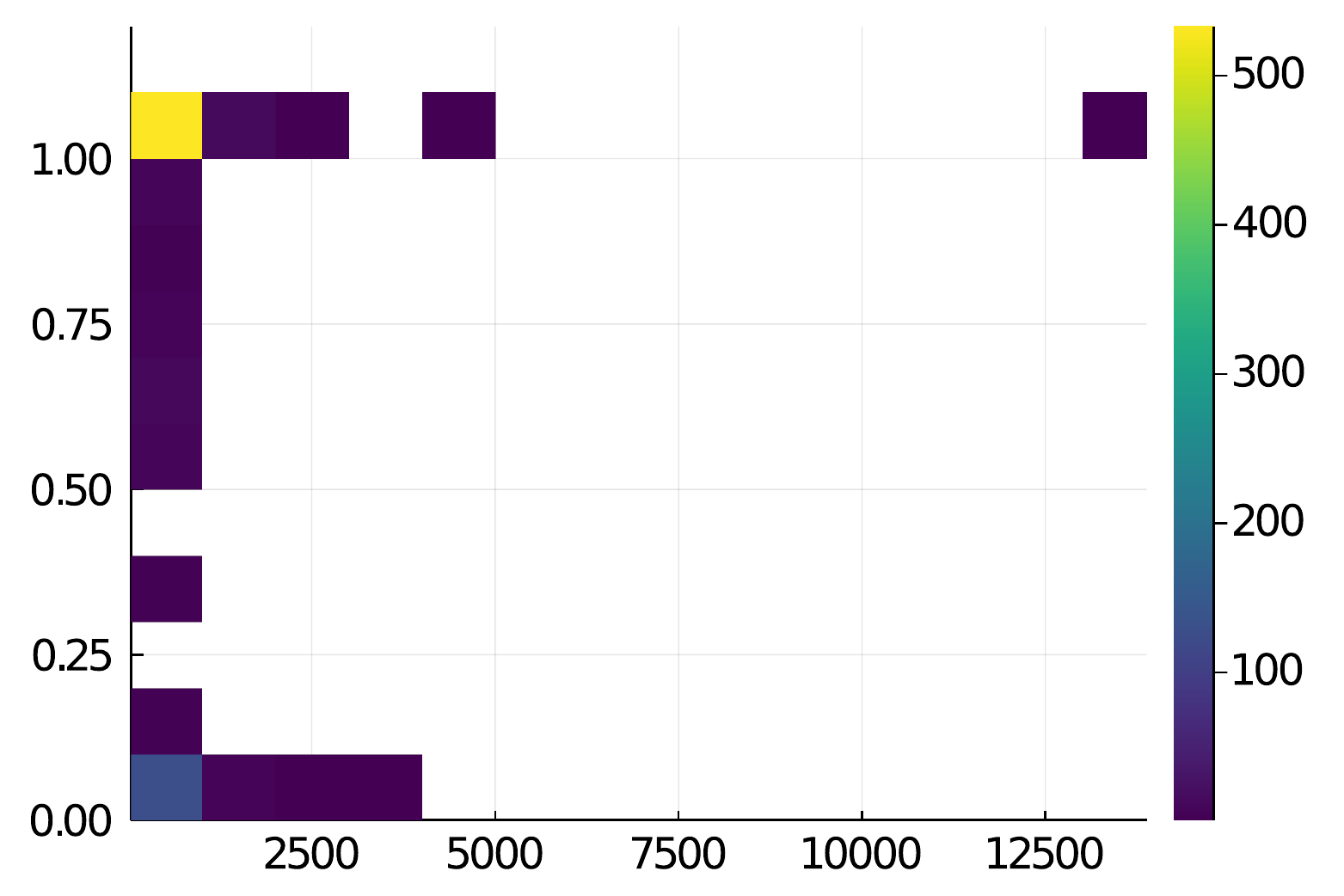}
     \end{subfigure}
     \ \ \
     \begin{subfigure}[b]{0.35\textwidth}
       \includegraphics[width=\textwidth]{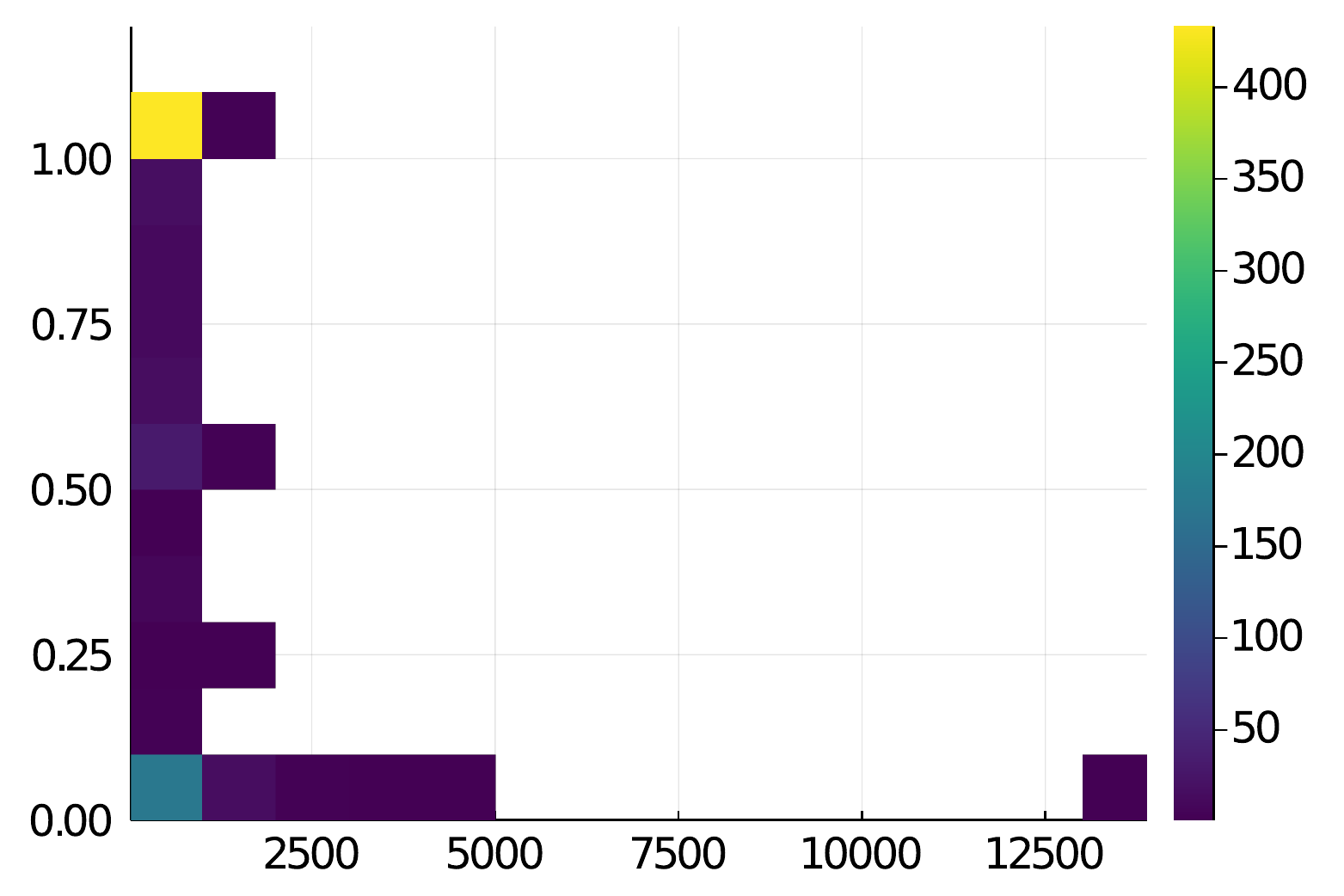}
     \end{subfigure}
\caption{Stability (left, OY axis) and groundedness (right, OY) by method size (OX) in Pluto}%
\Description{Stability and groundedness by method size in Pluto}%
\label{figs:size:Pluto:main}
\end{figure}

We generate graphs for all of the 10 packages listed in
Table~\ref{empirical:fig:top}, for all combinations of the properties of
interest\PAPERVERSIONINLINE{; the graphs are available in the extended
version~\cite{oopsla21jules:arx}}{; the graphs are provided in \appref{sec:app}}.
Most of the graphs look very similar to the ones from
Fig.~\ref{figs:size:Pluto:main}, which depicts Pluto---a package for creating
Jupyter-\hspace{0pt}like notebooks in Julia. In the following paragraphs, we
discuss features of these graphs and highlight the discrepancies.

The first distinctive feature of the graphs is the hot area in the top-left
corner: most of the 10 packages employ many small, stable/grounded methods;
the bottom-left corner is usually the second-most hot, so a significant number
of small methods are unstable/ungrounded. For the Knet package,
these two corners are reversed; for DifferentialEquations, they are reversed
only on the groundedness plot. Both of these facts are not surprising after seeing
Table~\ref{empirical:fig:top}, but having a visual tool to discover such facts
may be useful for package developers.

The second distinctive feature of these graphs is the behavior of large,
ungrounded methods (bottom-right part of the right-hand-side graph). The
``tail'' of large methods on the groundedness graphs almost always lies below
the $1$-level; furthermore, larger methods tend to be less grounded.
However, if we switch from groundedness to stability plots, a large portion of
the tail jumps to the $1$-level. This means larger methods are unlikely to be
grounded (as expected, because of the growing number of registers), but they
still can be stable and thus efficiently used by other methods. Pluto provides a
good example of such a method: its \c{explore!} method of size 13003 (right-most
rectangle on Fig.~\ref{figs:size:Pluto:main}, 330 lines in the source code) analyzes
Julia syntax trees for scope information with a massive \c{if/else if/..} statement.
This method has a very low chance of being grounded, and it was not grounded on the
runs we analyzed. However, the method has a concrete return type annotation, so
Julia (and the programmer) can easily see that it is stable.


In the case of the number of gotos and returns, the plots are largely similar
to the ones for method size, but they highlight one more package with low
groundedness. Namely, the Gen package (aimed at probabilistic
inference~\cite{JuliaGenPkgPub2019}) has the hottest area in the bottom-left
corner, contrary to the first general property we identified for the size-based
plots. Recall (Tables \ref{empirical:fig:all} and \ref{empirical:fig:top}) that
Gen's groundedness is 14\% less than the average on the whole corpus of
\goodpkgsnum packages.

\subsubsection{Manual Inspection}

\ADD{
To better understand the space of stable methods,
we performed a qualitative analysis of a sample of stable methods that
have either large sizes or many instances.

Many large methods have one common
feature: they often have a return type ascription on the method header of
the form:
}
\begin{lstlisting}[language=julia]
function f(...) :: Int
  ...
end
\end{lstlisting}
\ADD{
These ascriptions are a relatively new tool in Julia, and they are used only
occasionally, in our experience. An ascription makes the Julia compiler insert implicit
conversions on all return paths of the method. Conversions are user extendable:
if the user defines type \c{A}, they can also add methods of a special
function \c{convert} for conversion to \c{A}. This function will be called
when the compiler expects \c{A} but infers a different type, for example,
if the method returns \c{B}.
If the method returns \c{A}, however, then \c{convert} is a no-op.

Type ascriptions may be added simply as documentation, but they can also
be used to turn type instability into a run-time error:
if the ascribed type is concrete and a necessary conversion is not available,
the method will fail. This provides a useful, if
unexpected, way to assure that a large method never becomes unstable.

While about 85\% of type-stable methods in the top 10 packages are uninteresting
in that they always return the same type, sampling the rest illuminates
a clear pattern: the methods resemble what we
are used to see in statically typed languages with parametric polymorphism. Below
is a list of categories that we identify in this family.
}

\begin{itemize}

\lstset{basicstyle=\footnotesize\tt,linewidth=.93\textwidth}

\ADD{
\item
  Various forms of the identity function---a surprisingly popular function that
  packages keep reinventing. In an impure language, such as Julia, an identity
  function can produce various side effects.
  For example, the Genie package adds a caching effect to
  its variant of the identity function:
}
\begin{lstlisting}[language=julia]
# Define the secret token used in the app for encryption and salting.
function secret_token!(value::AbstractString=Generator.secret_token())
  SECRET_TOKEN[] = value
  return value
end
\end{lstlisting}

\ADD{
\item
  Container manipulations for various kinds of containers, such as arrays, trees, or
  tuples. For instance, the latter is exemplified by the following function
  from Flux, which maps a tuple of functions by applying them to the given argument:
}
\begin{lstlisting}[language=julia]
function extraChain(fs::Tuple, x)
  res = first(fs)(x)
  return (res, extraChain(Base.tail(fs), res)...)
end
extraChain(::Tuple{}, x) = ()
\end{lstlisting}

\ADD{
\item
  Smart constructors for user-defined polymorphic structures. For example, the following
  convenience function from JuMP creates an instance of the
  \c{VectorConstraint} structure with three fields, each of which is polymorphic:
}
\begin{lstlisting}[language=julia]
function build_constraint(_error::Function, Q::Symmetric{V,M}, ::PSDCone)
        where {V<:AbstractJuMPScalar,M<:AbstractMatrix{V}}
    n = LinearAlgebra.checksquare(Q)
    shape = SymmetricMatrixShape(n)
    return VectorConstraint(
        vectorize(Q, shape),
        MOI.PositiveSemidefiniteConeTriangle(n),
        shape)
end
\end{lstlisting}

\ADD{
\item
  Type computations---an unusually wide category for a dynamically typed
  language. Thus, for instance, the Gen package defines a type that represents generative
  functions in probabilistic programming, and a function that extracts the
  return and argument types: 
}
\begin{lstlisting}[language=julia]
# Abstract type for a generative function with return value type T and trace type U.
abstract type GenerativeFunction{T,U <: Trace} end
get_return_type(::GenerativeFunction{T,U}) where {T,U} = T
get_trace_type(::GenerativeFunction{T,U}) where {T,U} = U
\end{lstlisting}

\lstset{linewidth=\textwidth}

\end{itemize}

\subsection{Takeaways}

Our analysis shows that a Julia user can expect mostly stable (74\%) and
somewhat grounded (57\%) code in widely used Julia packages. If the authors
are intentional about performance and stability, as demonstrated by the Genie
package, those numbers can be much higher. Although our sample of packages is
too small to draw strong conclusions, we suggest that several factors can be
used by a Julia programmer to pinpoint potential sources of instability in their
package. For example, in some cases, varargs methods might indicate instability.
Large methods, especially ones with heavy control flow,
tend to not be type grounded but often are stable; in particular,
if they always return the same concrete type.
Finally, although highly polymorphic methods are neither stable nor unstable
in general, code written in the style of parametric polymorphism 
often suggests type stability.

Our dynamic analysis and visualization code is written in Julia (and some bash
code), and relies on the vanilla Julia implementation. Thus, it can be employed
by package developers to study type instability in their code, as well as check
for regressions.

\section{Related Work}

Type stability and groundedness are the consequences of Julia's compilation
strategy put into practice. The approach Julia takes is new and simpler than
other approaches to efficient compilation of dynamic code.

Attempts to efficiently compile dynamically dispatched languages go back nearly
as far as dynamically dispatched languages themselves.
\citet{HurricaneSmalltalk} used a combination of run-time-checked user-provided
types and a simple type inference algorithm to inline methods. \citet{CU89}
pioneered the just-in-time model of compilation in which methods are specialized
based on run-time information. \citet{holzle1994odd} followed up with method
inlining optimization based on recently observed types at the call site.
\citet{Psyco2004} specialized methods on invocation based on their arguments,
but this was limited to integers. Similarly, \citet{cannon2005localized} developed a
type-inferring just-in-time compiler for Python, but it was limited by the
precision of type inference. \citet{RATA} extended this approach with a more
sophisticated abstract interpretation-based inference system for JavaScript.

At the same time, trace-based compilers approached the problem from another
angle~\cite{chang2007efficient}.
Instead of inferring from method calls, these compilers had exact type
information for variables in straight-line fragments of the program called
traces. \citet{gal09} describes a trace-based compiler for JavaScript that avoids
some pitfalls of type stability, as traces can cross method boundaries.
However, it is more difficult to fix a program when tracing does not work well,
for the boundaries of traces are not apparent to the programmer.

\ADD{
Few of these approaches to compilation have been formalized.
\citet{CompilingWithTraces} described the core of a trace-based compiler with
two optimizations, variable folding and dead branch/store elimination.
\citet{VerifiedJITx86} formalized self-modifying code for x86. Finally,
\citet{popl18} formally described speculation and deoptimization and proved
correctness of some optimizations; \citet{oopsla21} extended and mechanized
these results.

The Julia compiler uses standard techniques, but differs considerably in how it
applies them. Many production just-in-time compilers rely on static type
information when it is available, as well as a combination of profiling and
speculation~\cite{TruffleIR,TruffleInterpreters}. Speculation allows these
compilers to perform virtual dispatch more efficiently~\cite{oopsla20c}. Profiling allows
for tuning optimizations to a specific workload~\cite{GoWithTheFlow,HHVMJIT},
eliminating overheads not required for cases observed during execution. Julia, on
the other hand, performs optimization only once per method instance.
This presents both advantages
and issues. For one, Julia's performance is more predictable than that of other
compilers, as the warmup is simple~\cite{VMsBlow}. Overall, Julia is able
to achieve high performance with a simple compiler.

}

\section{Conclusion}

\MODIFY{
Julia's performance is impressive, as it improves on that of a statically typed
language such as Java, whose compiler has been developed for two decades by an
army of engineers. This is achieved by a combination of careful language design,
simple but powerful compiler optimizations, and disciplined use of abstractions
by programmers.

In this paper, we formally define the notions of type stability and groundedness,
which guide programmers towards code idioms the compiler can optimize. To this
end, we model Julia's intermediate representation with an abstract machine
called \jules, prove the correctness of its JIT compiler,
and show that type groundedness is the property that enables the
compiler to devirtualize method calls. The weaker notion of type stability is
still useful as it allows callers of a function to be grounded. This
relationship between groundedness and stability explains the discrepancy between
the definition of stability and what programmers do in practice, namely,
inspecting
the compiler's output to look for variables that have abstract types. Our corpus
analysis of Julia packages shows that more than half of compiled methods are
grounded, and over two-thirds are stable. This suggests that developers follow
type-related performance guidelines.
}

\ADD{
Although our analysis suggests high presence of type-stable code,
some Julia packages, even among more popular ones, have significant unstable
portions. This may indicate an oversight from the package authors, albeit an
understandable one: the tools for detecting unstable code are
limited.}\!\footnote{Besides manually calling \texttt{@code\_warntype}, one can use
\href{https://github.com/aviatesk/JET.jl}{JET.jl} package in
``performance linting'' mode. JET relies on Julia's type inference;
its primary goal is to report likely dynamic-dispatch errors,
and it does not provide sound analysis.}
\ADD{There is
a future work opportunity in developing a user-facing static type system that
would soundly determine whether a method is type stable or not. This could
provide the benefits of a traditional (gradual) type system, and additionally,
allow programmers to ensure type stability. A less demanding approach to
facilitating stability would be to add primitives for reifying the programmer's
intent. For example, the user could write \c{assert(is_type_stable())} to
indicate that the method does not tolerate inefficient compilation.
}

Although the specific optimizations enabled by grounded code may not be as
important for other languages, they can benefit from the lessons learned from
Julia's performance model: namely, that the interface exposed by a compiler
can be just as important as the cleverness of the compiler itself. Performance
is not an isolated property: it is the result of a dialogue between the compiler
and the programmers that use it.

\section*{Data Availability Statement}
The paper is accompanied by the artifact~\cite{artifact} reproducing the results
of~\secref{sec:empirical}. In particular, the artifact contains the list of
1000 Julia packages analyzed (with the exact package versions), as well as Bash and
Julia scripts performing stability analysis and generating
Tables~\ref{empirical:fig:all} and \ref{empirical:fig:top}, and
Fig.~\ref{figs:size:Pluto:main}.

\section*{Acknowledgments}

We thank Ming-Ho Yee and the anonymous reviewers for their insightful comments
and suggestions to improve this paper.

This work was supported by Office of Naval Research (ONR) award 503353, the
National Science Foundation awards 1759736, 1925644, 1618732, CCF-1909143 and
CCF-1908389 the Czech Ministry of Education from the Czech Operational Programme
Research, Development, and Education, under grant agreement No.
CZ.02.1.01/0.0/0.0/15\_003/0000421, and the European Research Council under the
European Union’s Horizon 2020 research and innovation programme, under grant
agreement No. 695412.

\bibliographystyle{ACM-Reference-Format}{}
\bibliography{bib/jv,bib/all,bib/lj}

\PAPERVERSION{}{
\appendix

\section{Graphs for \secref{sec:empirical}}\label{sec:app}

This appendix contains graphs similar to the ones described in
\secref{sssect:graphs} for all 10 packages discussed in \secref{sec:empirical}.
There are 6 graphs per package: the top two show the relationship between the
method size and stability (left) or groundedness (right); the other four graphs
connect the two type-related properties with control-flow features: the number
of gotos and the number of returns in a method instance.

Note that the bottom four graphs for every package are different from the top
two in that they group method instances, not methods. Therefore,
the bottom four graphs have all data bins either at level $OY=0$ or $1$, because
we always know whether a method instance is stable (grounded) or not.
The change comes from the fact that the control-flow features in question
depend on compiled code and the way it was optimized: e.\;g., an \c{if true} in a
method can be optimized away during compilation.

\clearpage

\subsection{Package: DifferentialEquations}
\begin{figure}[h]
     \begin{subfigure}[b]{0.49\textwidth}
       \includegraphics[width=\textwidth]{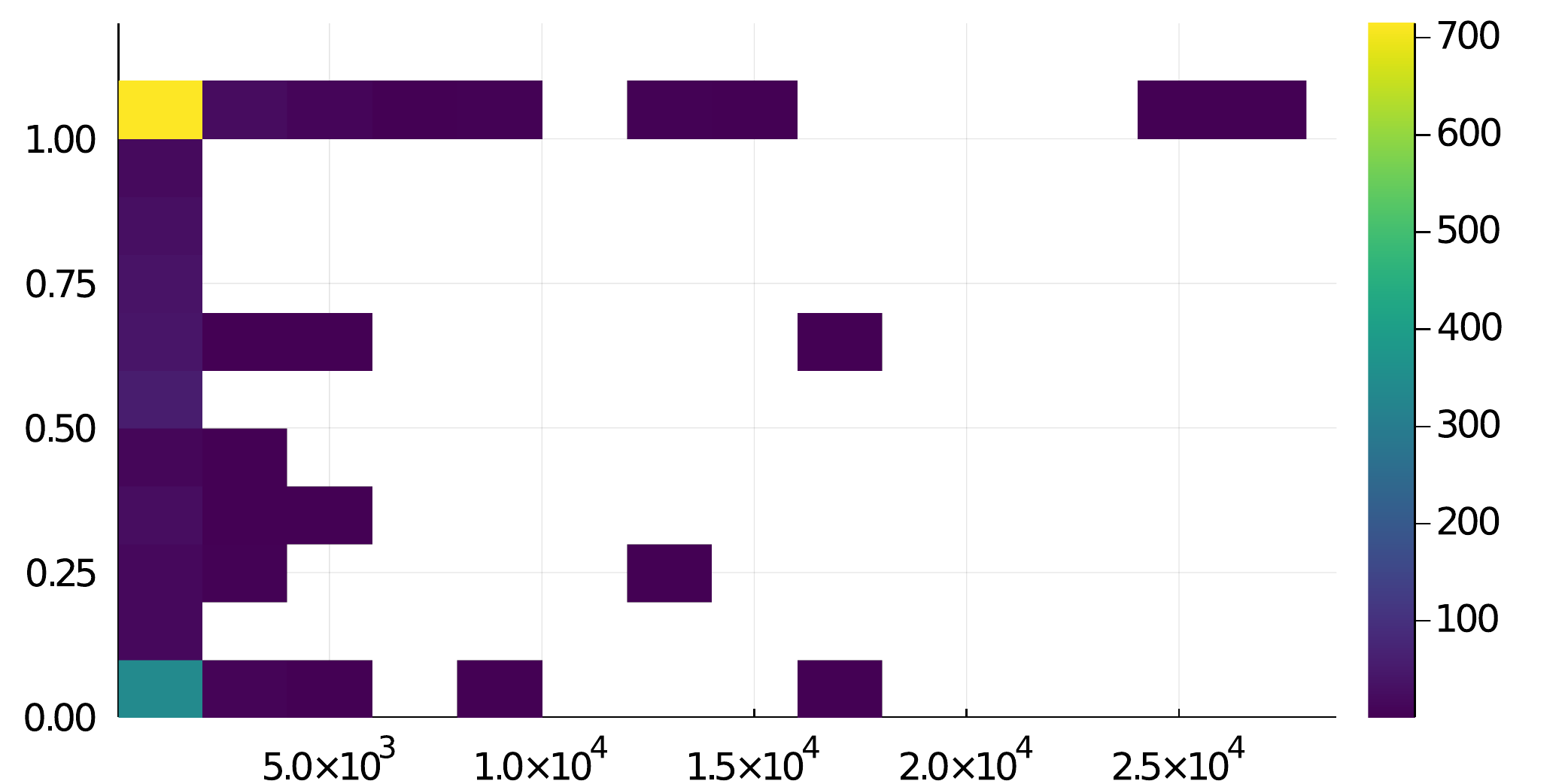}
     \end{subfigure}
     \ \
     \begin{subfigure}[b]{0.49\textwidth}
       \includegraphics[width=\textwidth]{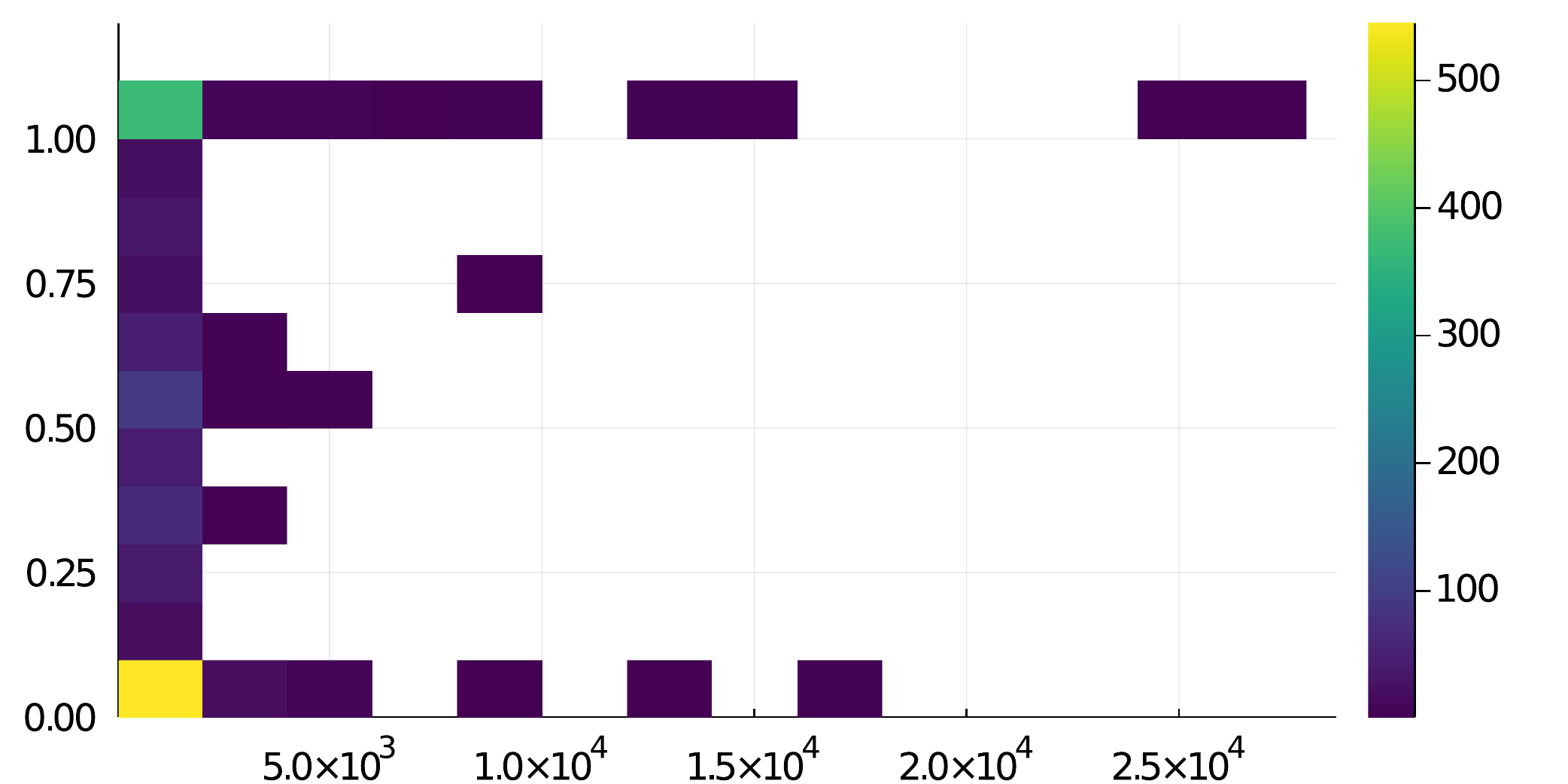}
     \end{subfigure}
\caption{Stability (left, OY axis) and groundedness (right, OY) by method size (OX)}%
\Description{Stability and groundedness by method size in DifferentialEquations}%
\label{figs:size:DifferentialEquations}
\end{figure}

\begin{figure}[h]
     \begin{subfigure}[b]{0.49\textwidth}
       \includegraphics[width=\textwidth]{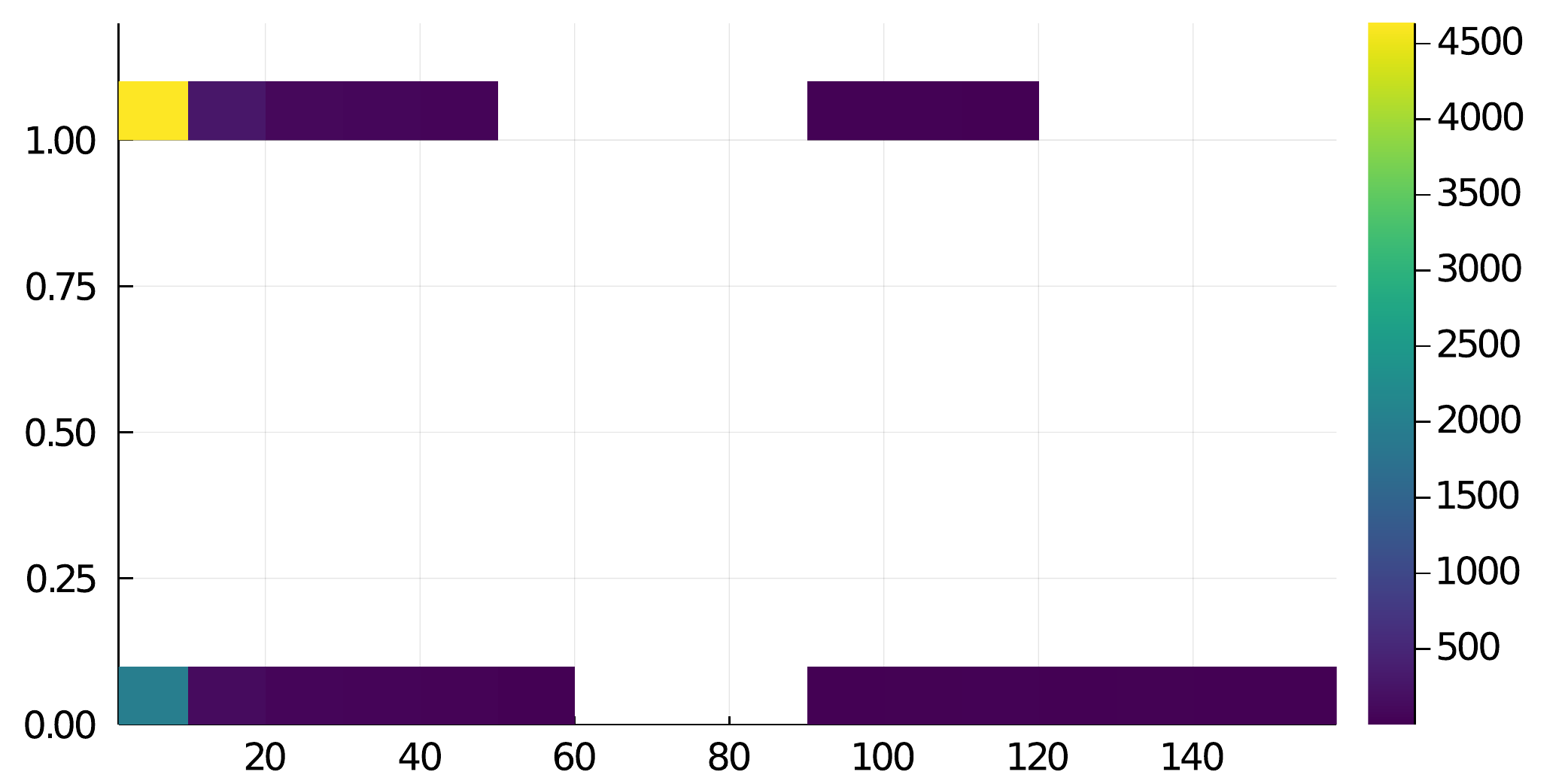}
     \end{subfigure}
     \ \
     \begin{subfigure}[b]{0.49\textwidth}
       \includegraphics[width=\textwidth]{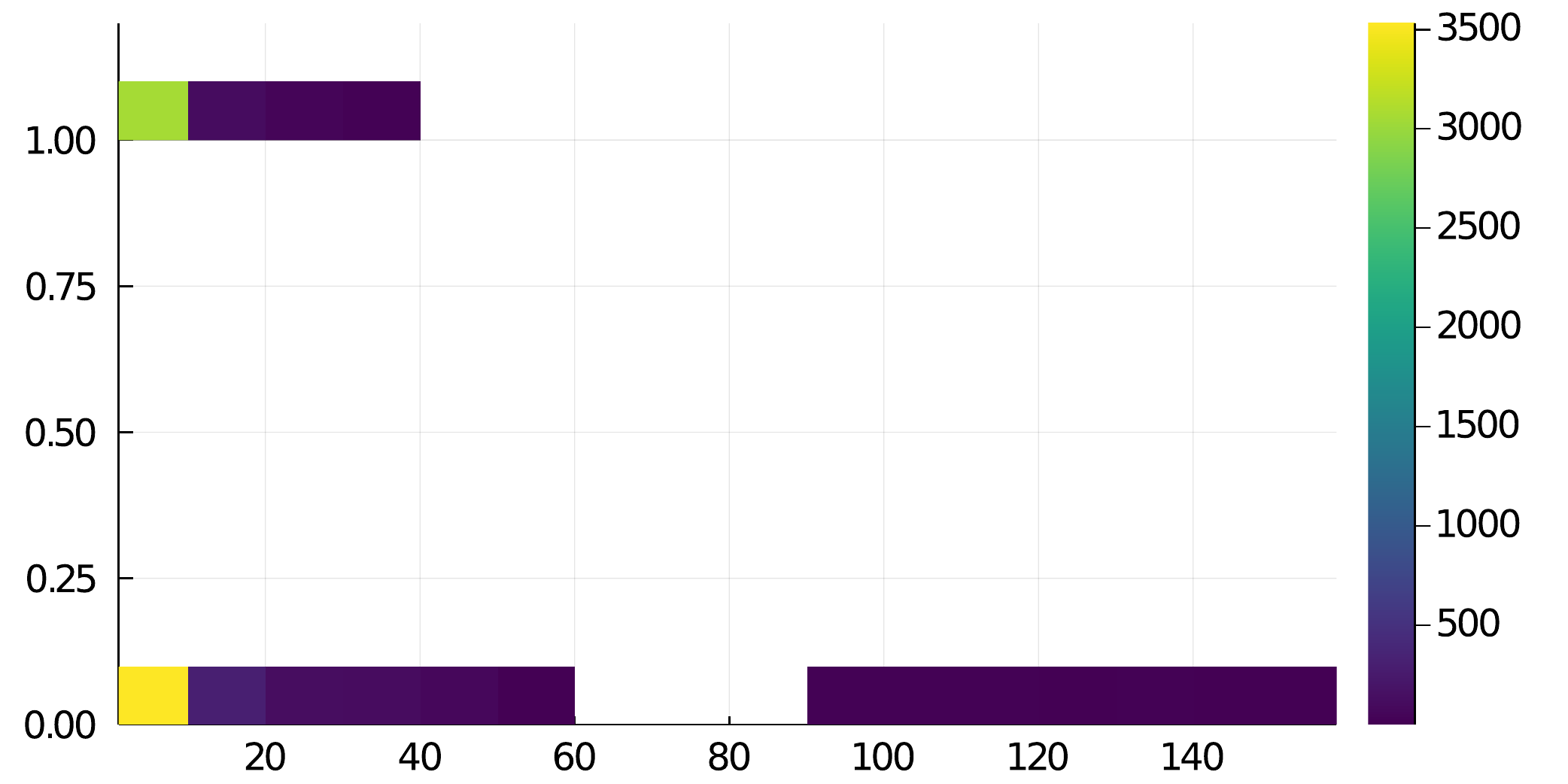}
     \end{subfigure}
\caption{Stability (left, OY axis) and groundedness (right, OY) by number of gotos in method instances (OX)}%
\Description{Stability and groundedness by number of goto's in method instances in DifferentialEquations}%
\label{figs:gotos:DifferentialEquations}
\end{figure}

\begin{figure}[h]
     \begin{subfigure}[b]{0.49\textwidth}
       \includegraphics[width=\textwidth]{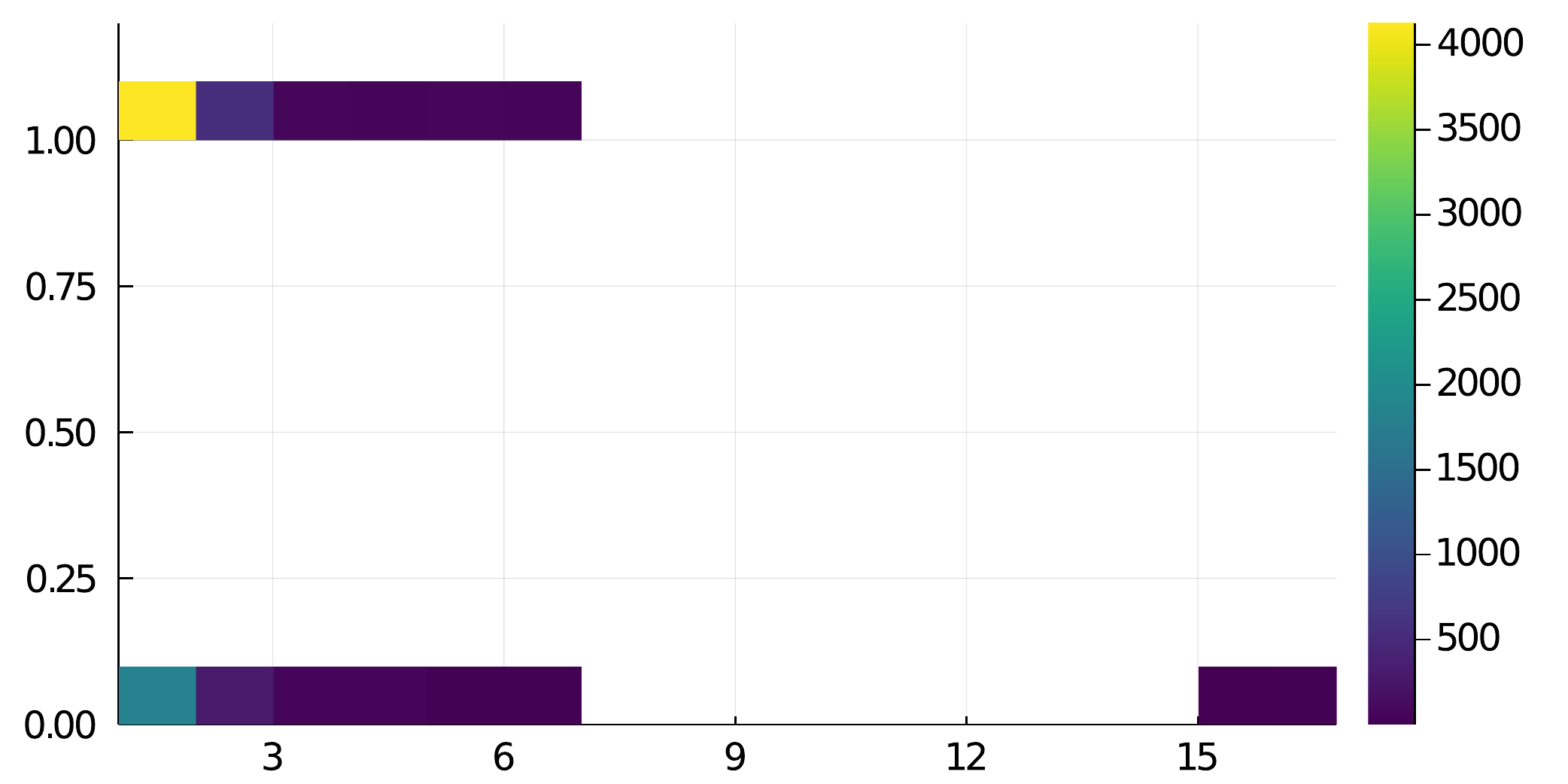}
     \end{subfigure}
     \ \
     \begin{subfigure}[b]{0.49\textwidth}
       \includegraphics[width=\textwidth]{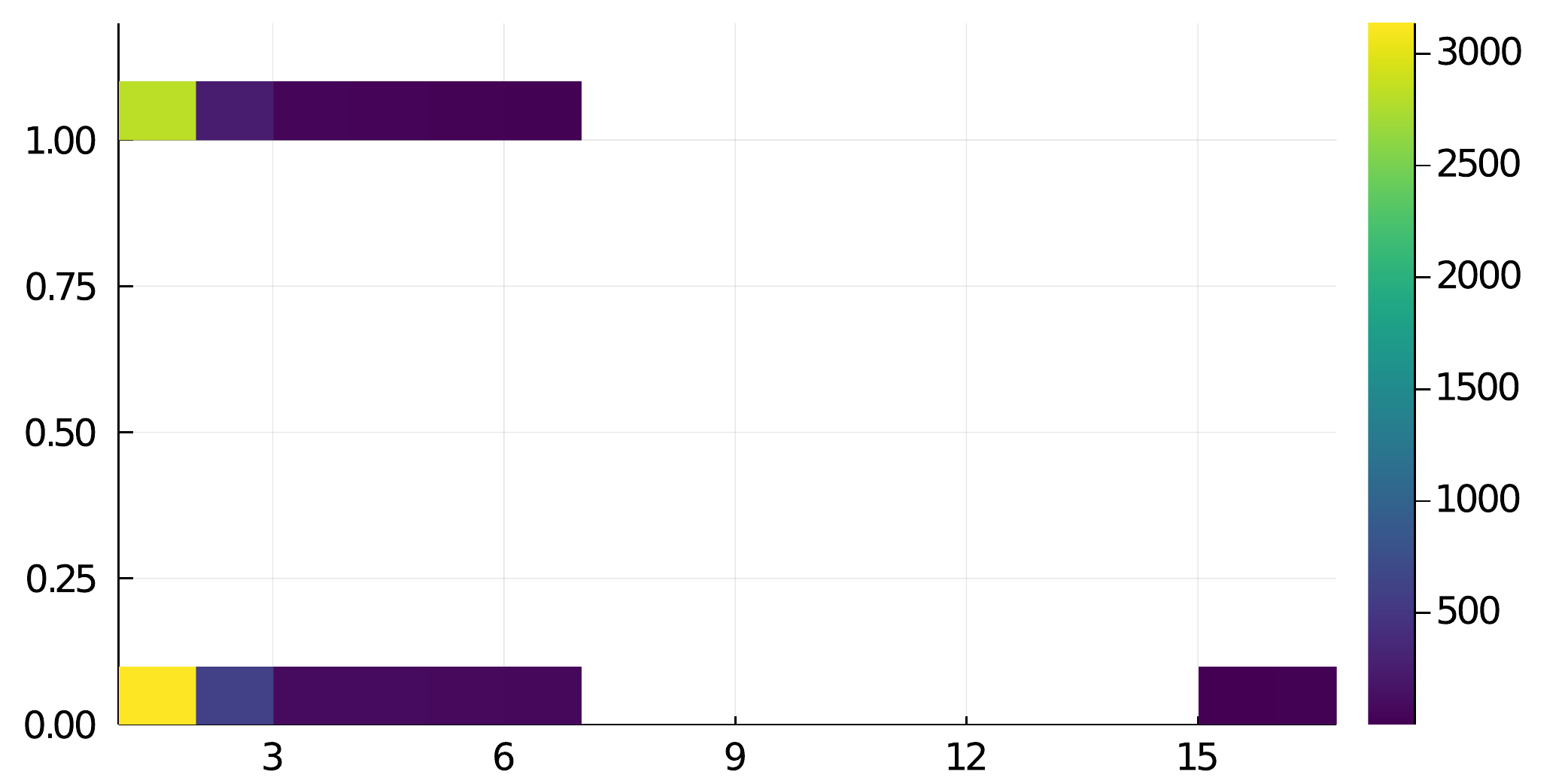}
     \end{subfigure}
\caption{Stability (left, OY axis) and groundedness (right, OY) by number of returns in method instances (OX)}%
\Description{Stability and groundedness by number of returns in method instances in DifferentialEquations}%
\label{figs:returns:DifferentialEquations}
\end{figure}
\clearpage
\subsection{Package: Flux}
\begin{figure}[h]
     \begin{subfigure}[b]{0.49\textwidth}
       \includegraphics[width=\textwidth]{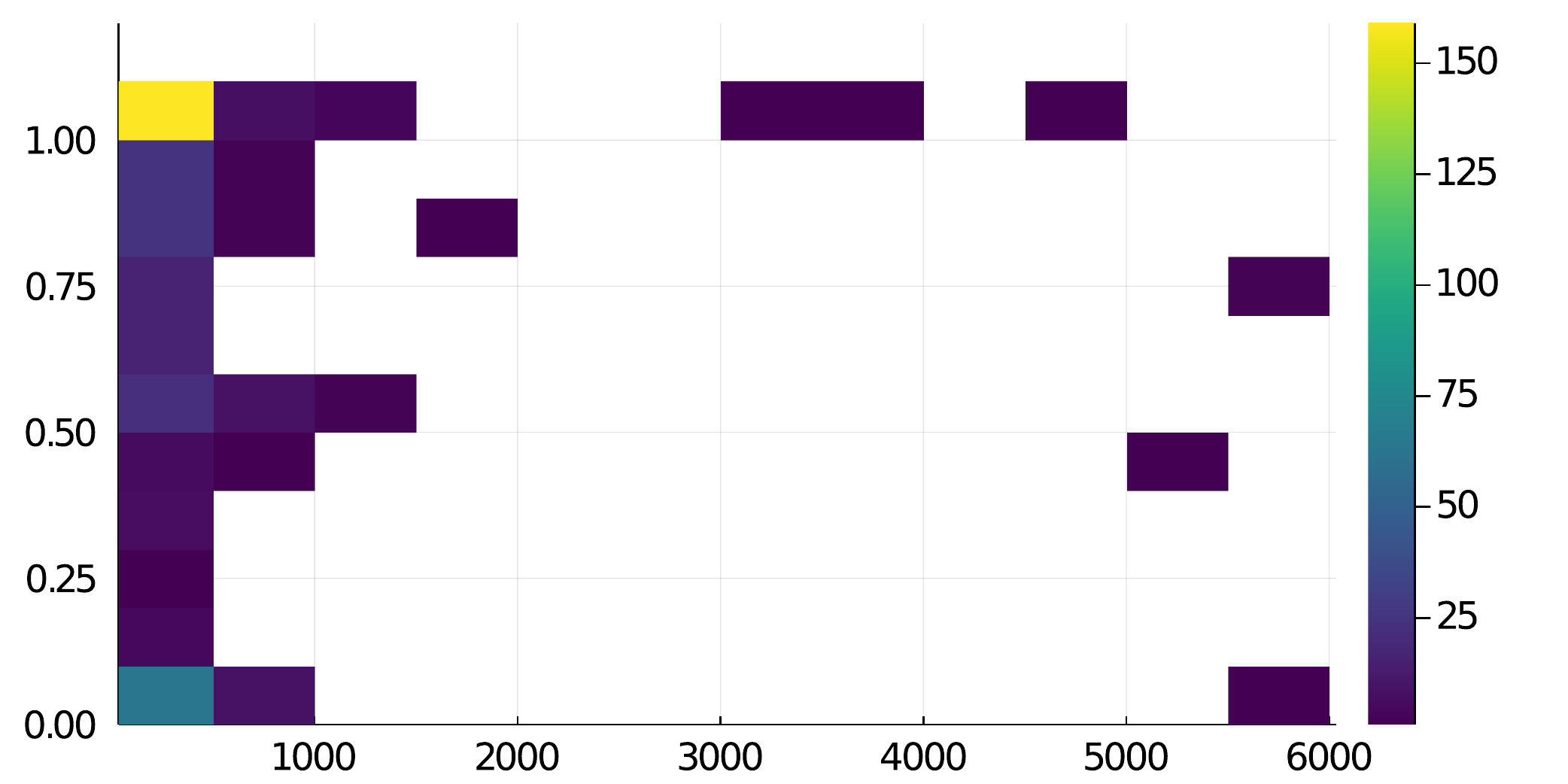}
     \end{subfigure}
     \ \
     \begin{subfigure}[b]{0.49\textwidth}
       \includegraphics[width=\textwidth]{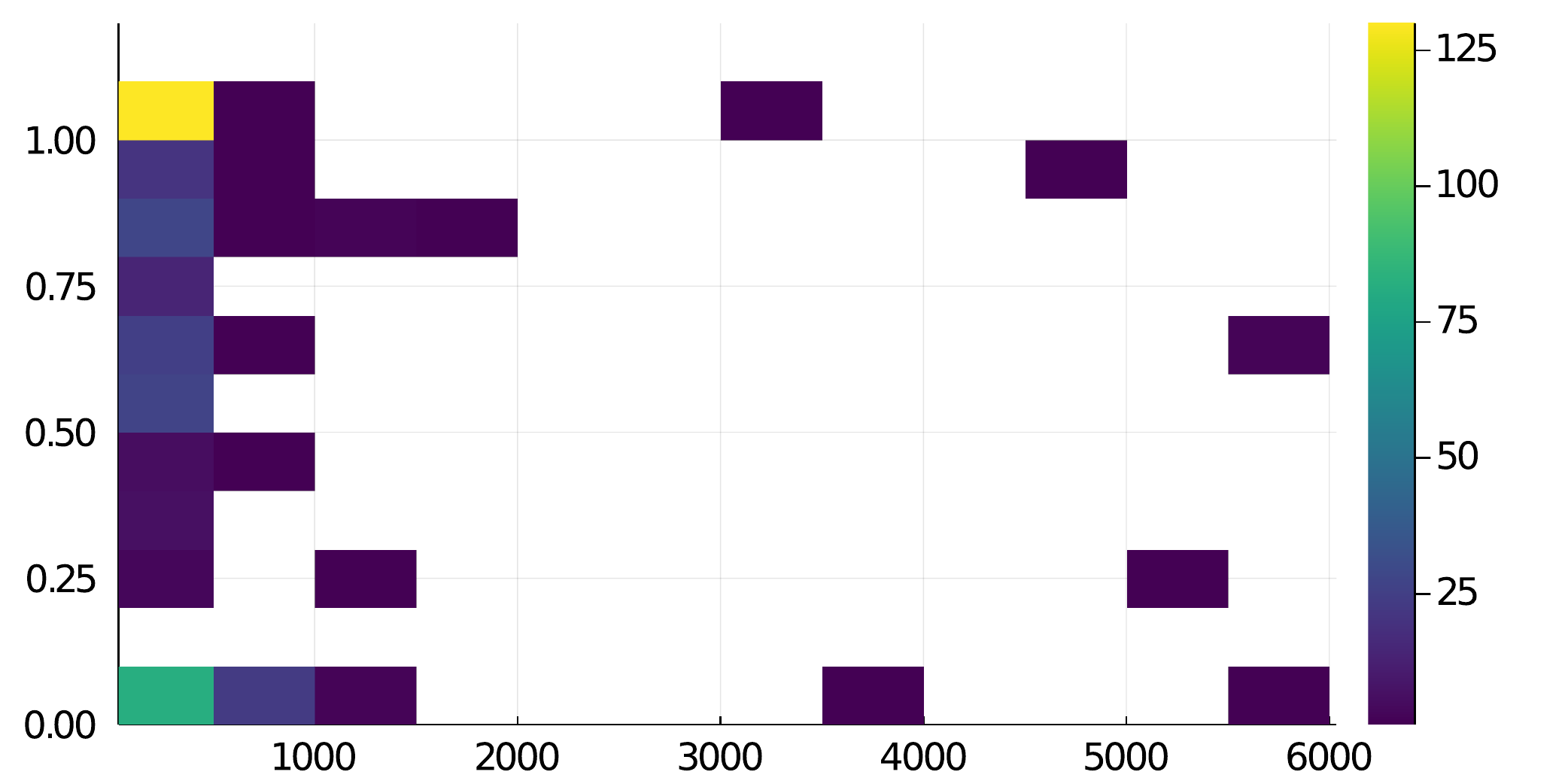}
     \end{subfigure}
\caption{Stability (left, OY axis) and groundedness (right, OY) by method size (OX)}%
\Description{Stability and groundedness by method size in Flux}%
\label{figs:size:Flux}
\end{figure}

\begin{figure}[h]
     \begin{subfigure}[b]{0.49\textwidth}
       \includegraphics[width=\textwidth]{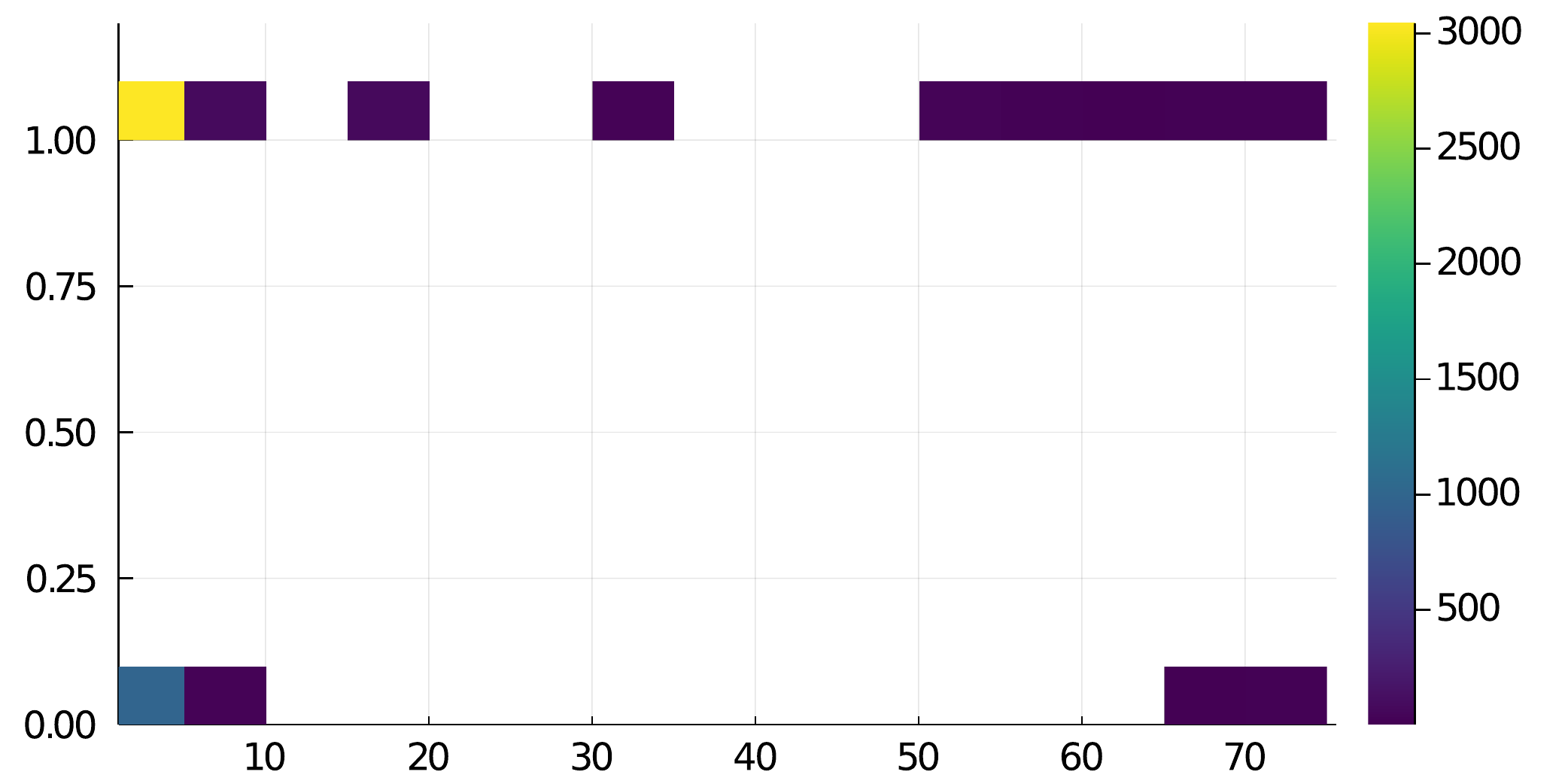}
     \end{subfigure}
     \ \
     \begin{subfigure}[b]{0.49\textwidth}
       \includegraphics[width=\textwidth]{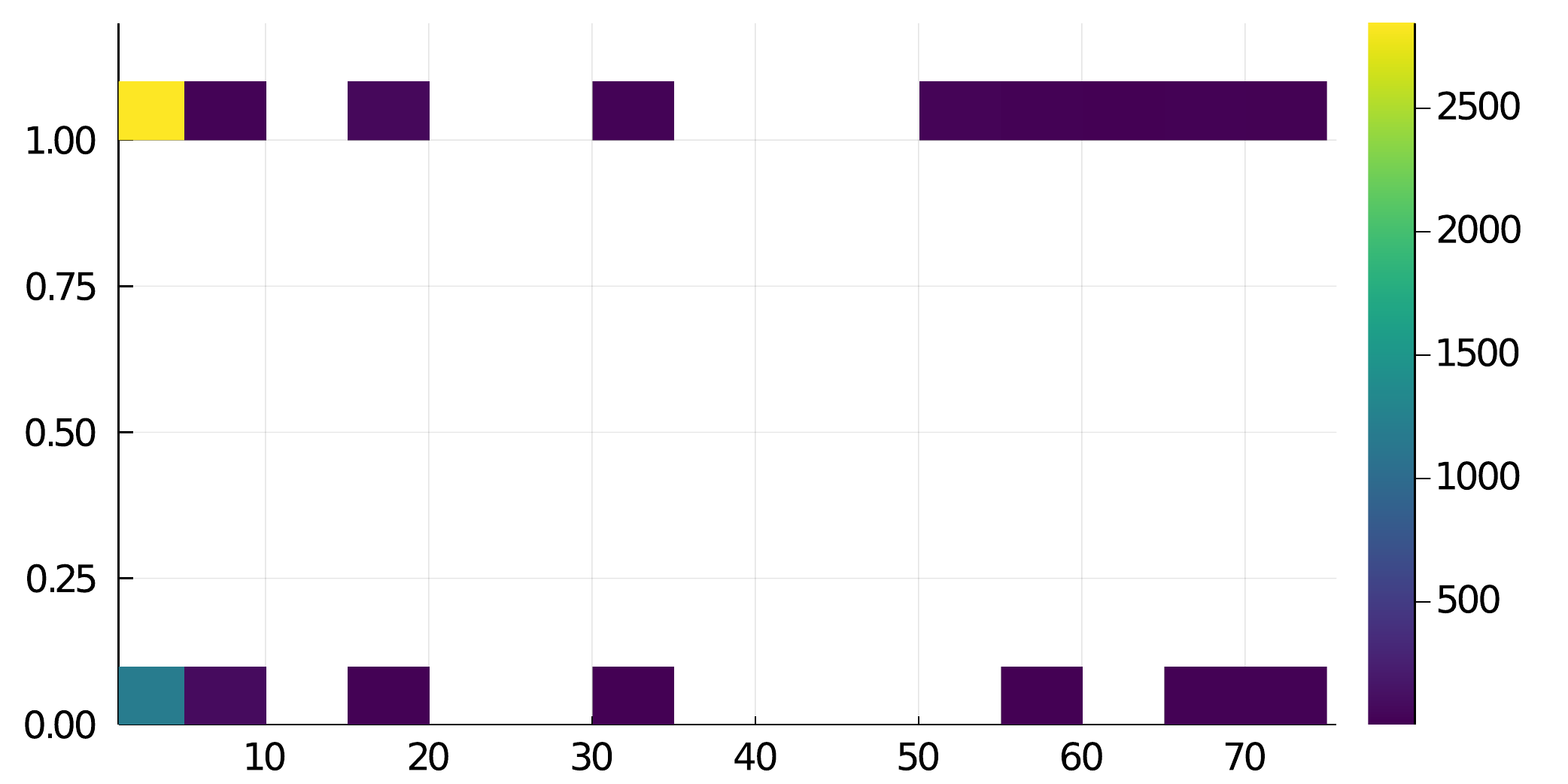}
     \end{subfigure}
\caption{Stability (left, OY axis) and groundedness (right, OY) by number of gotos in method instances (OX)}%
\Description{Stability and groundedness by number of goto's in method instances in Flux}%
\label{figs:gotos:Flux}
\end{figure}

\begin{figure}[h]
     \begin{subfigure}[b]{0.49\textwidth}
       \includegraphics[width=\textwidth]{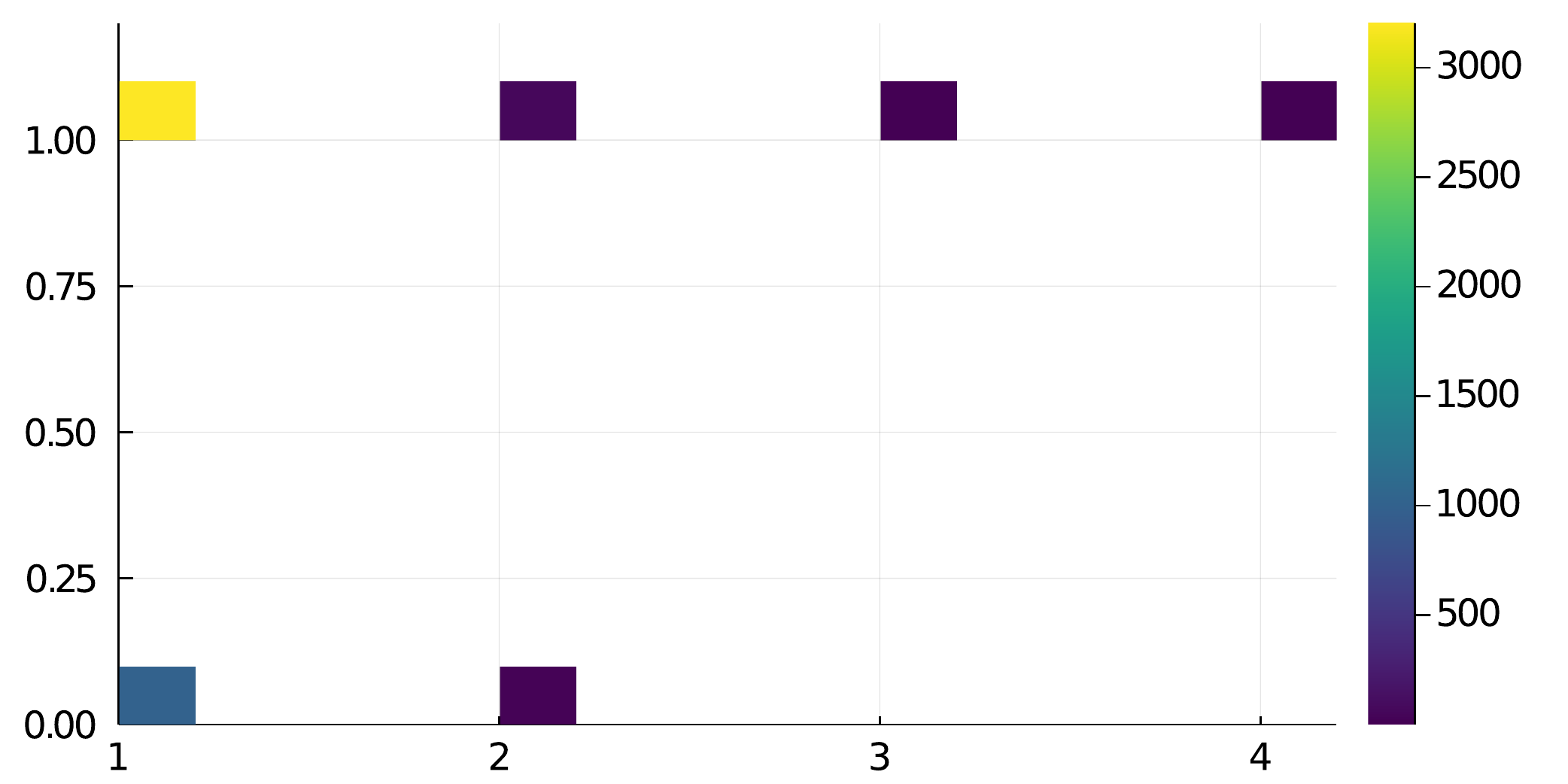}
     \end{subfigure}
     \ \
     \begin{subfigure}[b]{0.49\textwidth}
       \includegraphics[width=\textwidth]{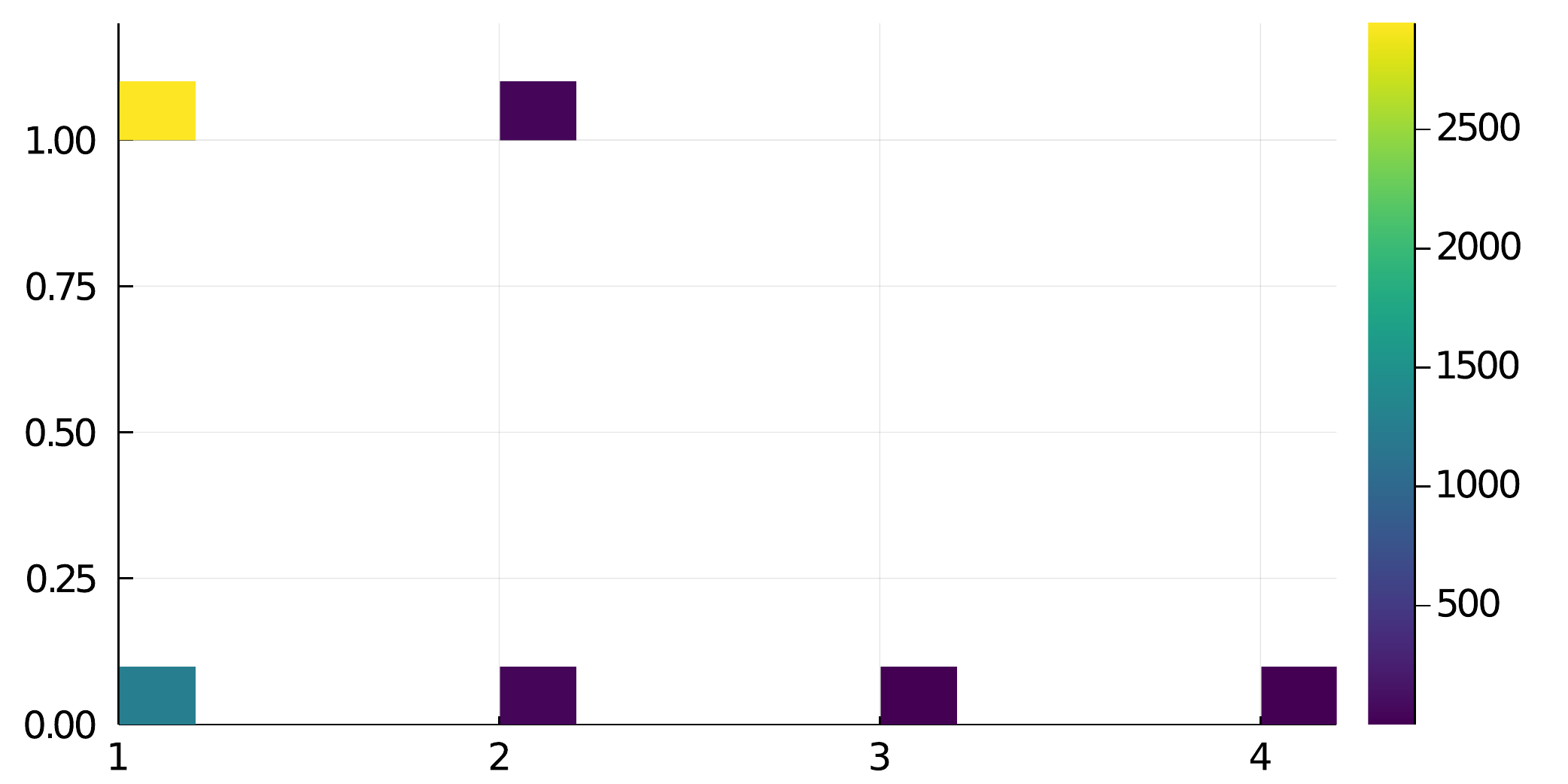}
     \end{subfigure}
\caption{Stability (left, OY axis) and groundedness (right, OY) by number of returns in method instances (OX)}%
\Description{Stability and groundedness by number of returns in method instances in Flux}%
\label{figs:returns:Flux}
\end{figure}
\clearpage
\subsection{Package: Gadfly}
\begin{figure}[h]
     \begin{subfigure}[b]{0.49\textwidth}
       \includegraphics[width=\textwidth]{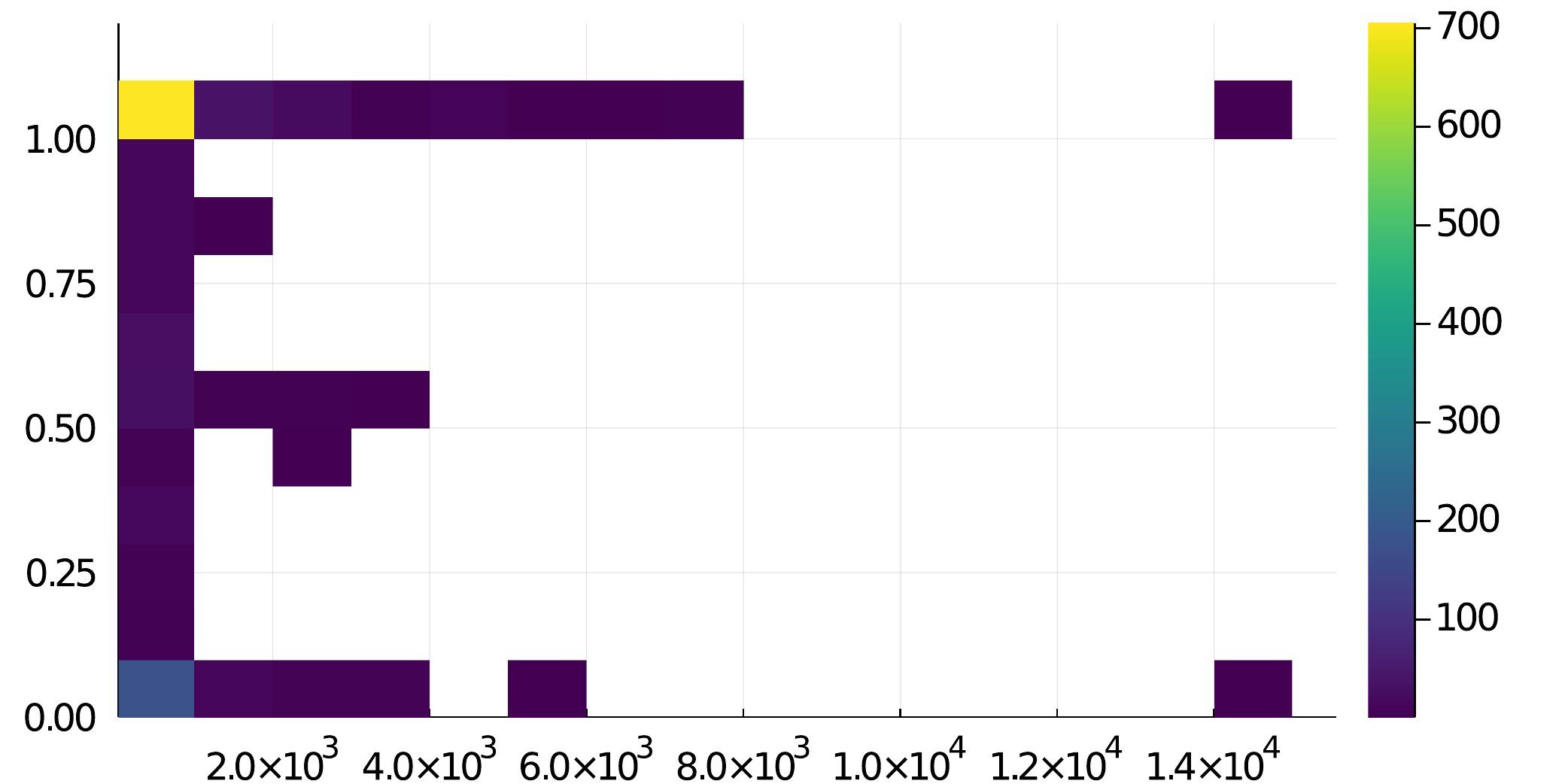}
     \end{subfigure}
     \ \
     \begin{subfigure}[b]{0.49\textwidth}
       \includegraphics[width=\textwidth]{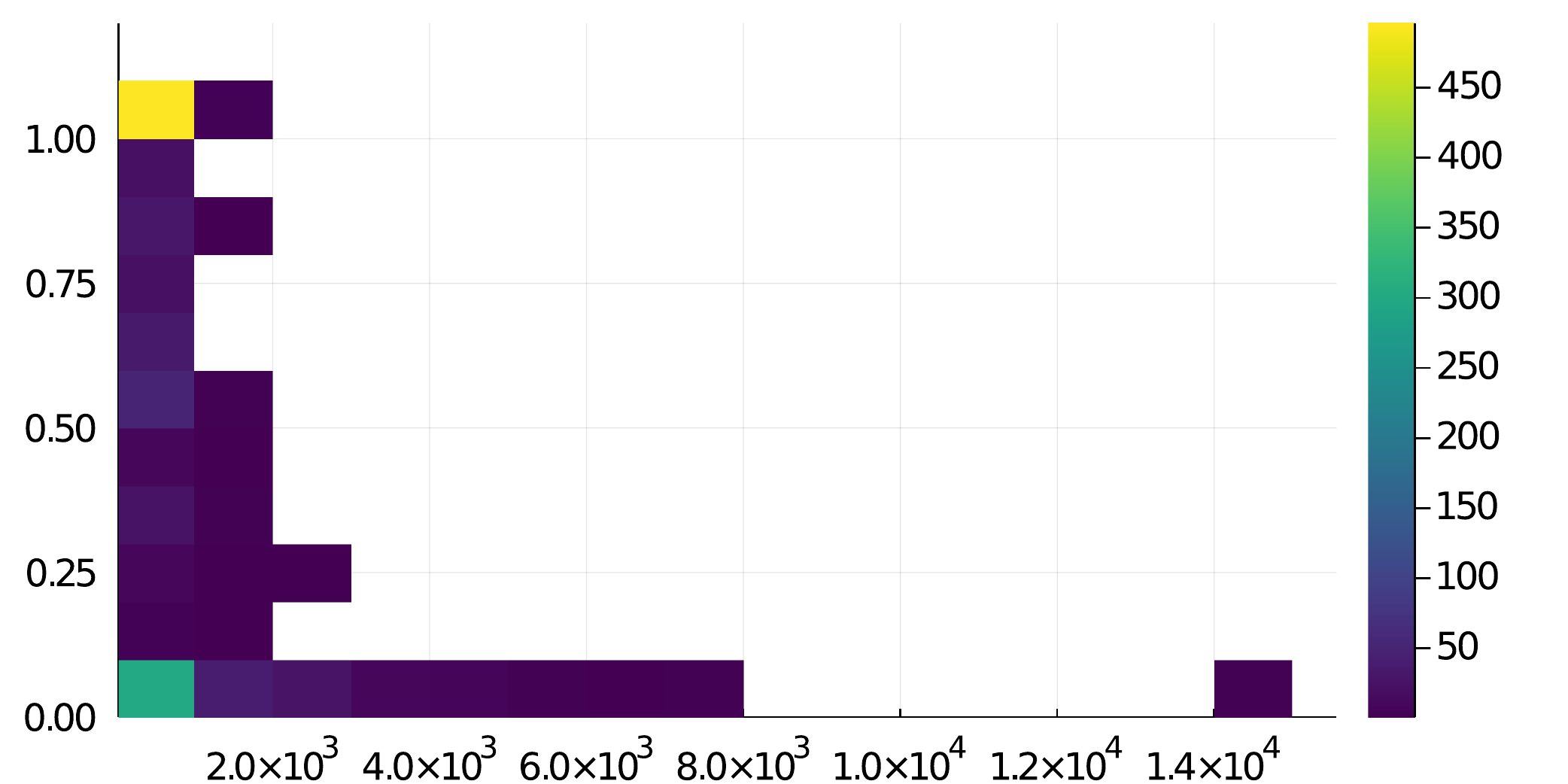}
     \end{subfigure}
\caption{Stability (left, OY axis) and groundedness (right, OY) by method size (OX)}%
\Description{Stability and groundedness by method size in Gadfly}%
\label{figs:size:Gadfly}
\end{figure}

\begin{figure}[h]
     \begin{subfigure}[b]{0.49\textwidth}
       \includegraphics[width=\textwidth]{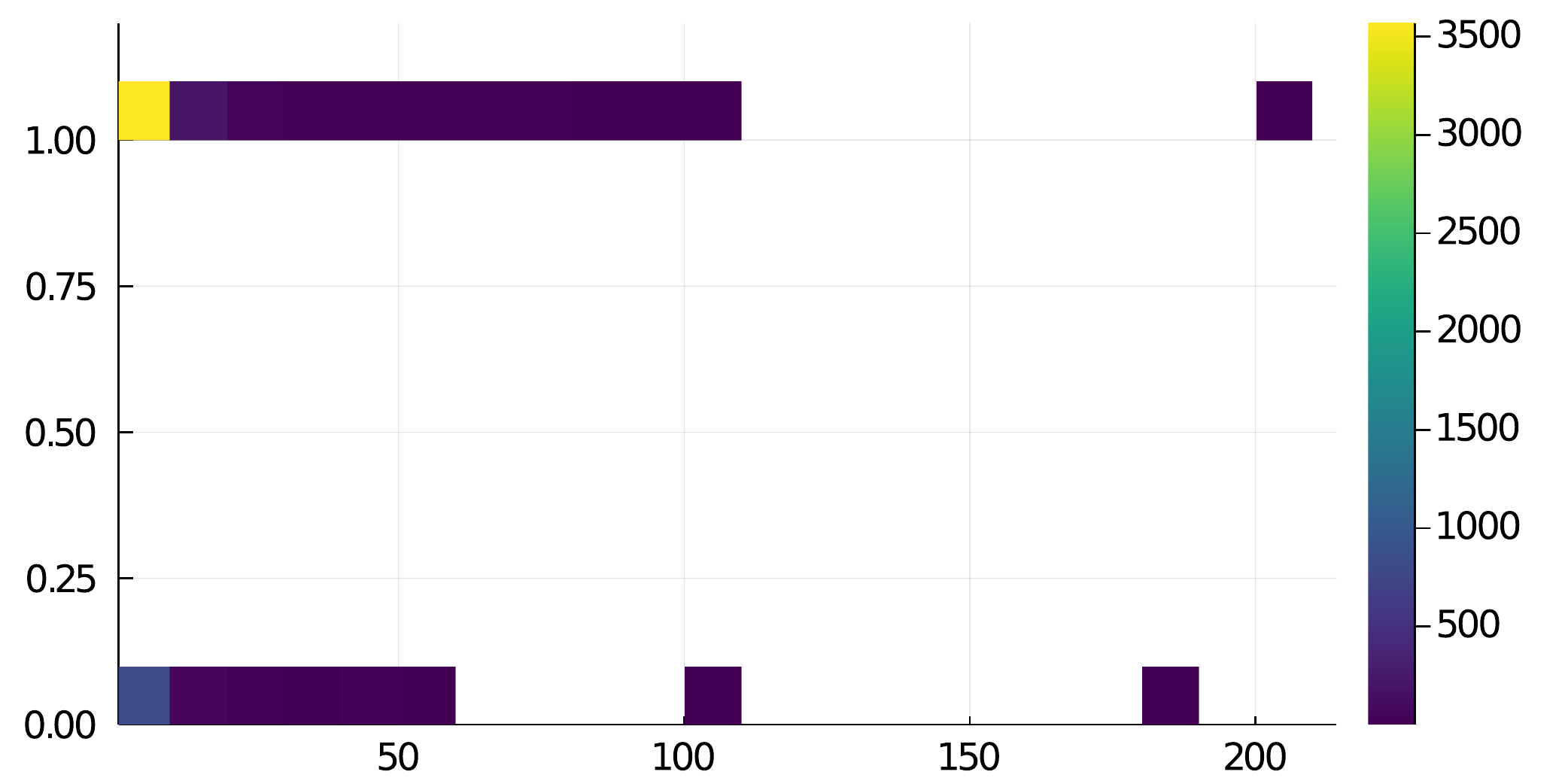}
     \end{subfigure}
     \ \
     \begin{subfigure}[b]{0.49\textwidth}
       \includegraphics[width=\textwidth]{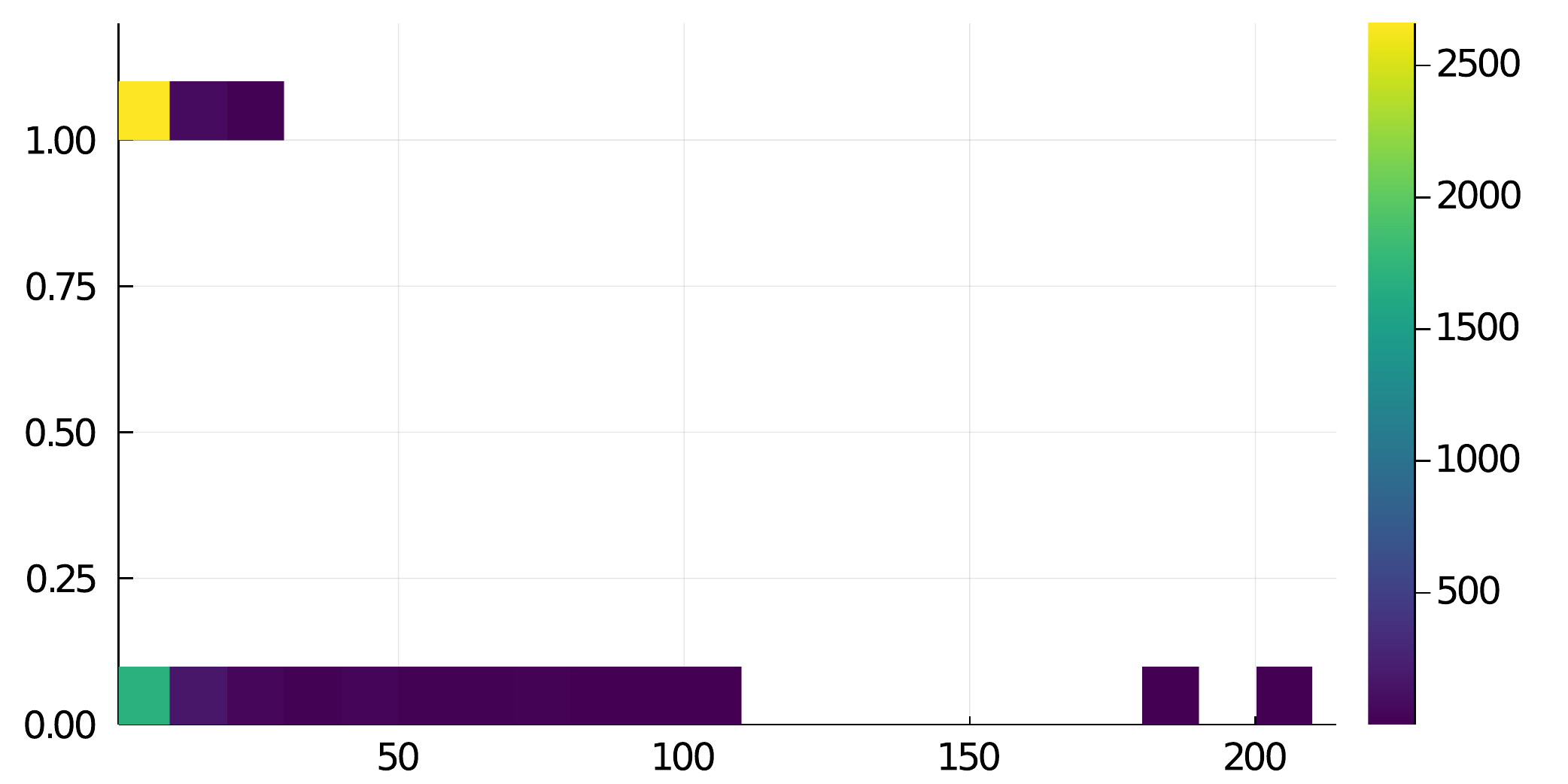}
     \end{subfigure}
\caption{Stability (left, OY axis) and groundedness (right, OY) by number of gotos in method instances (OX)}%
\Description{Stability and groundedness by number of goto's in method instances in Gadfly}%
\label{figs:gotos:Gadfly}
\end{figure}

\begin{figure}[h]
     \begin{subfigure}[b]{0.49\textwidth}
       \includegraphics[width=\textwidth]{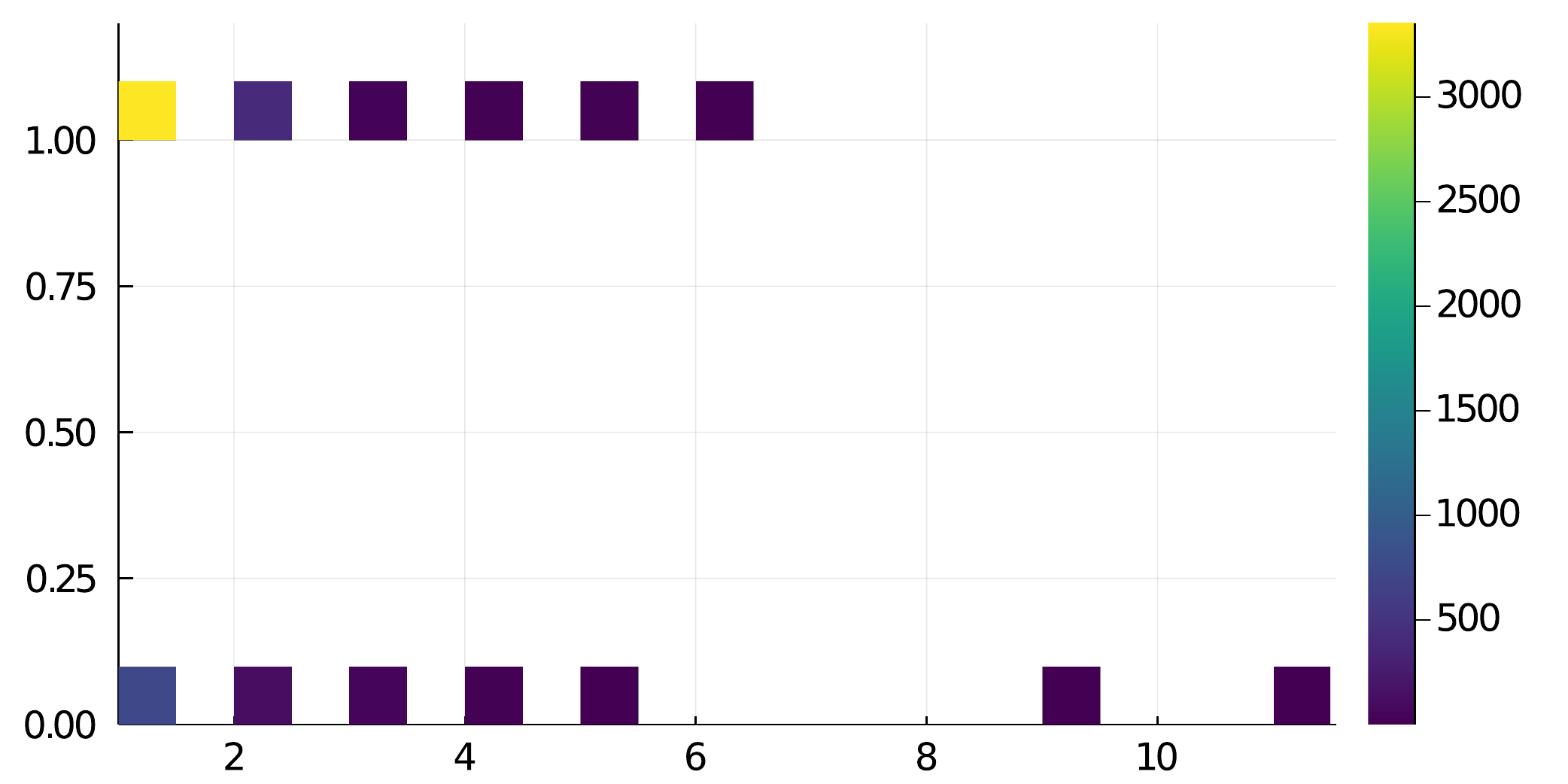}
     \end{subfigure}
     \ \
     \begin{subfigure}[b]{0.49\textwidth}
       \includegraphics[width=\textwidth]{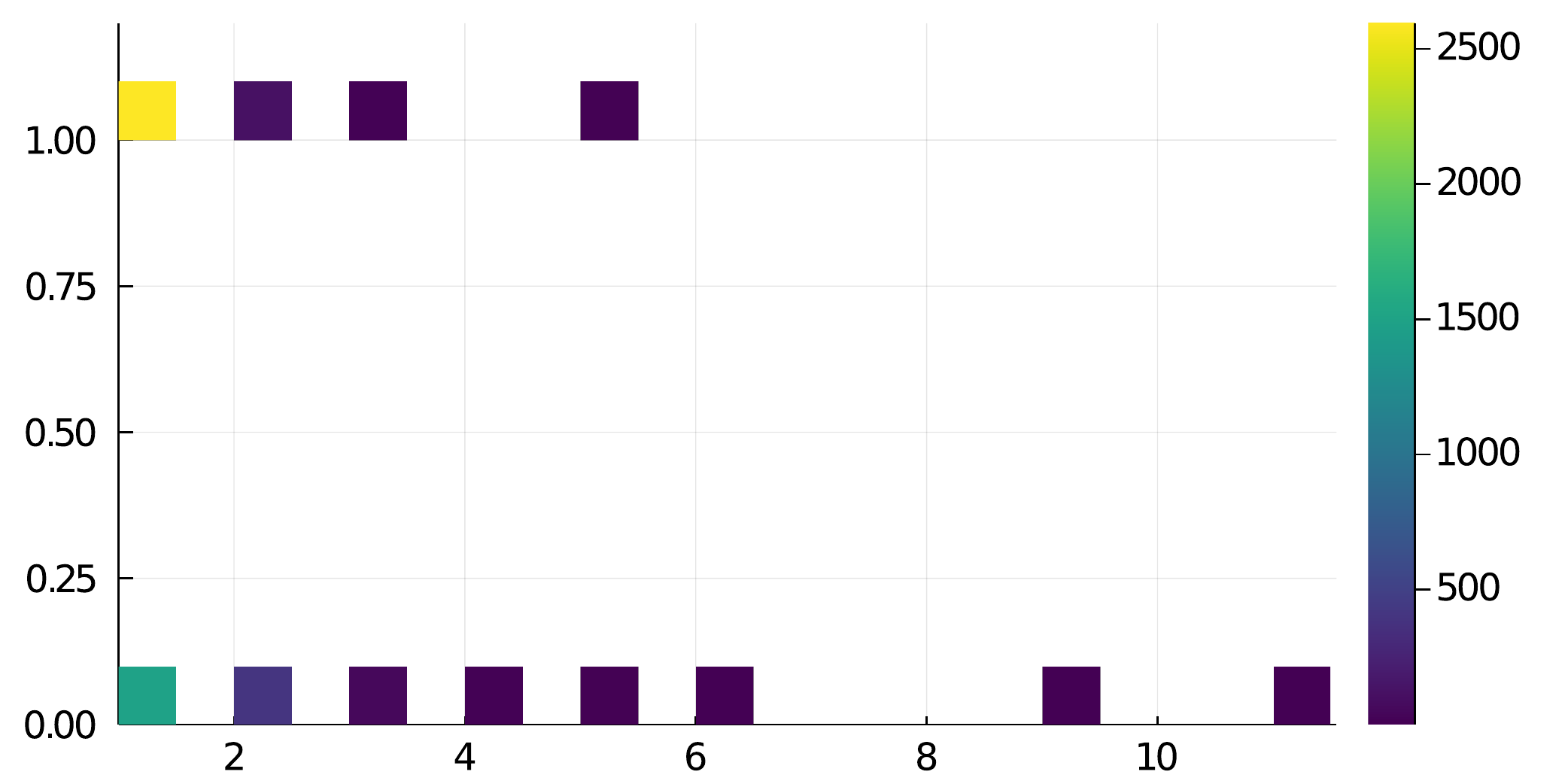}
     \end{subfigure}
\caption{Stability (left, OY axis) and groundedness (right, OY) by number of returns in method instances (OX)}%
\Description{Stability and groundedness by number of returns in method instances in Gadfly}%
\label{figs:returns:Gadfly}
\end{figure}
\clearpage
\subsection{Package: Gen}
\begin{figure}[h]
     \begin{subfigure}[b]{0.49\textwidth}
       \includegraphics[width=\textwidth]{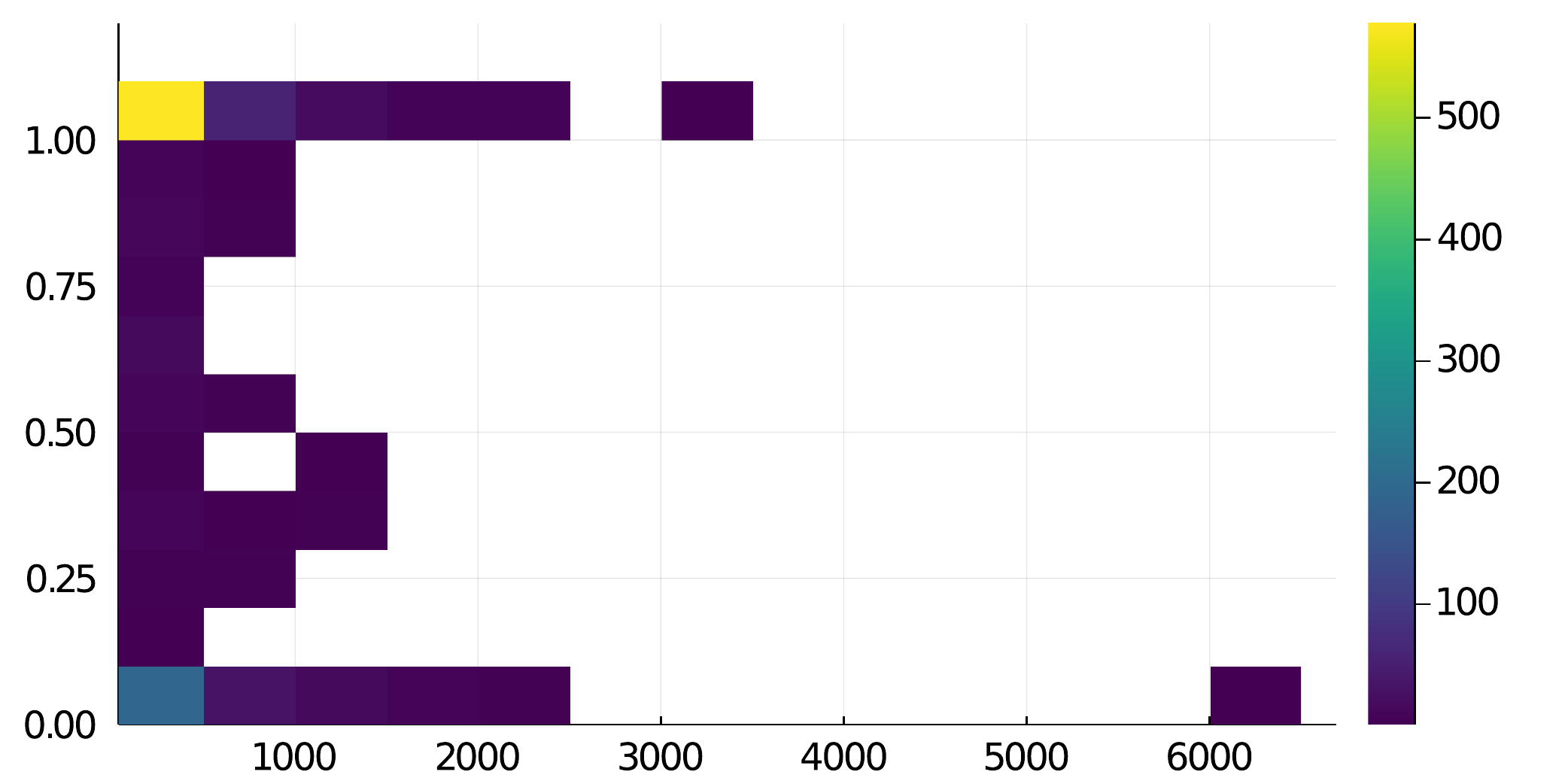}
     \end{subfigure}
     \ \
     \begin{subfigure}[b]{0.49\textwidth}
       \includegraphics[width=\textwidth]{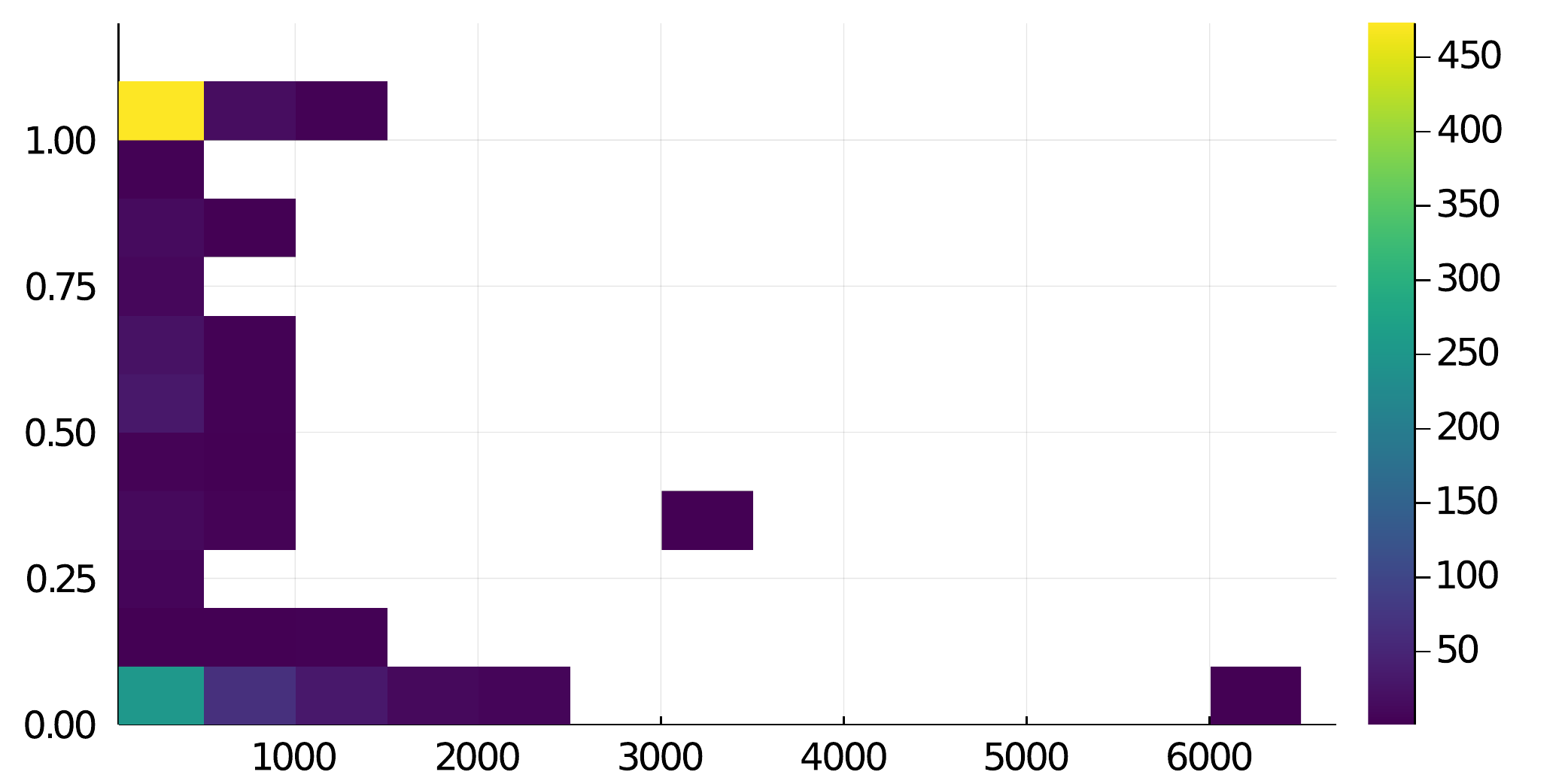}
     \end{subfigure}
\caption{Stability (left, OY axis) and groundedness (right, OY) by method size (OX)}%
\Description{Stability and groundedness by method size in Gen}%
\label{figs:size:Gen}
\end{figure}

\begin{figure}[h]
     \begin{subfigure}[b]{0.49\textwidth}
       \includegraphics[width=\textwidth]{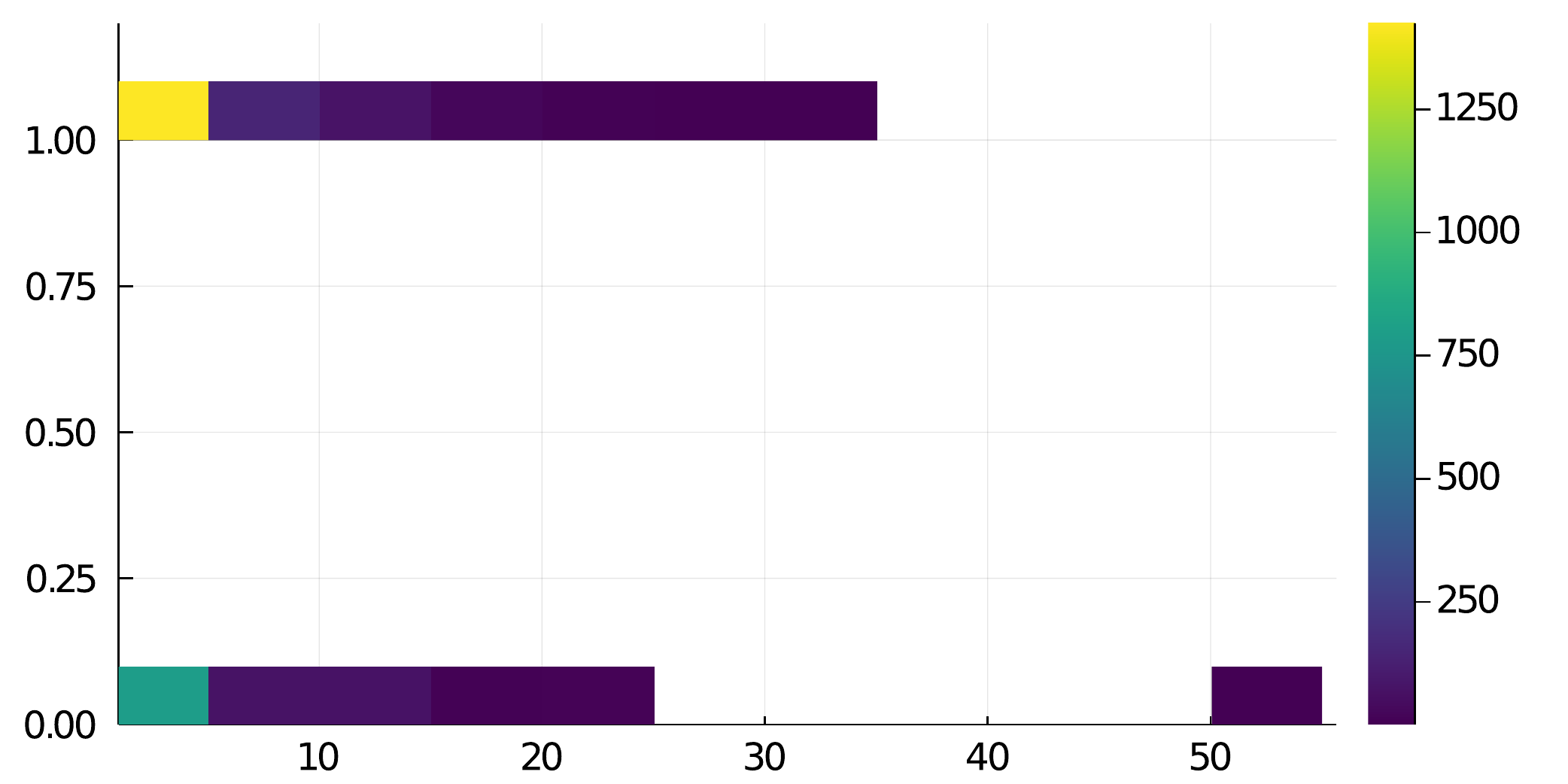}
     \end{subfigure}
     \ \
     \begin{subfigure}[b]{0.49\textwidth}
       \includegraphics[width=\textwidth]{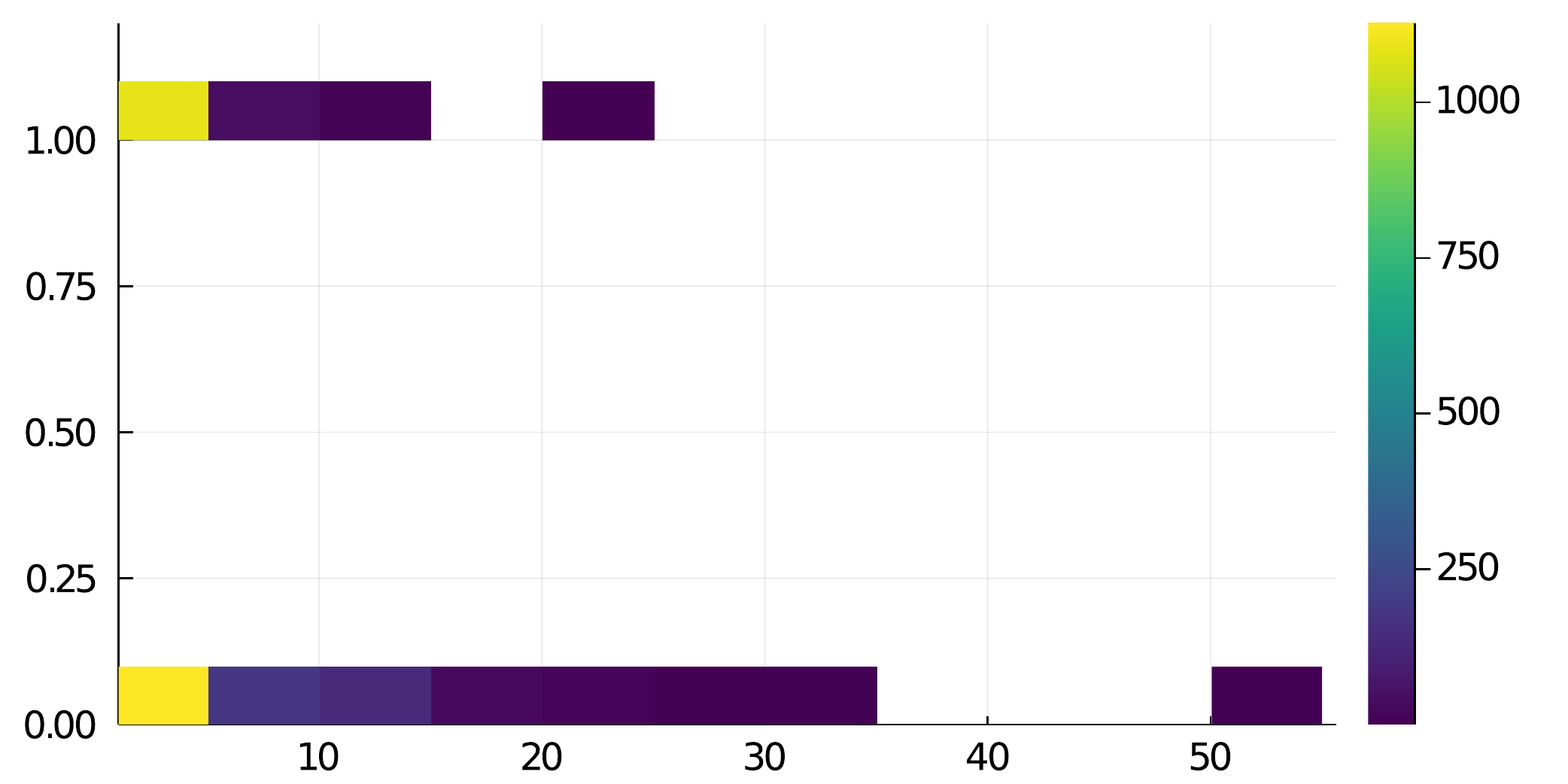}
     \end{subfigure}
\caption{Stability (left, OY axis) and groundedness (right, OY) by number of gotos in method instances (OX)}%
\Description{Stability and groundedness by number of goto's in method instances in Gen}%
\label{figs:gotos:Gen}
\end{figure}

\begin{figure}[h]
     \begin{subfigure}[b]{0.49\textwidth}
       \includegraphics[width=\textwidth]{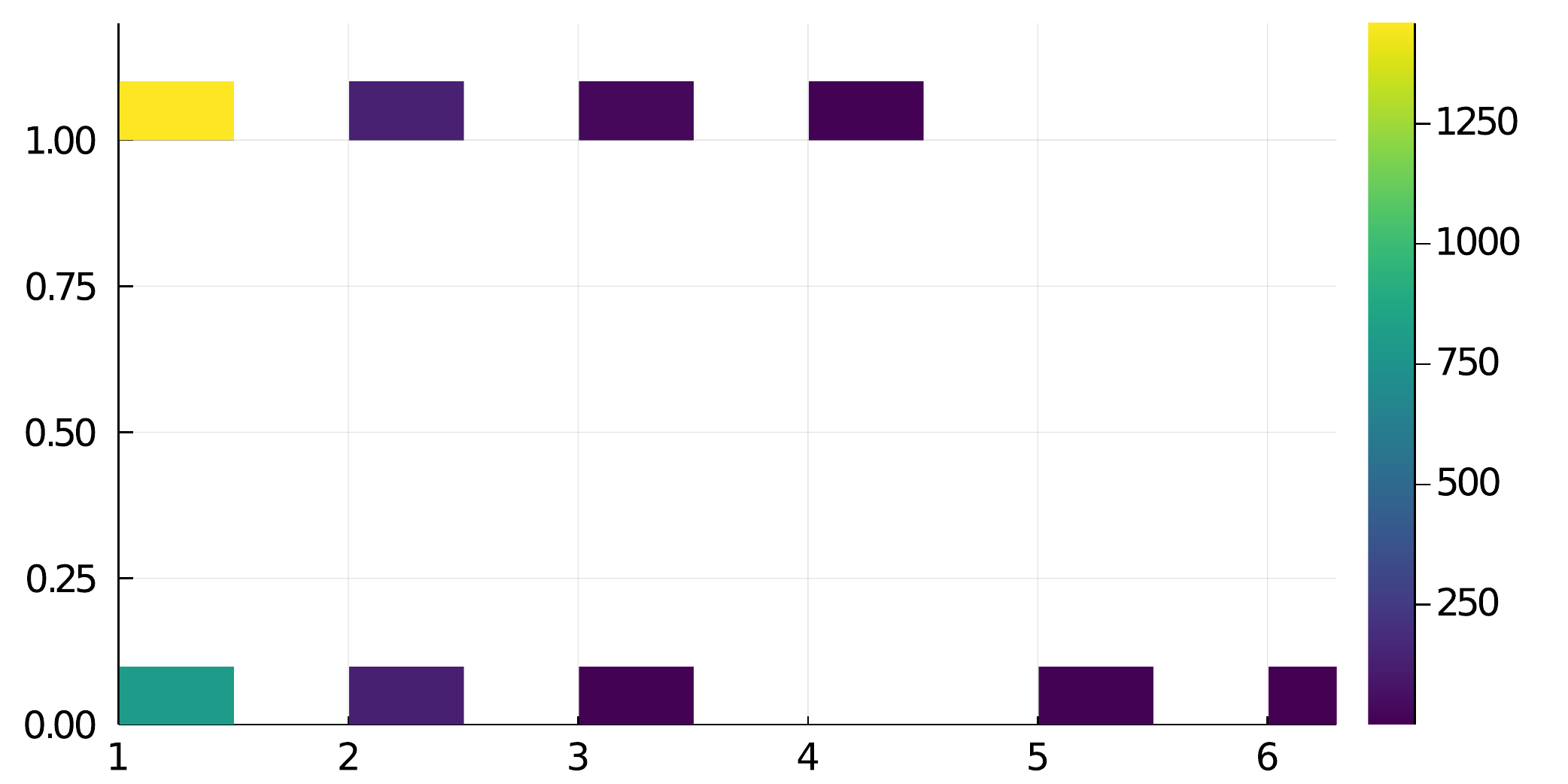}
     \end{subfigure}
     \ \
     \begin{subfigure}[b]{0.49\textwidth}
       \includegraphics[width=\textwidth]{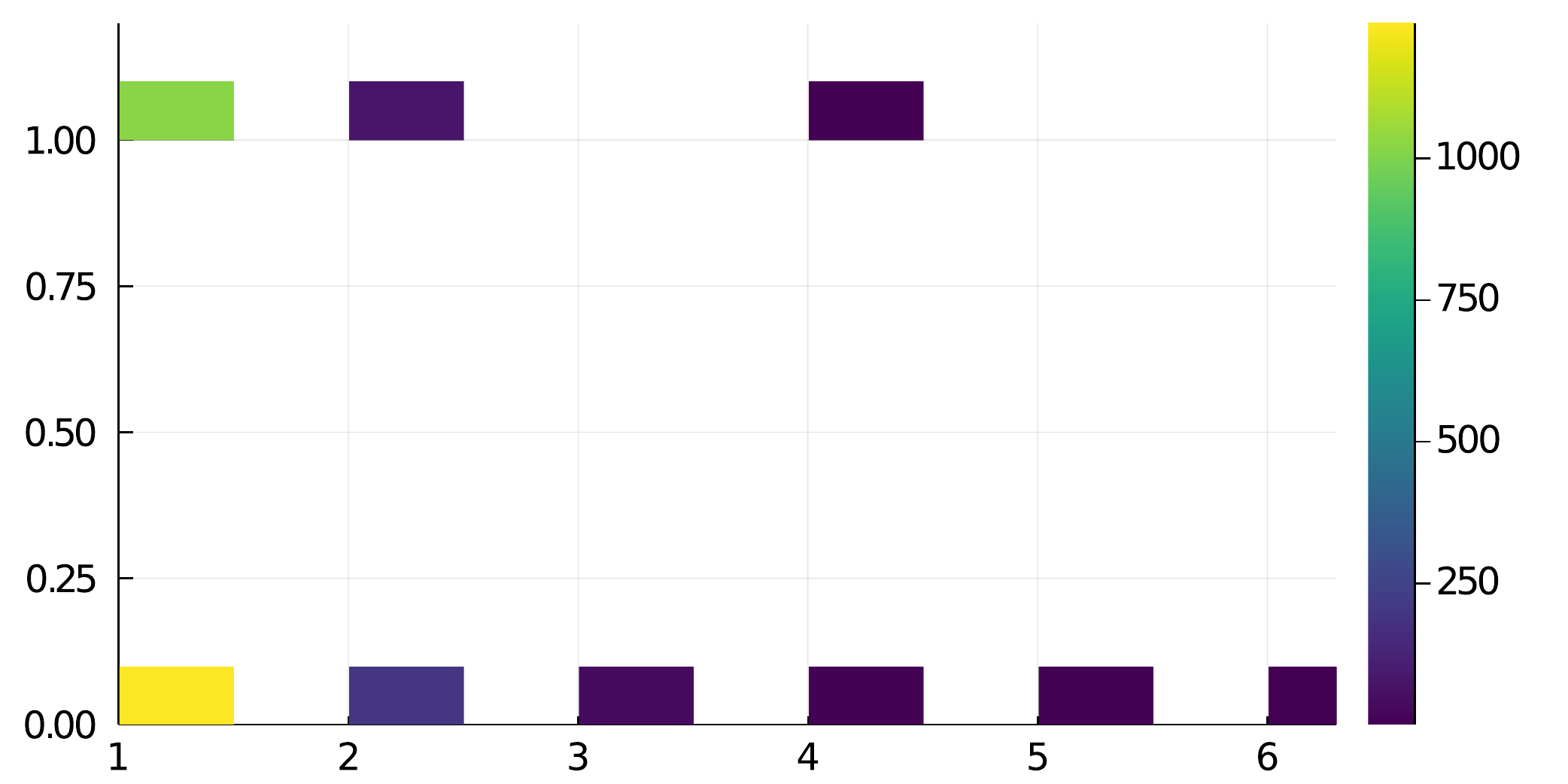}
     \end{subfigure}
\caption{Stability (left, OY axis) and groundedness (right, OY) by number of returns in method instances (OX)}%
\Description{Stability and groundedness by number of returns in method instances in Gen}%
\label{figs:returns:Gen}
\end{figure}
\clearpage
\subsection{Package: Genie}
\begin{figure}[h]
     \begin{subfigure}[b]{0.49\textwidth}
       \includegraphics[width=\textwidth]{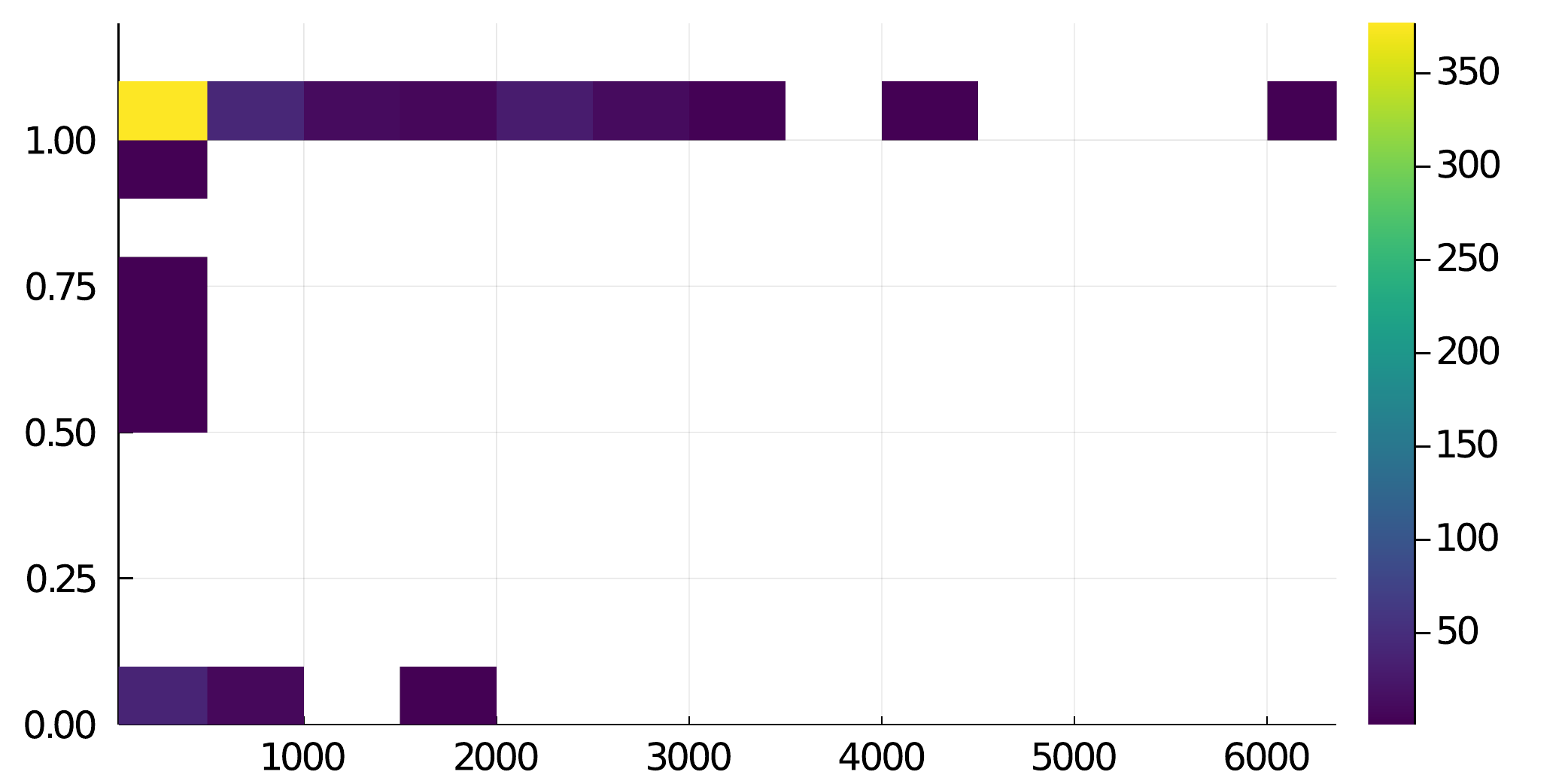}
     \end{subfigure}
     \ \
     \begin{subfigure}[b]{0.49\textwidth}
       \includegraphics[width=\textwidth]{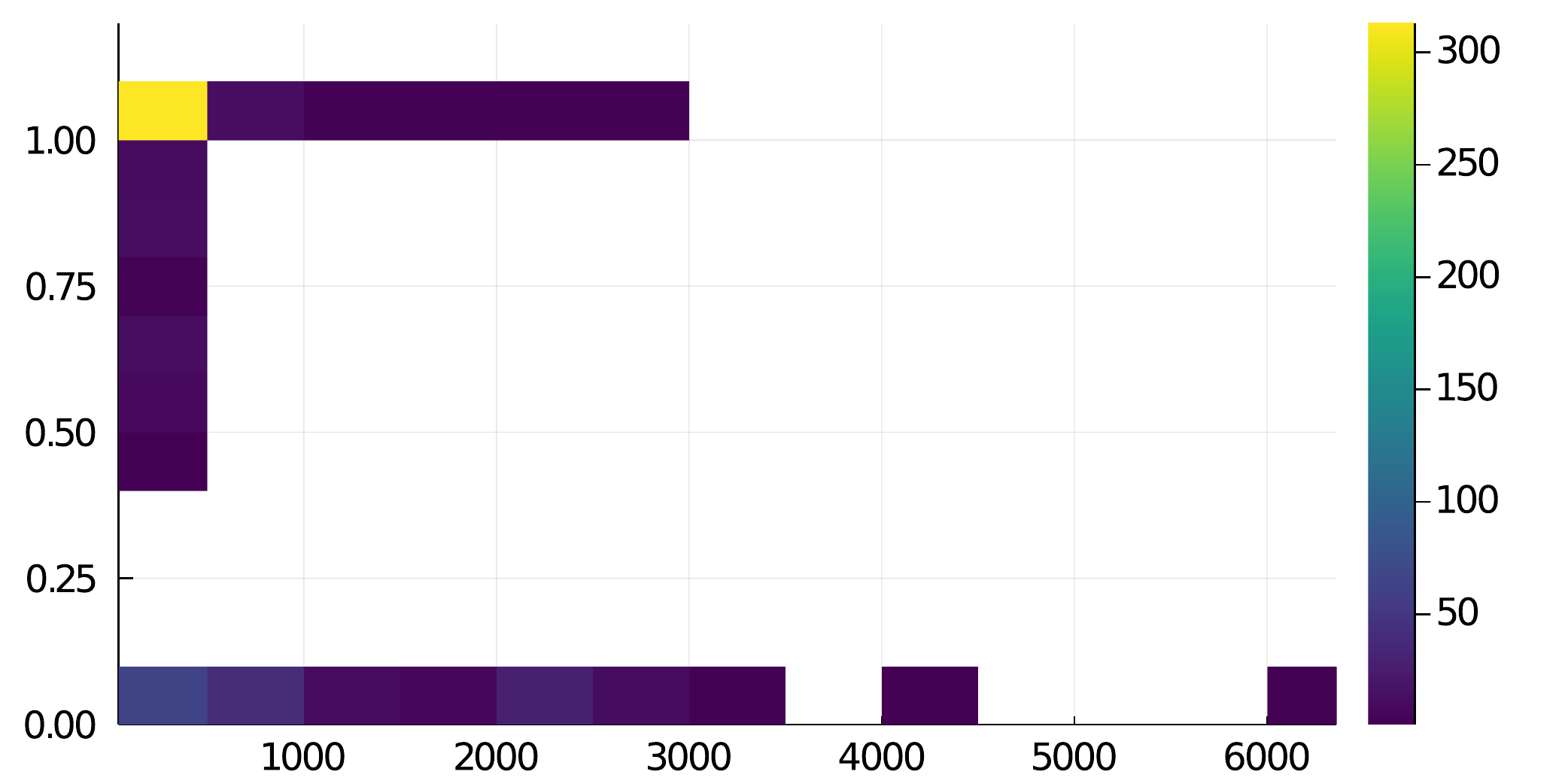}
     \end{subfigure}
\caption{Stability (left, OY axis) and groundedness (right, OY) by method size (OX)}%
\Description{Stability and groundedness by method size in Genie}%
\label{figs:size:Genie}
\end{figure}

\begin{figure}[h]
     \begin{subfigure}[b]{0.49\textwidth}
       \includegraphics[width=\textwidth]{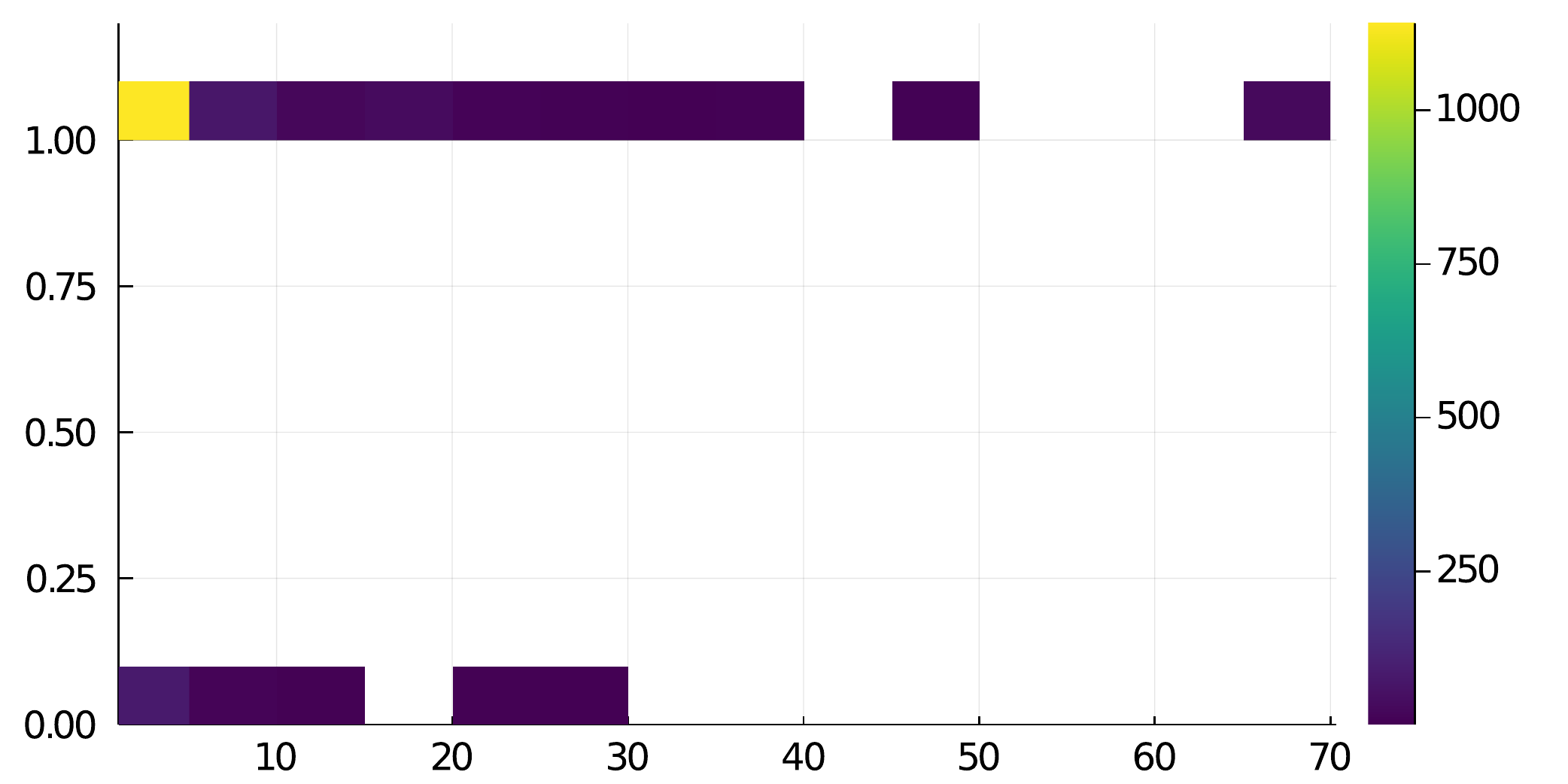}
     \end{subfigure}
     \ \
     \begin{subfigure}[b]{0.49\textwidth}
       \includegraphics[width=\textwidth]{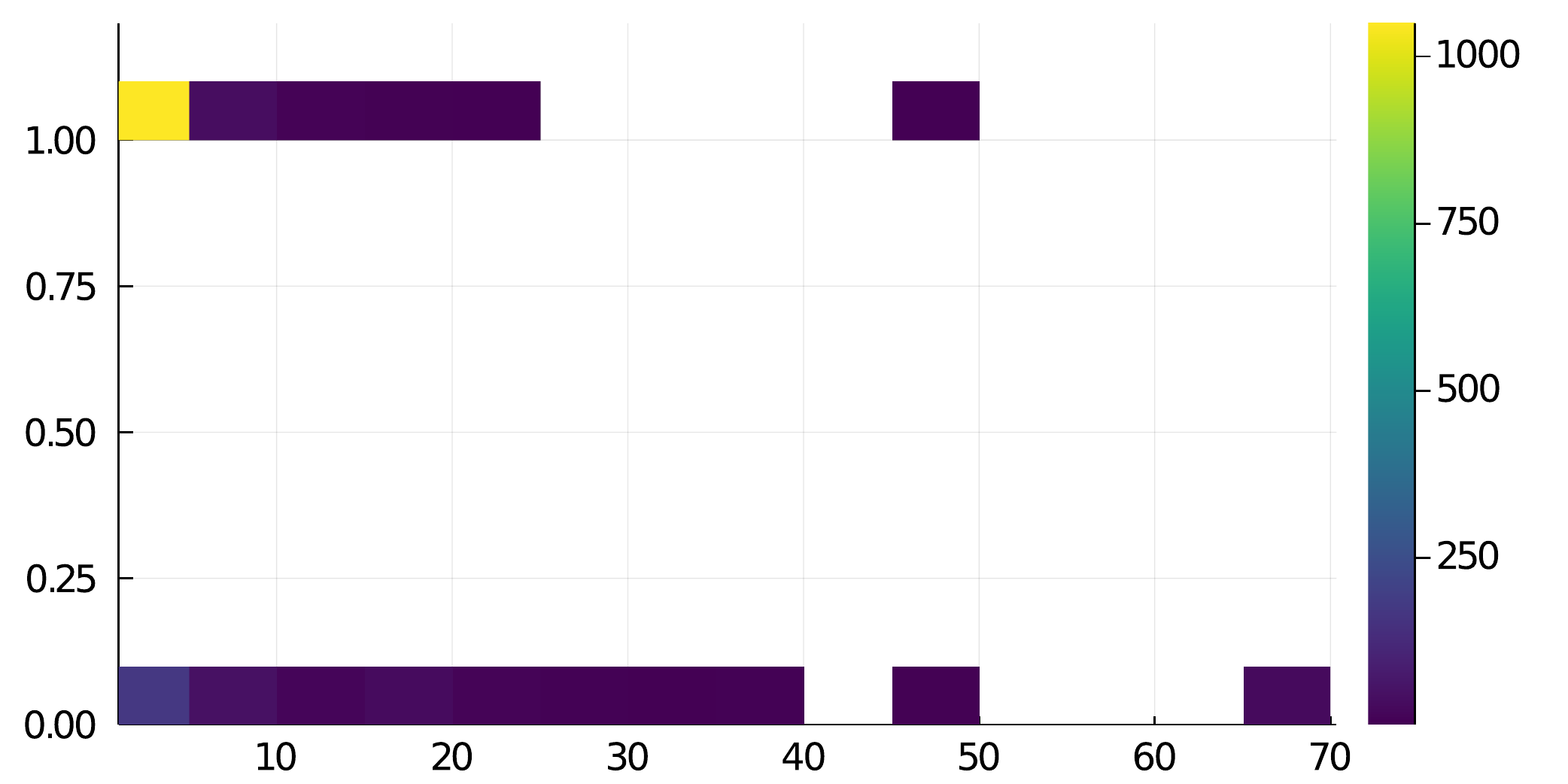}
     \end{subfigure}
\caption{Stability (left, OY axis) and groundedness (right, OY) by number of gotos in method instances (OX)}%
\Description{Stability and groundedness by number of goto's in method instances in Genie}%
\label{figs:gotos:Genie}
\end{figure}

\begin{figure}[h]
     \begin{subfigure}[b]{0.49\textwidth}
       \includegraphics[width=\textwidth]{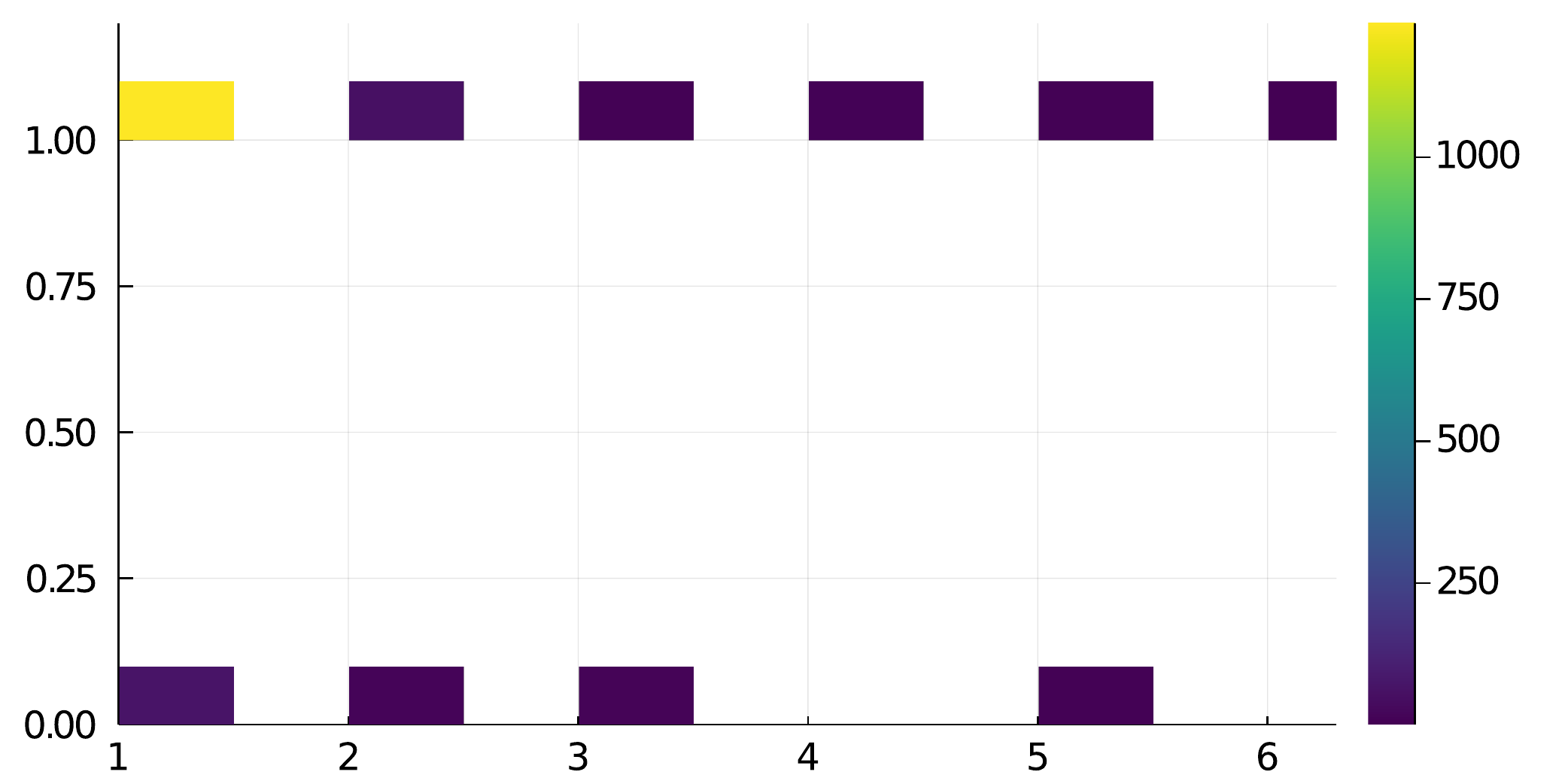}
     \end{subfigure}
     \ \
     \begin{subfigure}[b]{0.49\textwidth}
       \includegraphics[width=\textwidth]{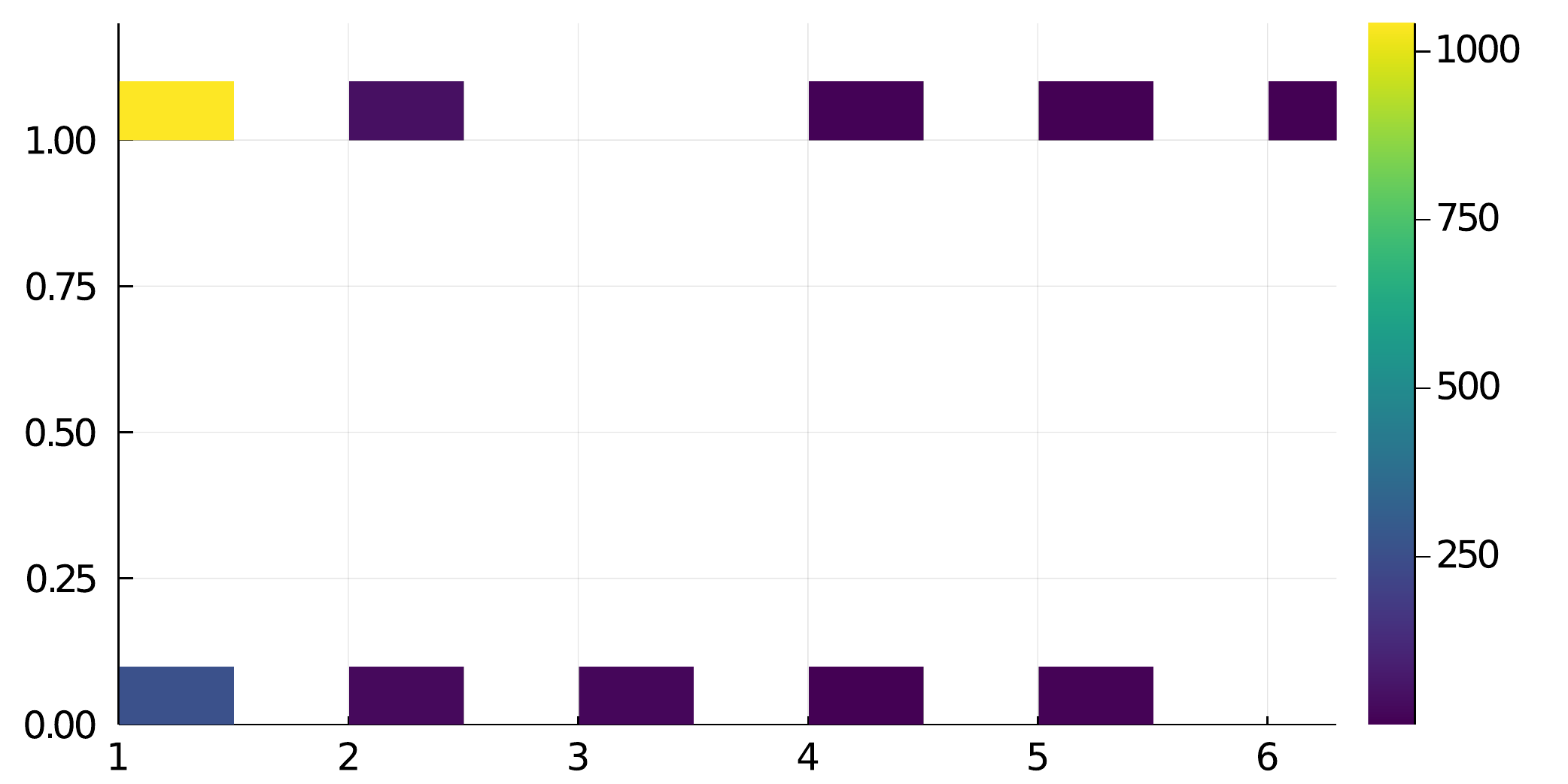}
     \end{subfigure}
\caption{Stability (left, OY axis) and groundedness (right, OY) by number of returns in method instances (OX)}%
\Description{Stability and groundedness by number of returns in method instances in Genie}%
\label{figs:returns:Genie}
\end{figure}
\clearpage
\subsection{Package: IJulia}
\begin{figure}[h]
     \begin{subfigure}[b]{0.49\textwidth}
       \includegraphics[width=\textwidth]{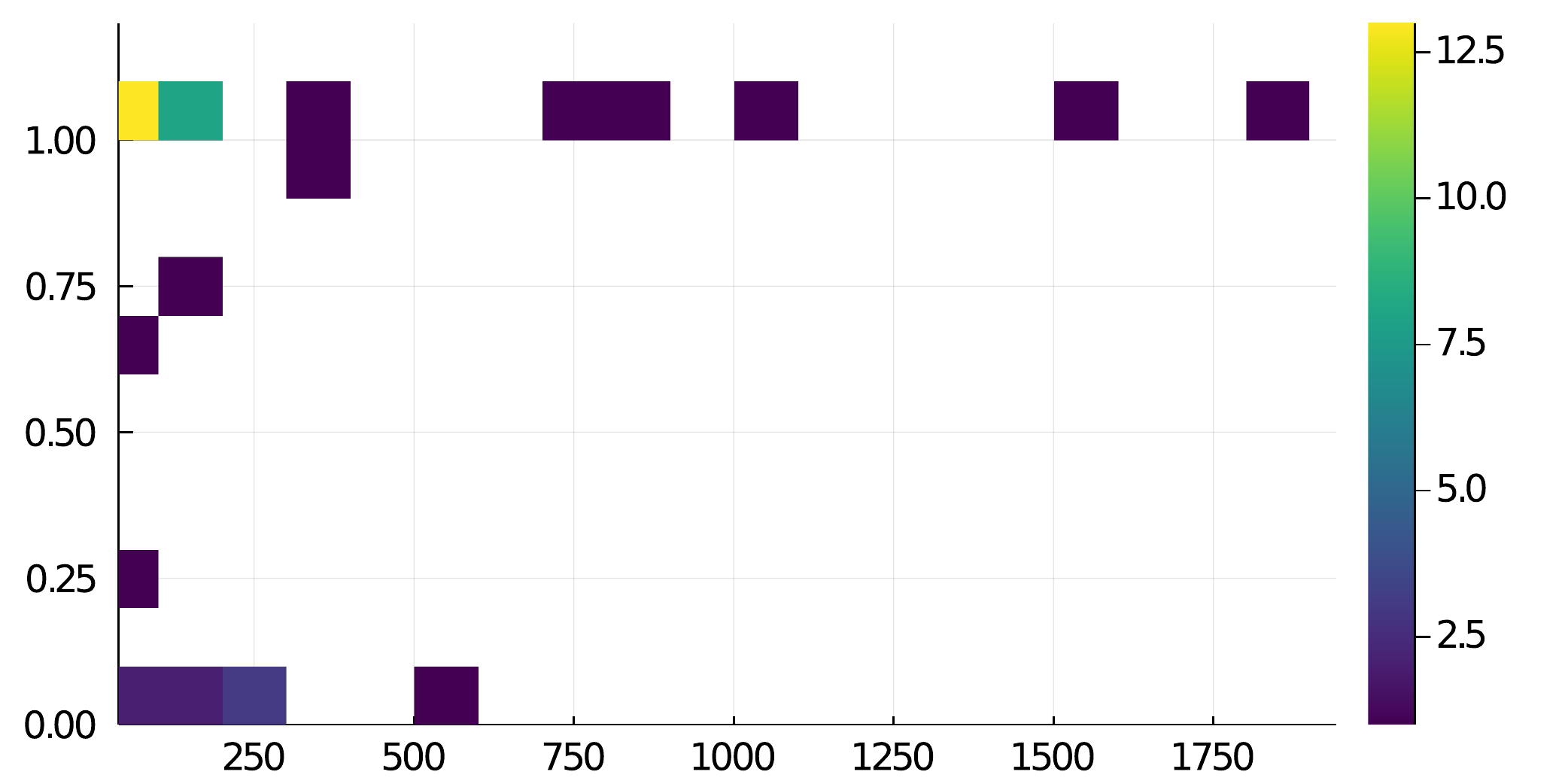}
     \end{subfigure}
     \ \
     \begin{subfigure}[b]{0.49\textwidth}
       \includegraphics[width=\textwidth]{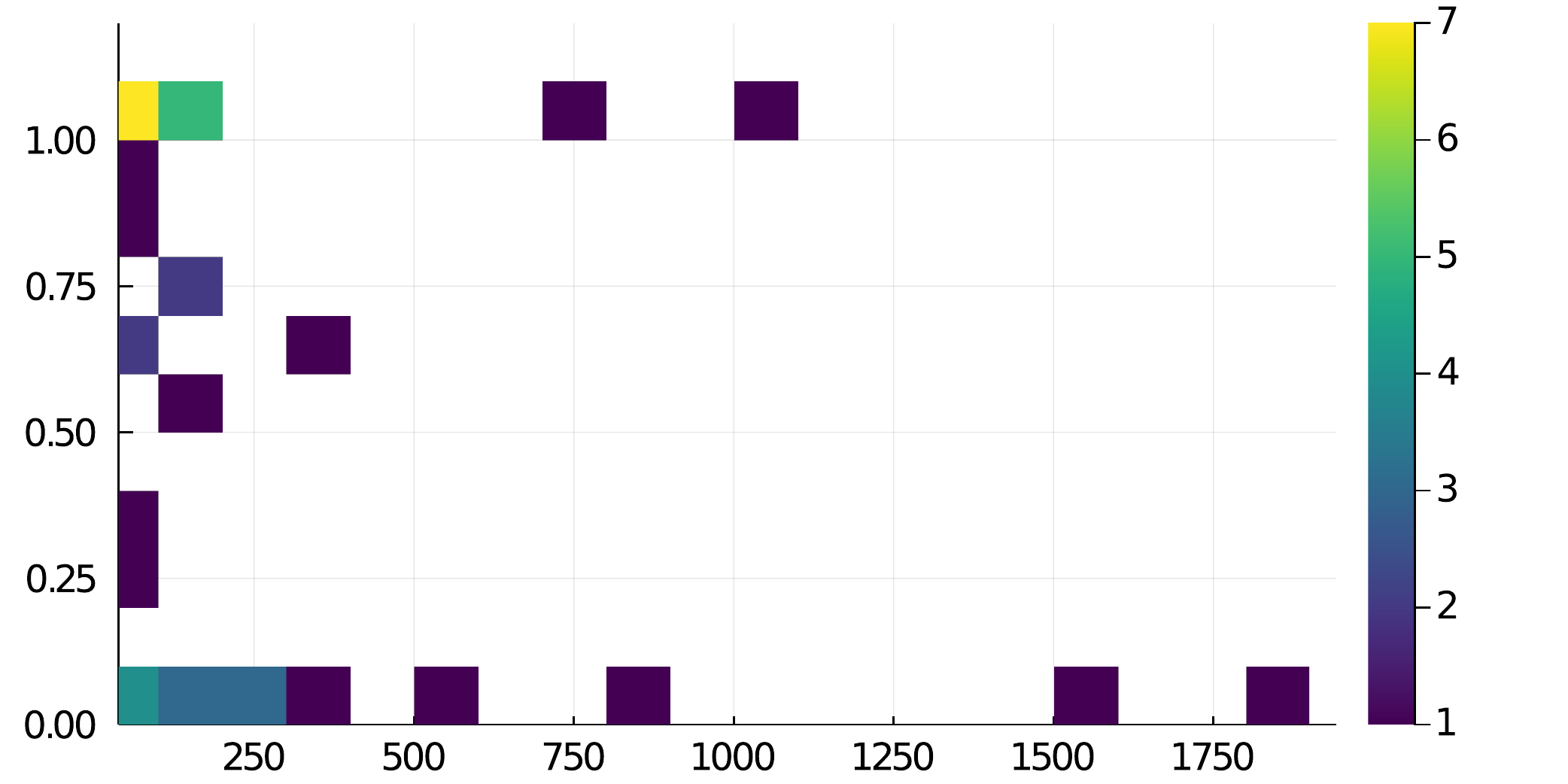}
     \end{subfigure}
\caption{Stability (left, OY axis) and groundedness (right, OY) by method size (OX)}%
\Description{Stability and groundedness by method size in IJulia}%
\label{figs:size:IJulia}
\end{figure}

\begin{figure}[h]
     \begin{subfigure}[b]{0.49\textwidth}
       \includegraphics[width=\textwidth]{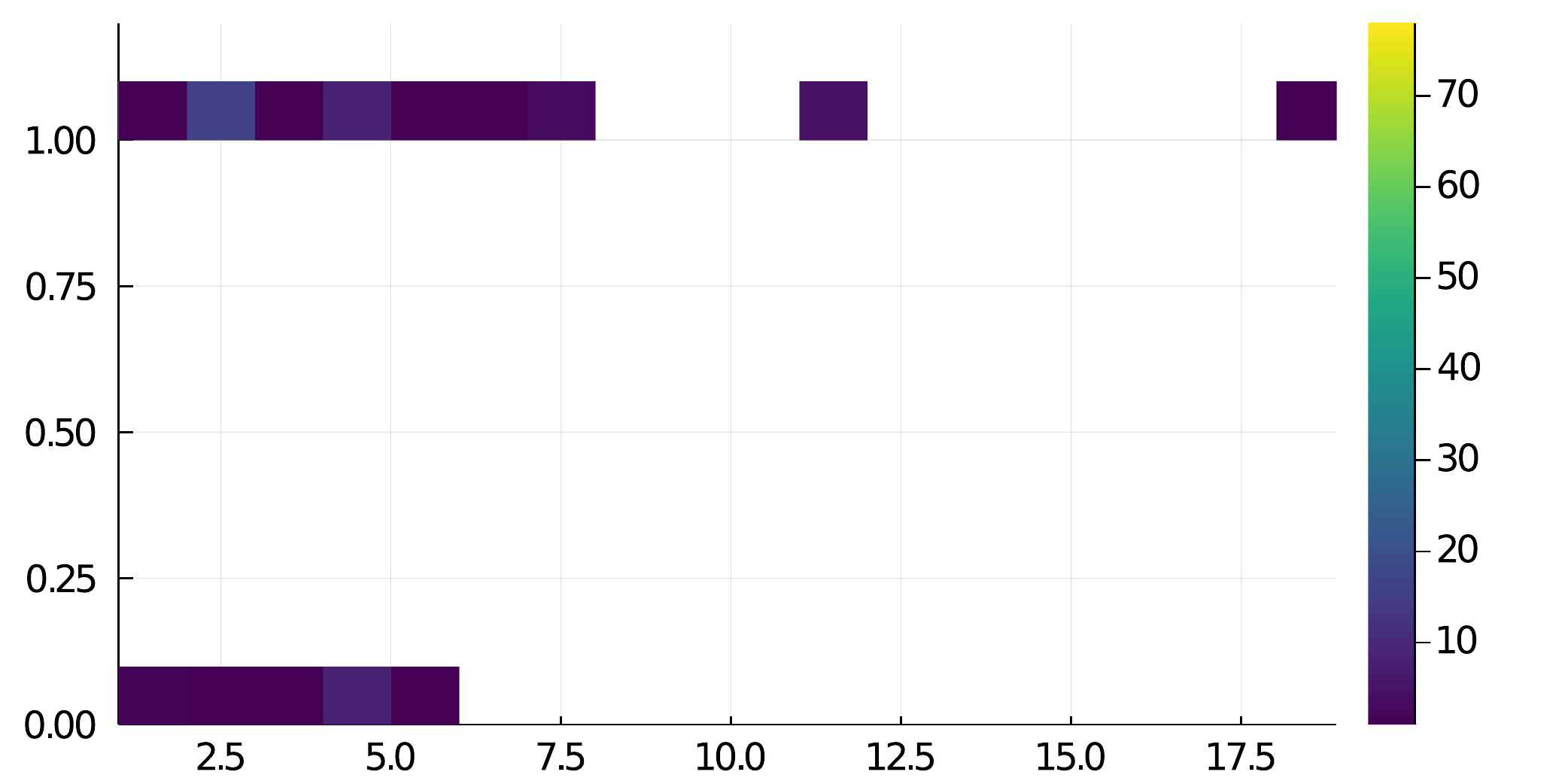}
     \end{subfigure}
     \ \
     \begin{subfigure}[b]{0.49\textwidth}
       \includegraphics[width=\textwidth]{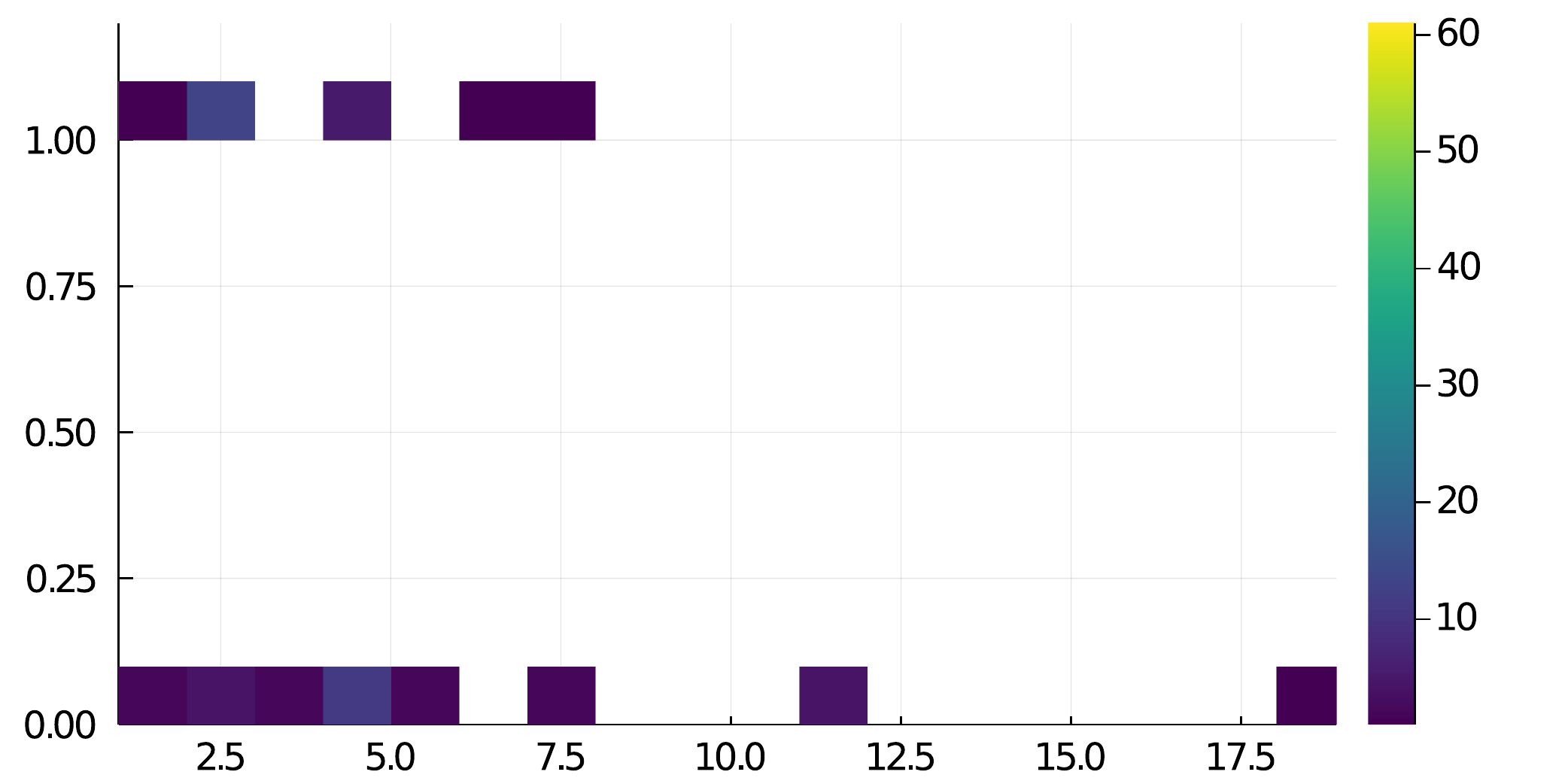}
     \end{subfigure}
\caption{Stability (left, OY axis) and groundedness (right, OY) by number of gotos in method instances (OX)}%
\Description{Stability and groundedness by number of goto's in method instances in IJulia}%
\label{figs:gotos:IJulia}
\end{figure}

\begin{figure}[h]
     \begin{subfigure}[b]{0.49\textwidth}
       \includegraphics[width=\textwidth]{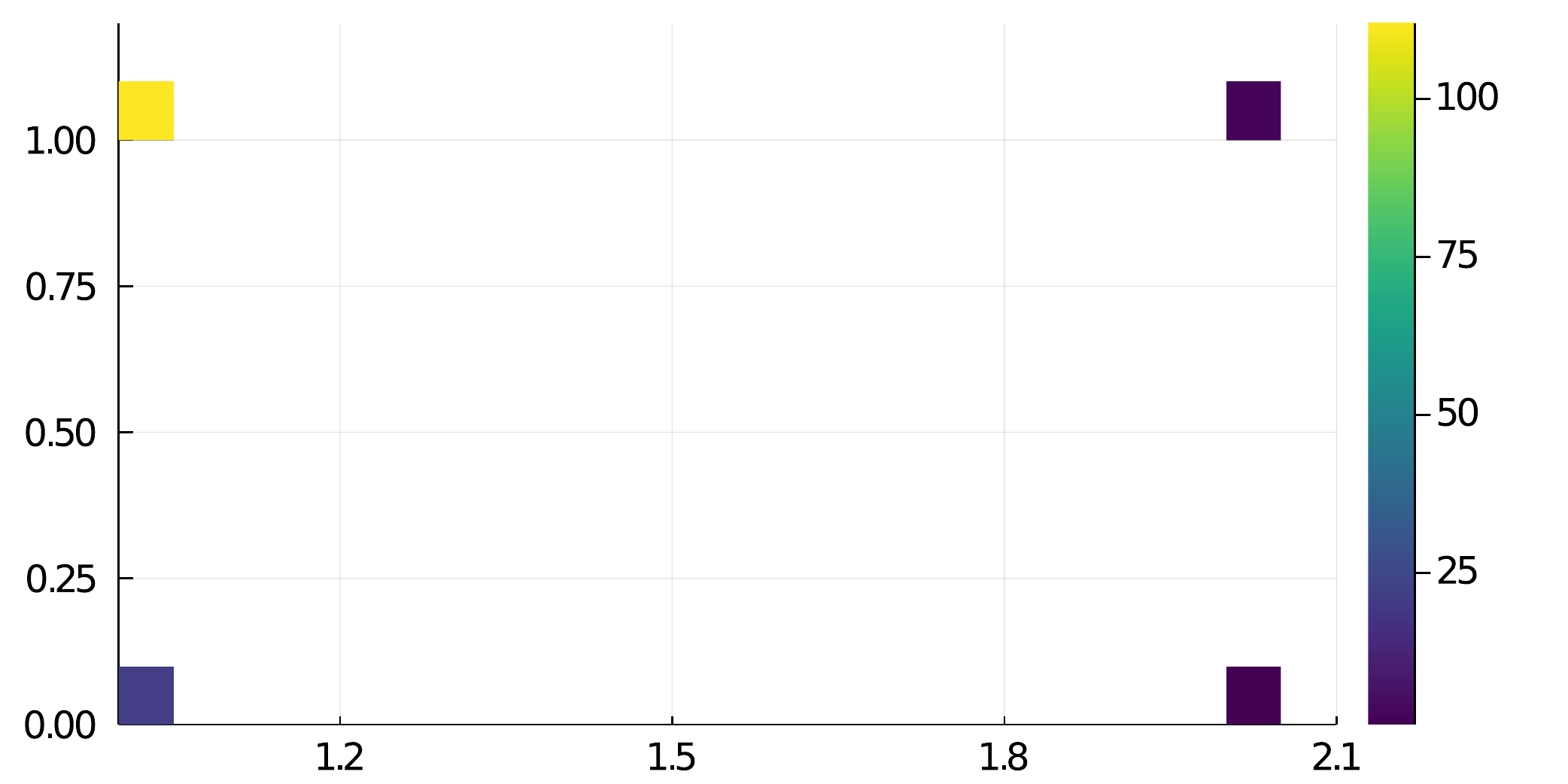}
     \end{subfigure}
     \ \
     \begin{subfigure}[b]{0.49\textwidth}
       \includegraphics[width=\textwidth]{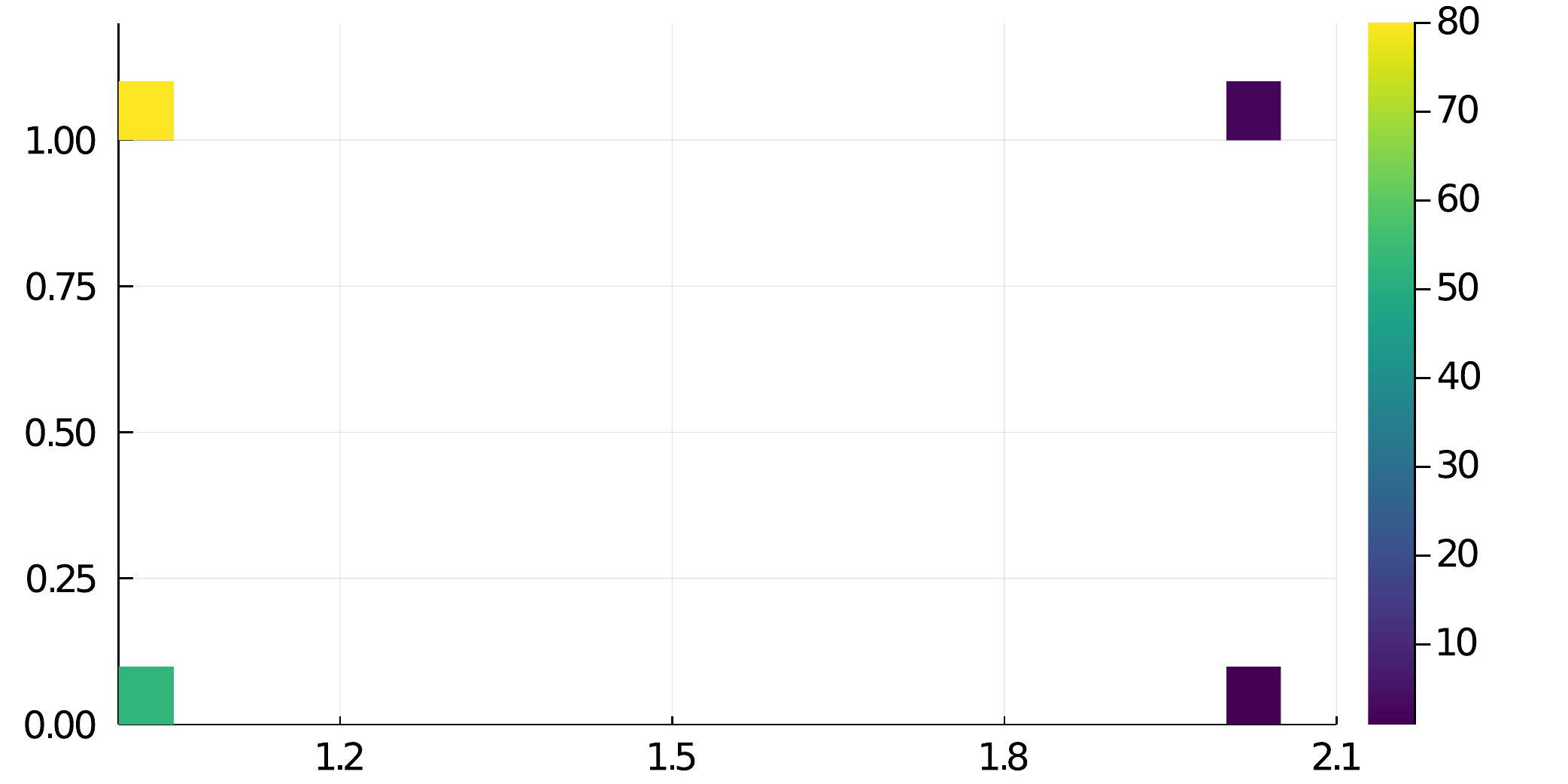}
     \end{subfigure}
\caption{Stability (left, OY axis) and groundedness (right, OY) by number of returns in method instances (OX)}%
\Description{Stability and groundedness by number of returns in method instances in IJulia}%
\label{figs:returns:IJulia}
\end{figure}
\clearpage
\subsection{Package: JuMP}
\begin{figure}[h]
     \begin{subfigure}[b]{0.49\textwidth}
       \includegraphics[width=\textwidth]{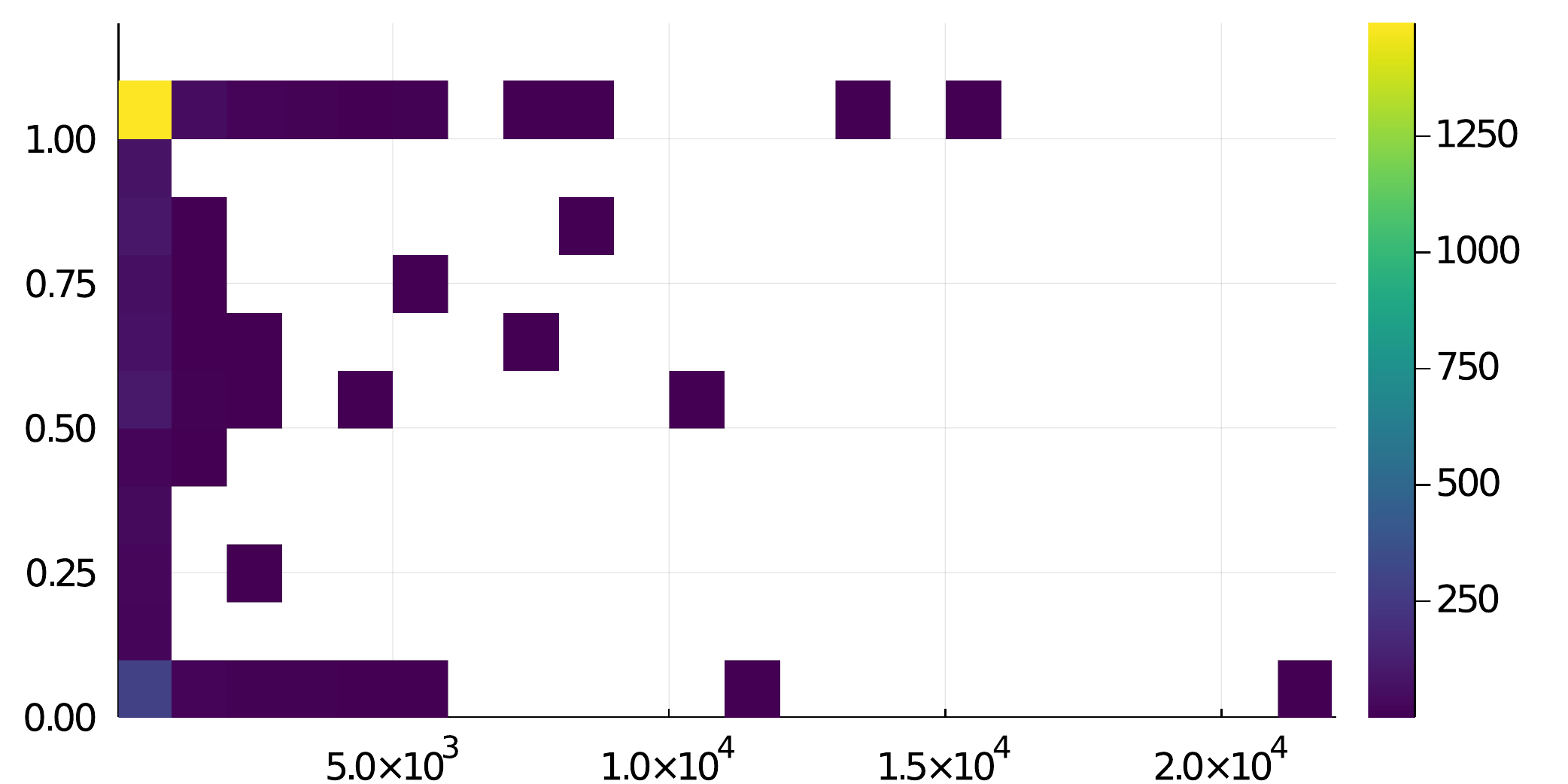}
     \end{subfigure}
     \ \
     \begin{subfigure}[b]{0.49\textwidth}
       \includegraphics[width=\textwidth]{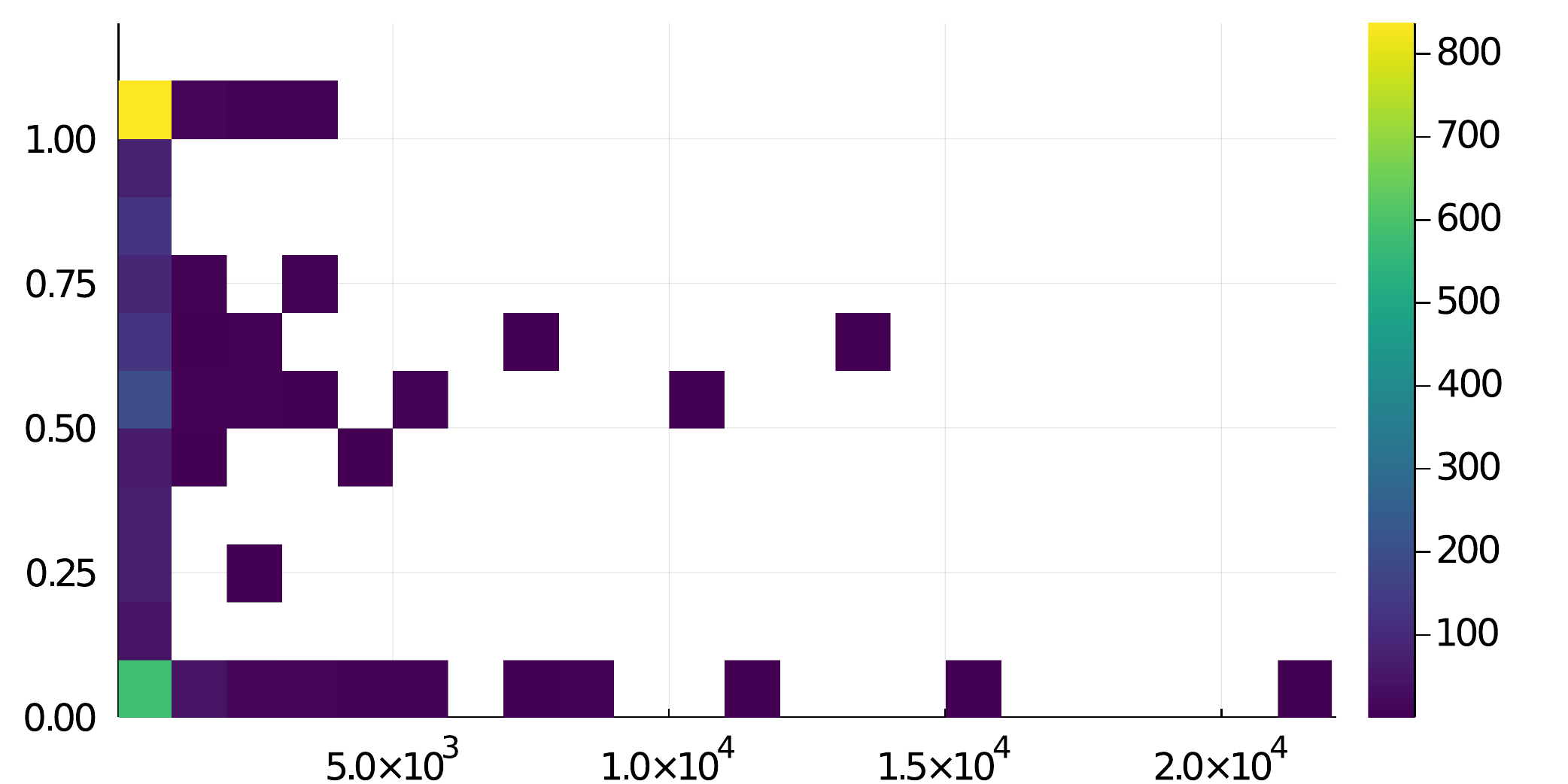}
     \end{subfigure}
\caption{Stability (left, OY axis) and groundedness (right, OY) by method size (OX)}%
\Description{Stability and groundedness by method size in JuMP}%
\label{figs:size:JuMP}
\end{figure}

\begin{figure}[h]
     \begin{subfigure}[b]{0.49\textwidth}
       \includegraphics[width=\textwidth]{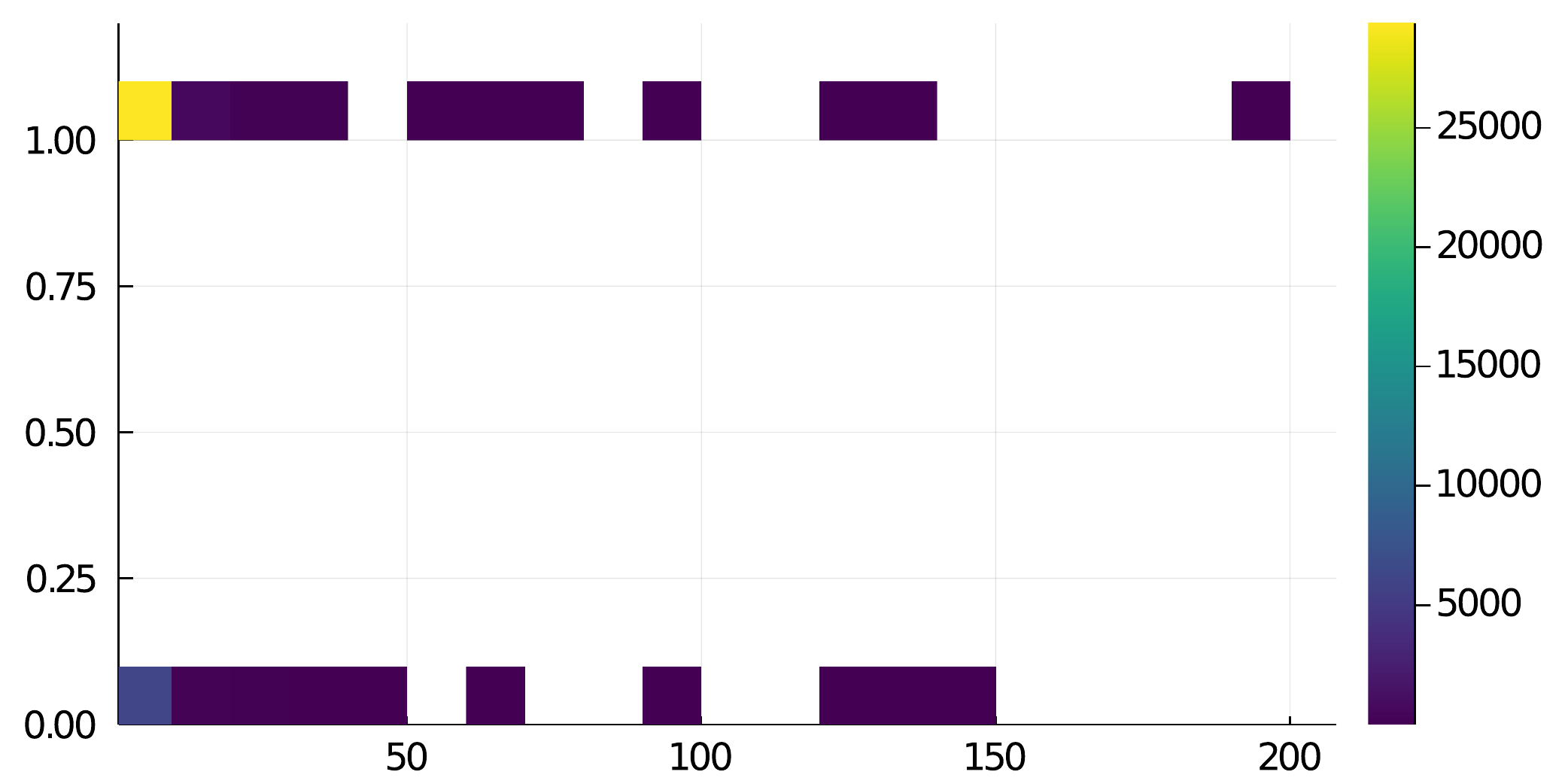}
     \end{subfigure}
     \ \
     \begin{subfigure}[b]{0.49\textwidth}
       \includegraphics[width=\textwidth]{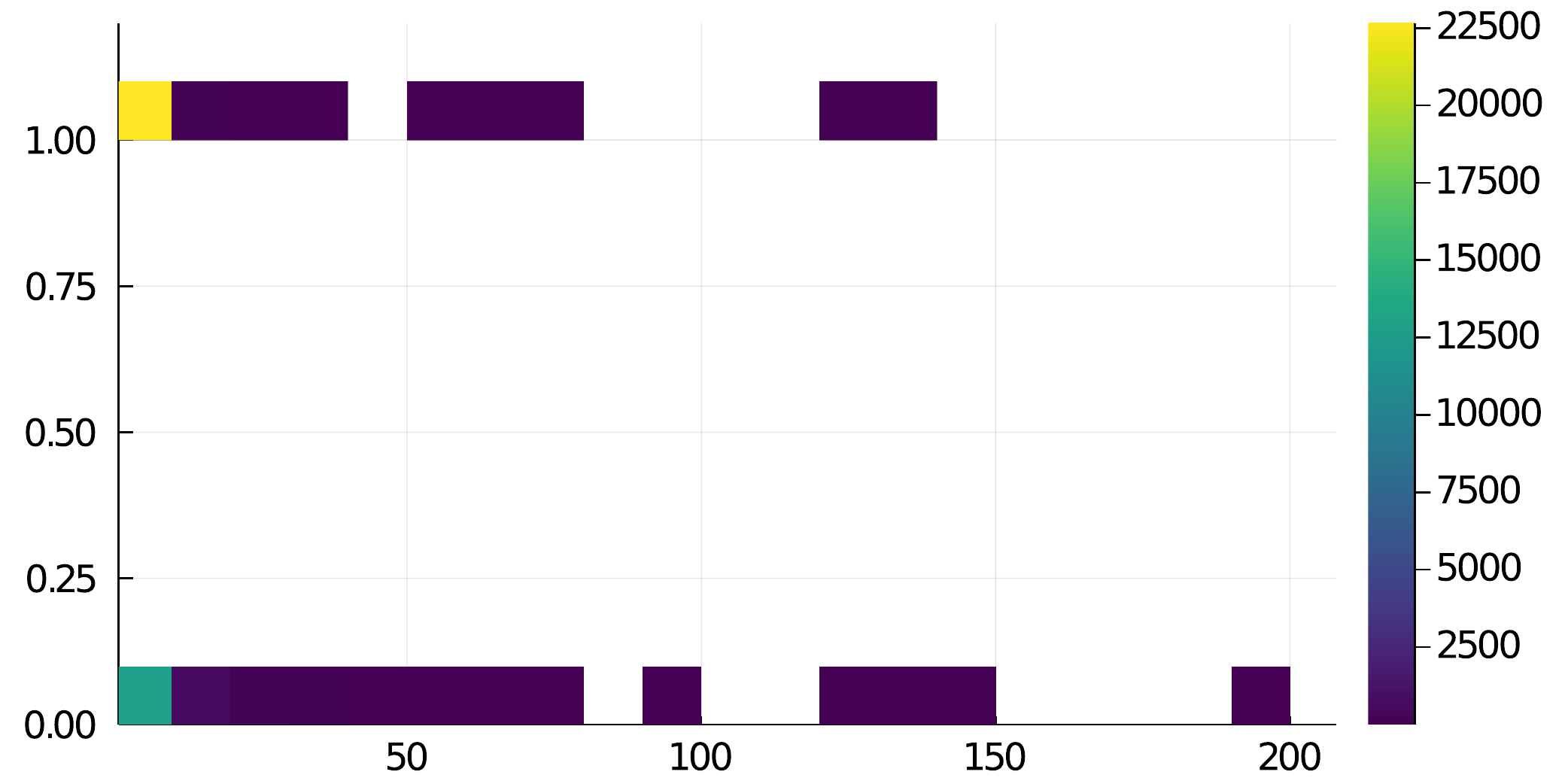}
     \end{subfigure}
\caption{Stability (left, OY axis) and groundedness (right, OY) by number of gotos in method instances (OX)}%
\Description{Stability and groundedness by number of goto's in method instances in JuMP}%
\label{figs:gotos:JuMP}
\end{figure}

\begin{figure}[h]
     \begin{subfigure}[b]{0.49\textwidth}
       \includegraphics[width=\textwidth]{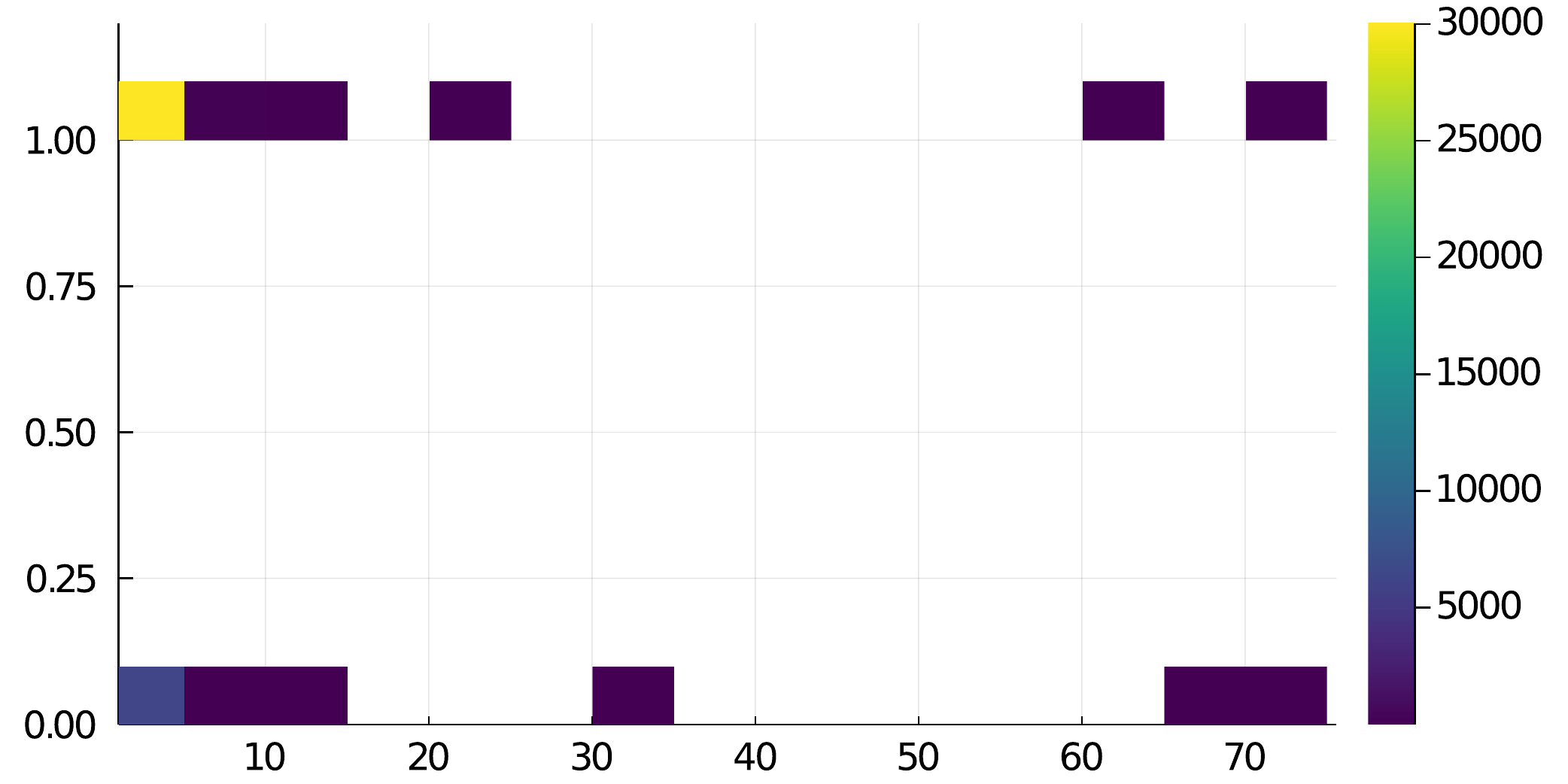}
     \end{subfigure}
     \ \
     \begin{subfigure}[b]{0.49\textwidth}
       \includegraphics[width=\textwidth]{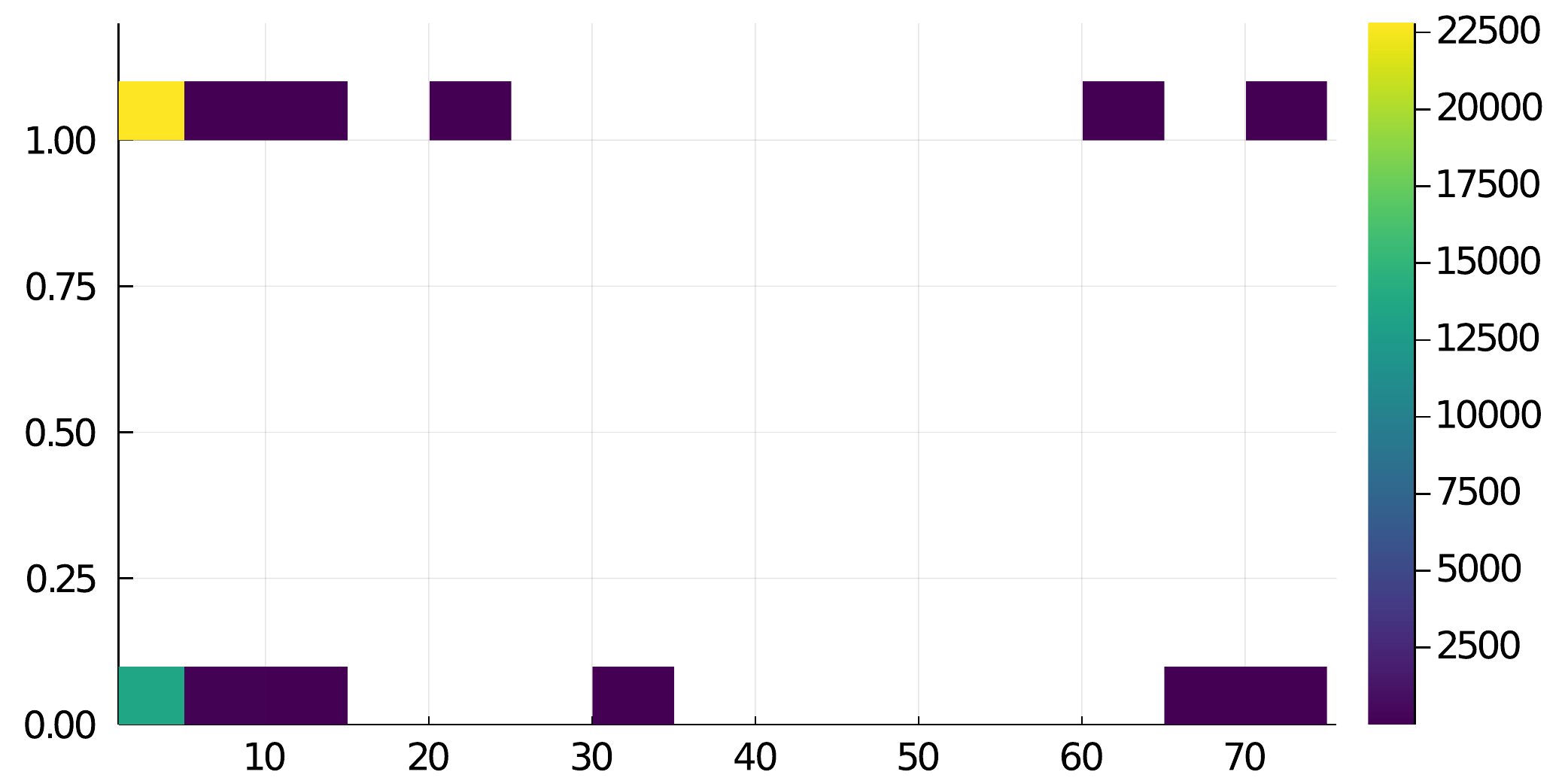}
     \end{subfigure}
\caption{Stability (left, OY axis) and groundedness (right, OY) by number of returns in method instances (OX)}%
\Description{Stability and groundedness by number of returns in method instances in JuMP}%
\label{figs:returns:JuMP}
\end{figure}
\clearpage
\subsection{Package: Knet}
\begin{figure}[h]
     \begin{subfigure}[b]{0.49\textwidth}
       \includegraphics[width=\textwidth]{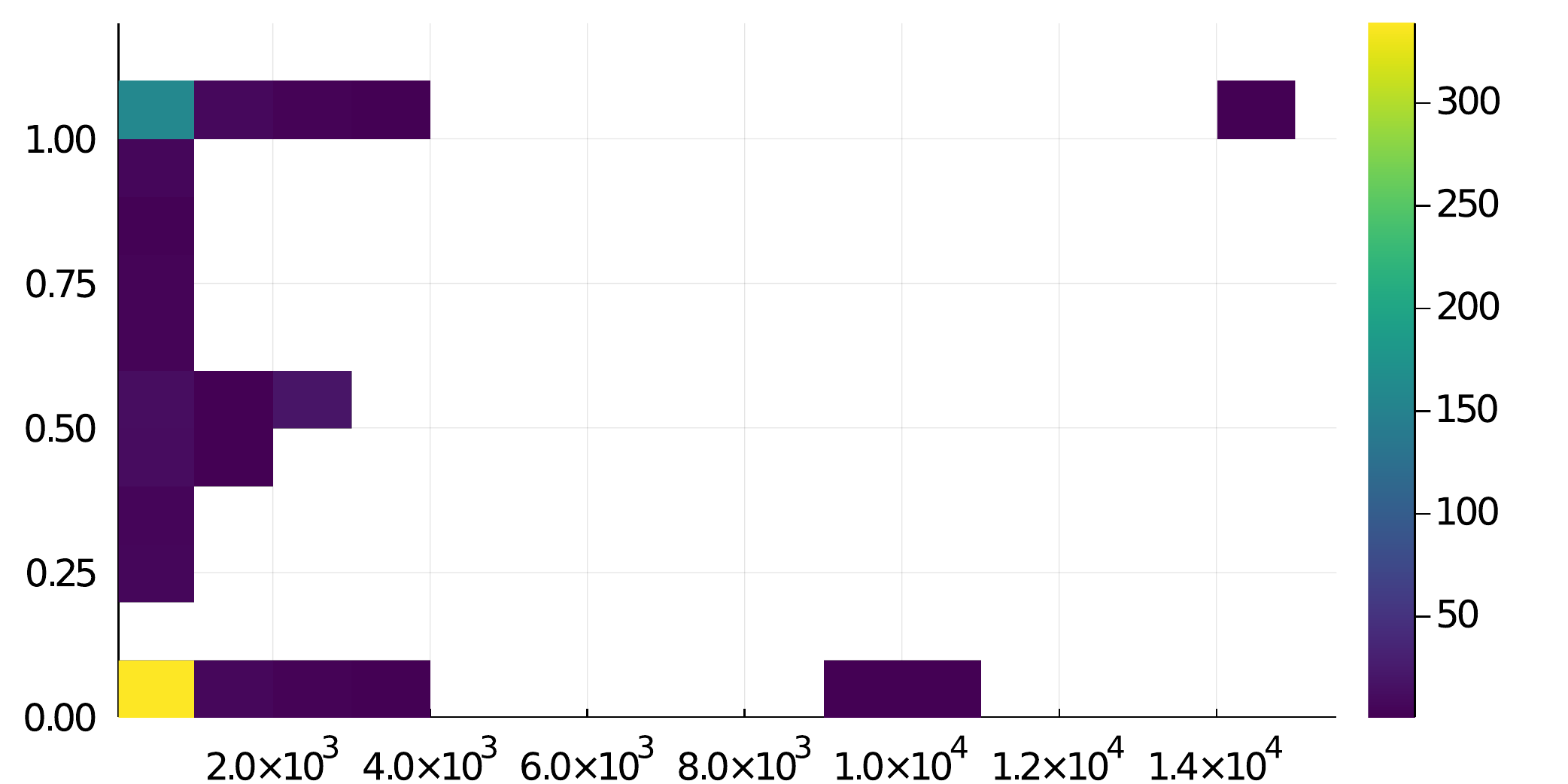}
     \end{subfigure}
     \ \
     \begin{subfigure}[b]{0.49\textwidth}
       \includegraphics[width=\textwidth]{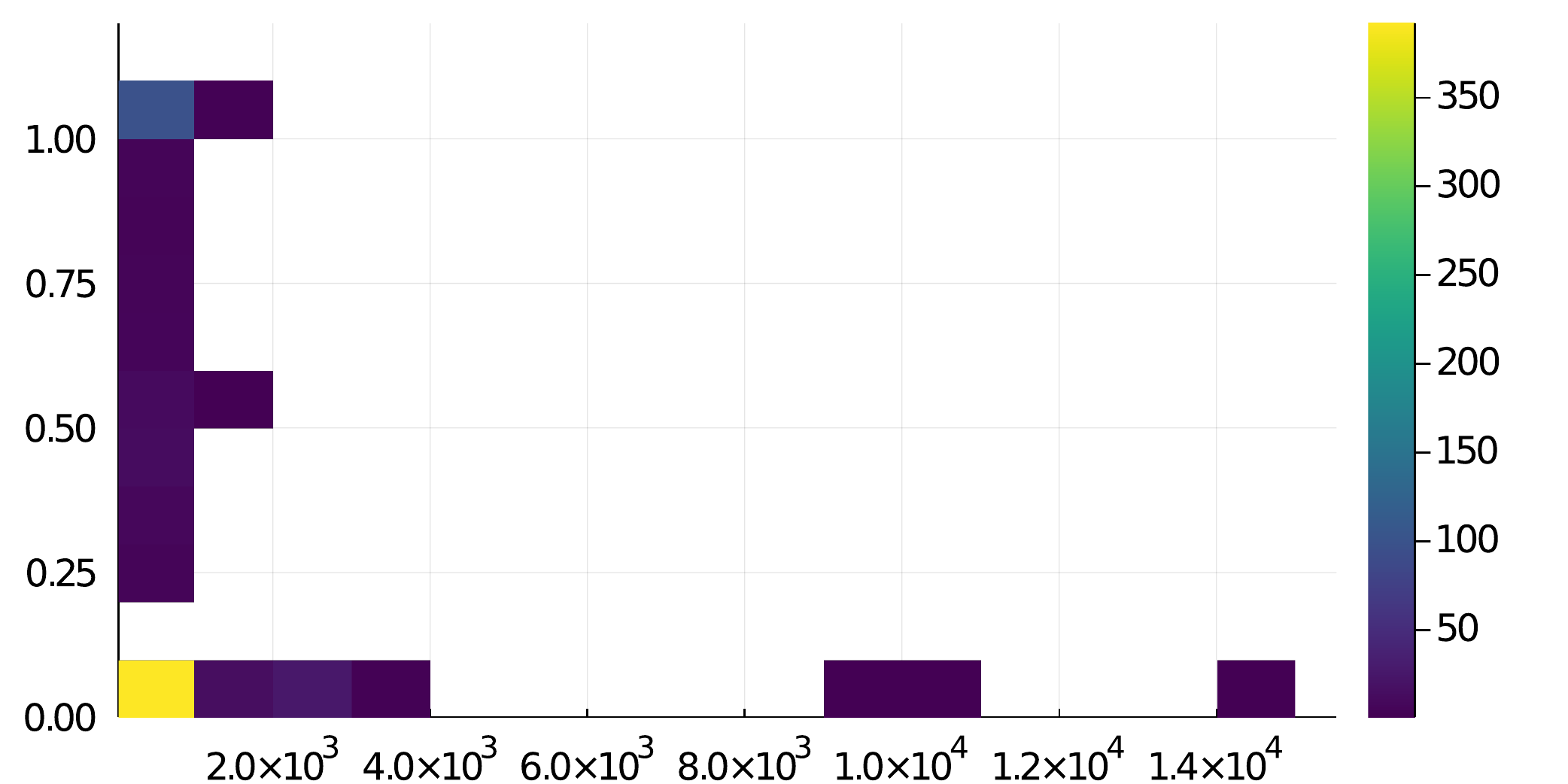}
     \end{subfigure}
\caption{Stability (left, OY axis) and groundedness (right, OY) by method size (OX)}%
\Description{Stability and groundedness by method size in Knet}%
\label{figs:size:Knet}
\end{figure}

\begin{figure}[h]
     \begin{subfigure}[b]{0.49\textwidth}
       \includegraphics[width=\textwidth]{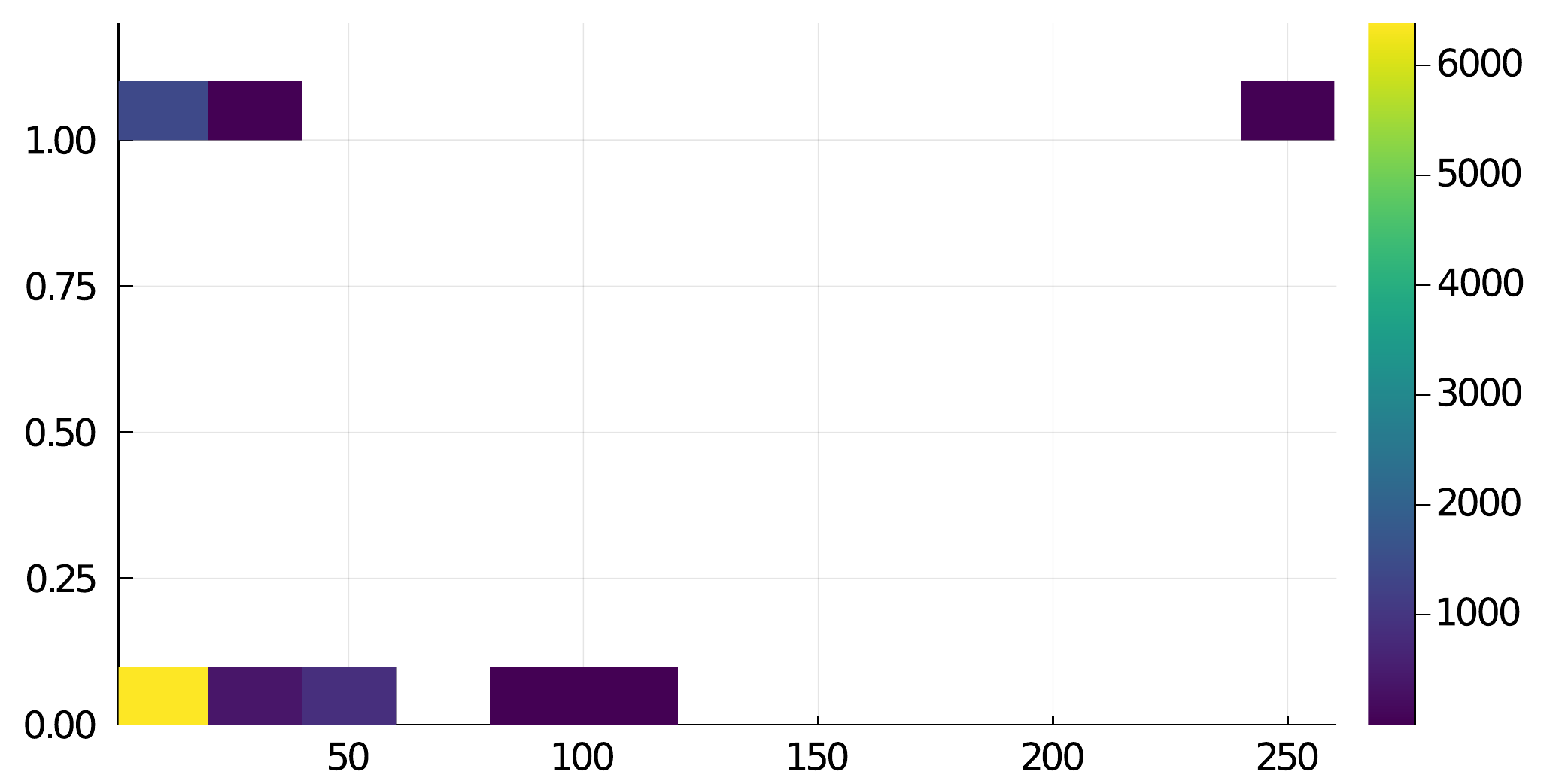}
     \end{subfigure}
     \ \
     \begin{subfigure}[b]{0.49\textwidth}
       \includegraphics[width=\textwidth]{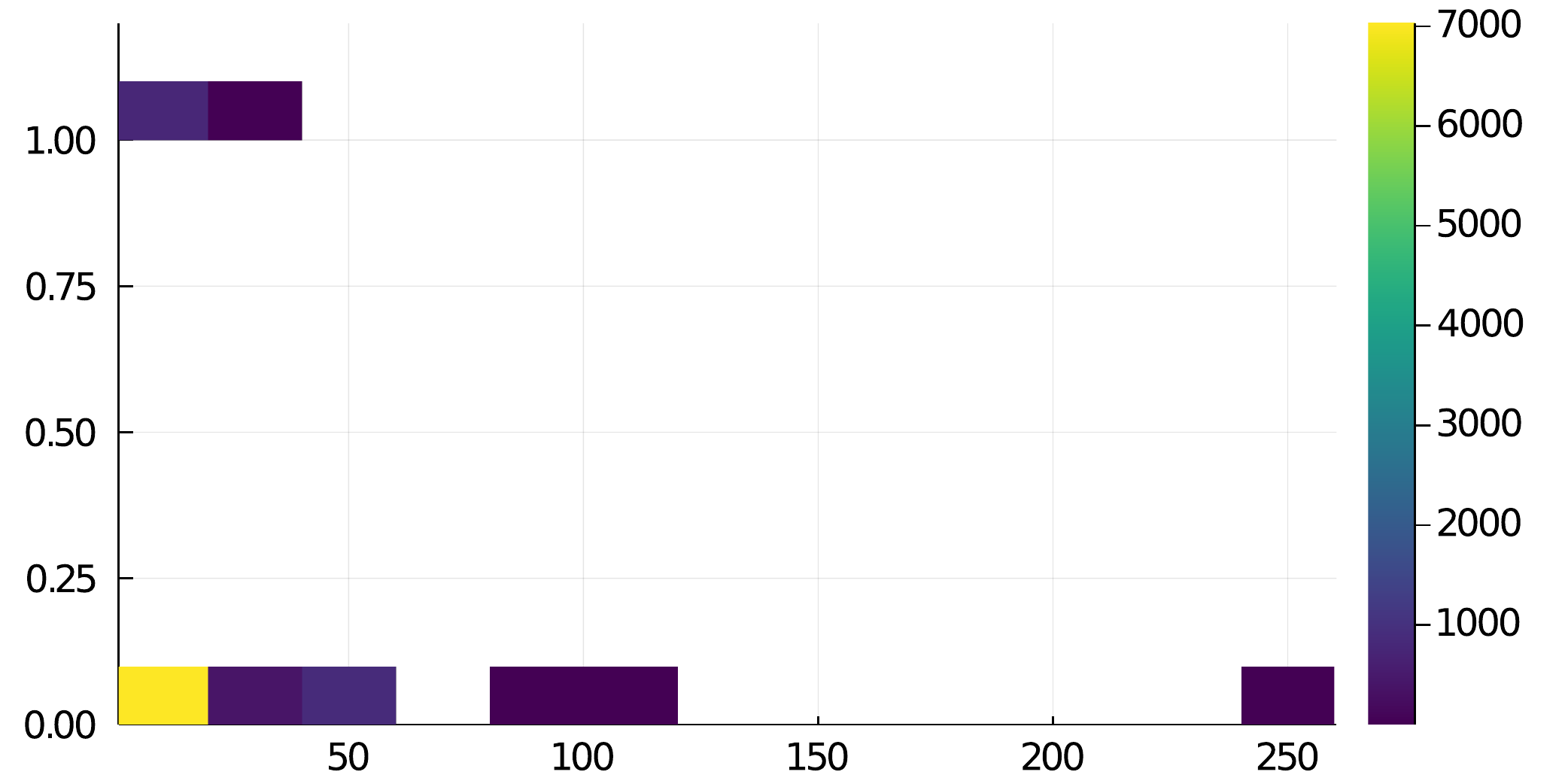}
     \end{subfigure}
\caption{Stability (left, OY axis) and groundedness (right, OY) by number of gotos in method instances (OX)}%
\Description{Stability and groundedness by number of goto's in method instances in Knet}%
\label{figs:gotos:Knet}
\end{figure}

\begin{figure}[h]
     \begin{subfigure}[b]{0.49\textwidth}
       \includegraphics[width=\textwidth]{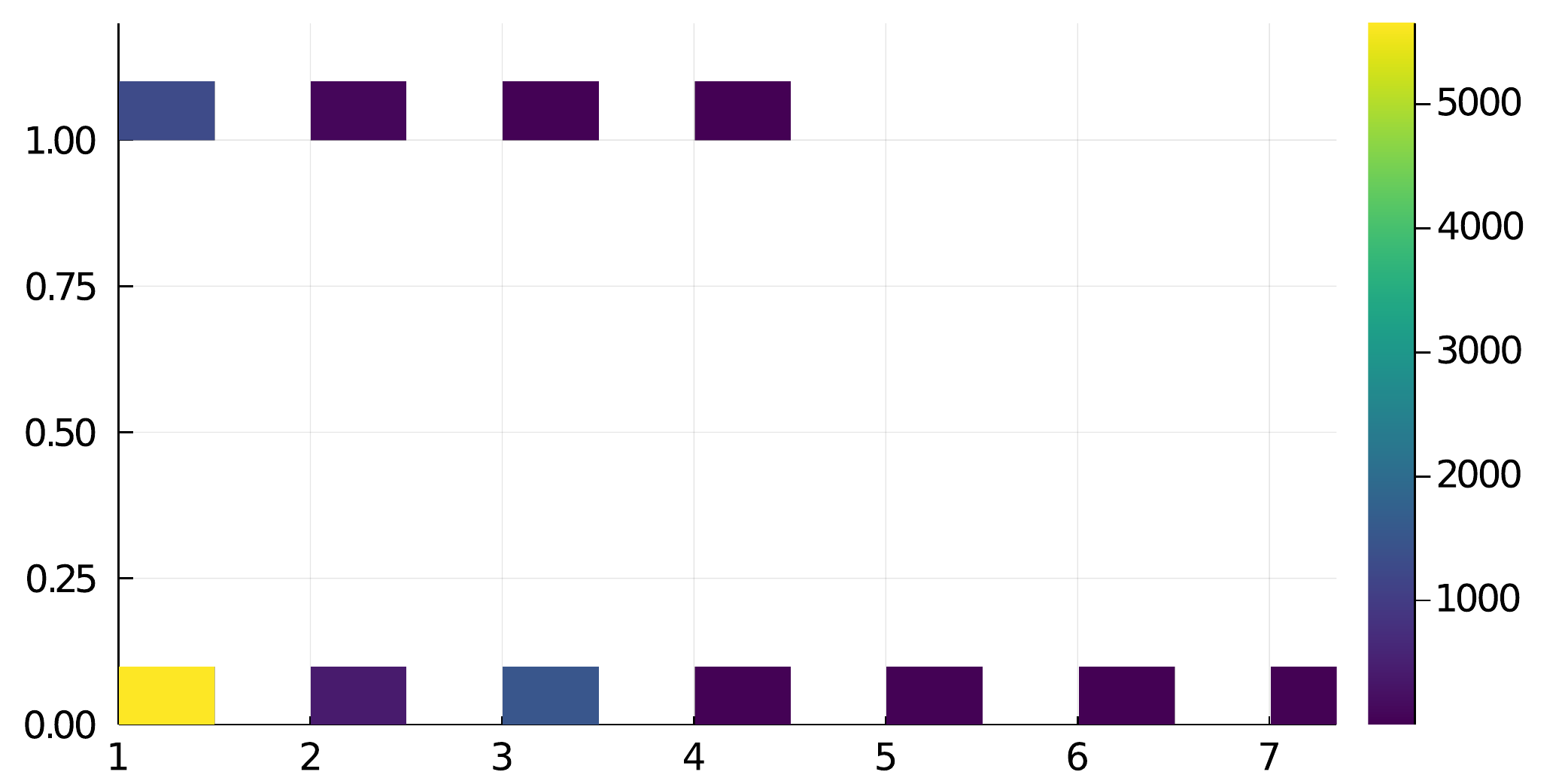}
     \end{subfigure}
     \ \
     \begin{subfigure}[b]{0.49\textwidth}
       \includegraphics[width=\textwidth]{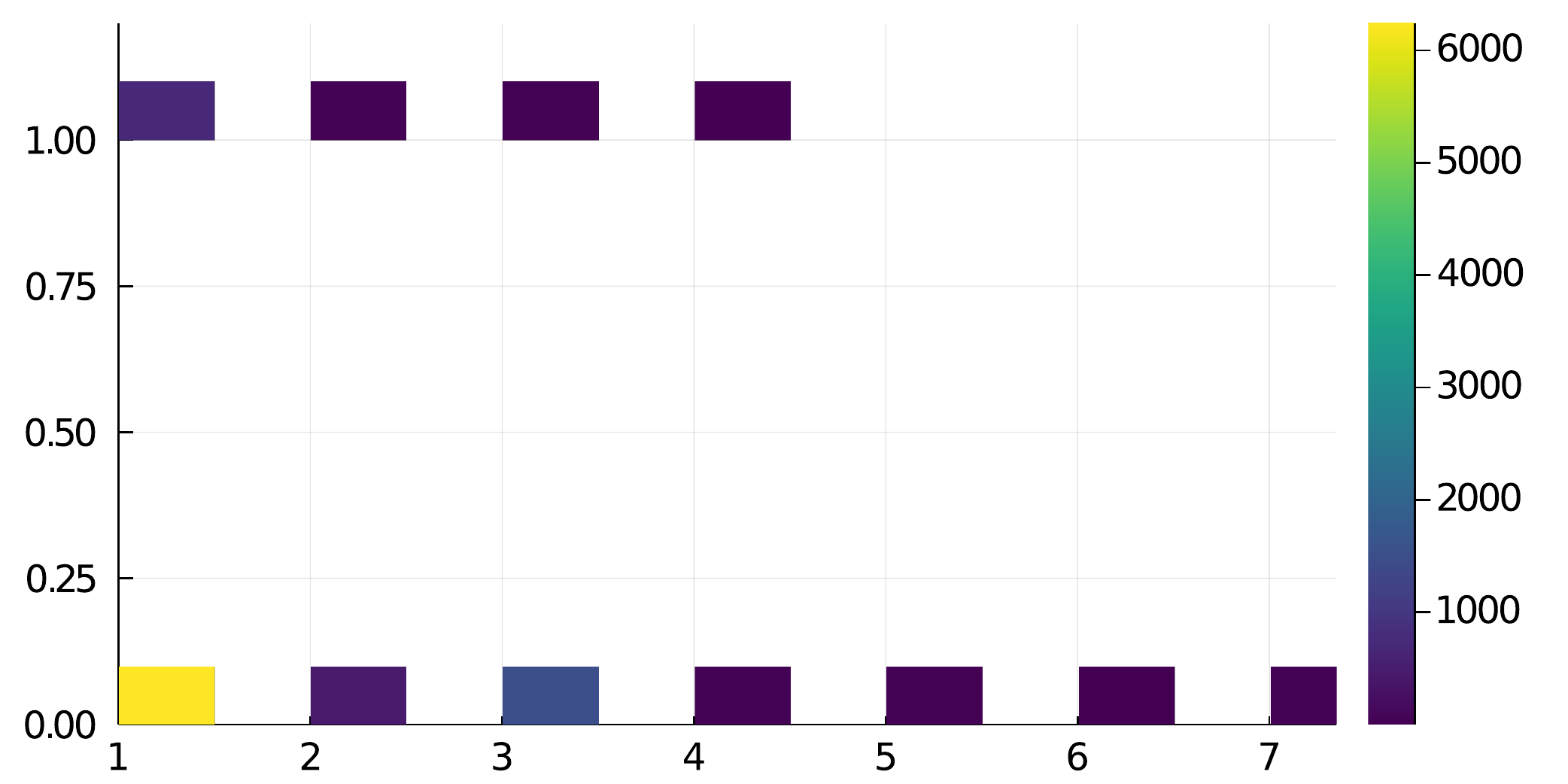}
     \end{subfigure}
\caption{Stability (left, OY axis) and groundedness (right, OY) by number of returns in method instances (OX)}%
\Description{Stability and groundedness by number of returns in method instances in Knet}%
\label{figs:returns:Knet}
\end{figure}
\clearpage
\subsection{Package: Plots}
\begin{figure}[h]
     \begin{subfigure}[b]{0.49\textwidth}
       \includegraphics[width=\textwidth]{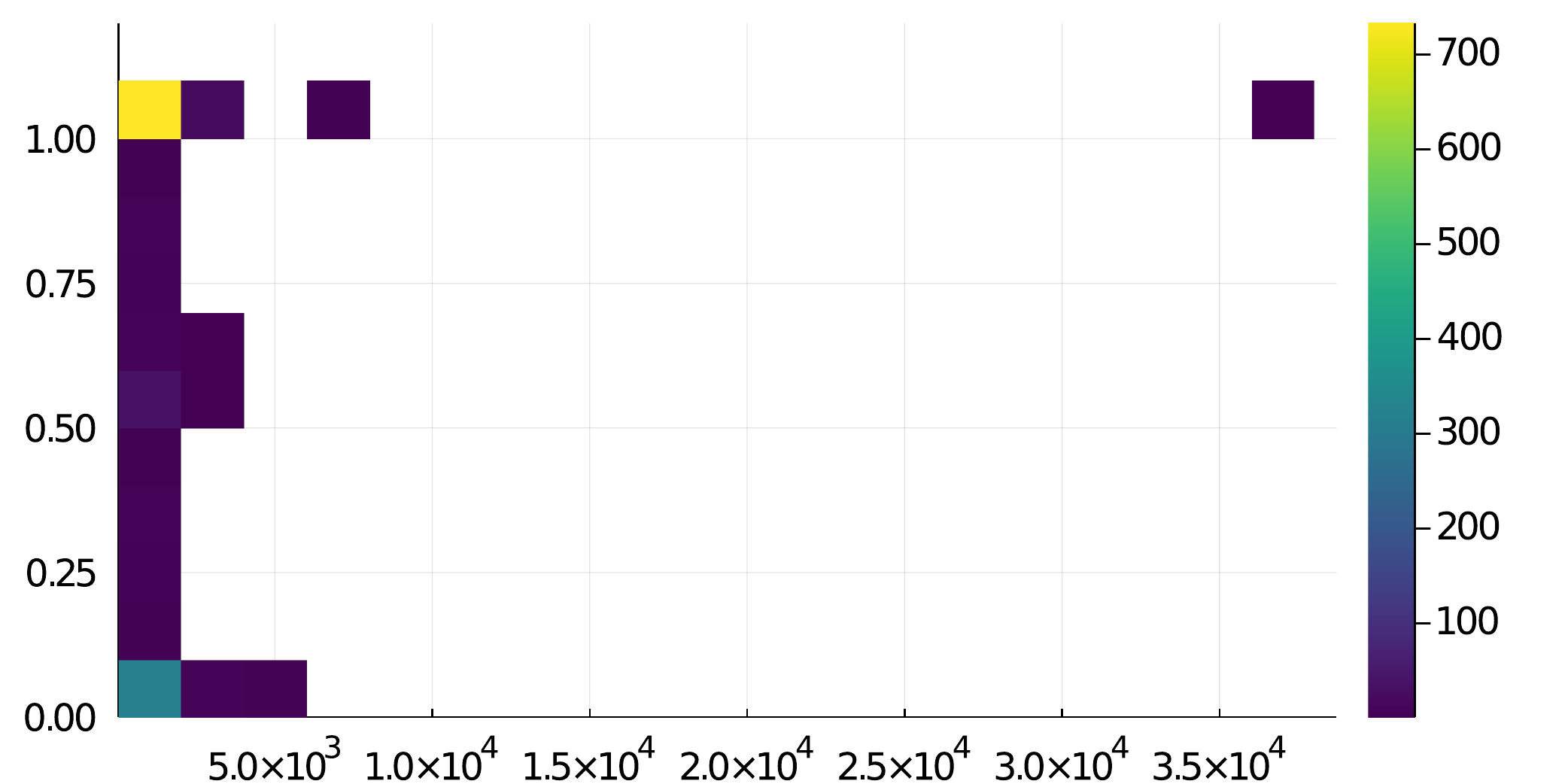}
     \end{subfigure}
     \ \
     \begin{subfigure}[b]{0.49\textwidth}
       \includegraphics[width=\textwidth]{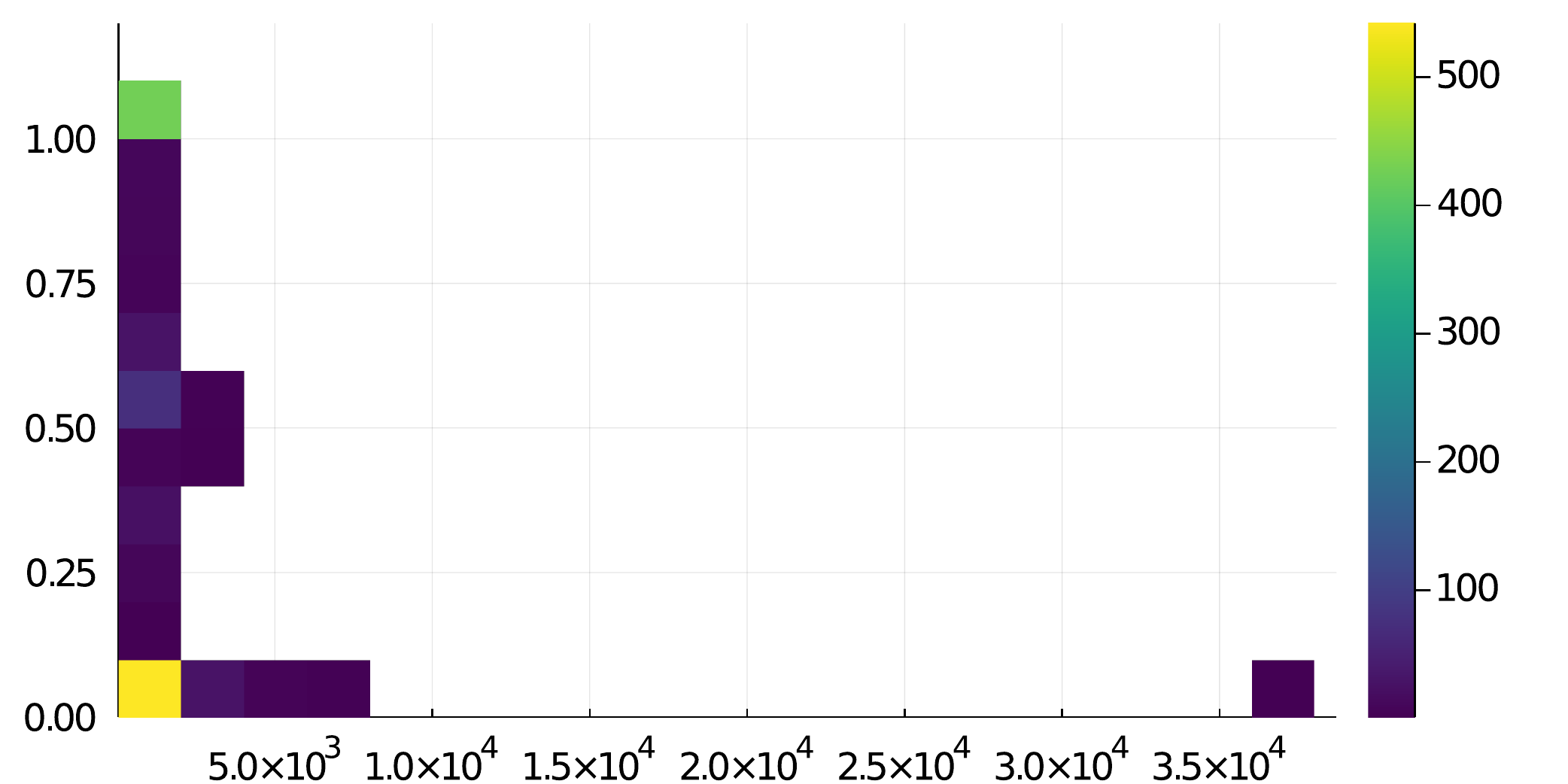}
     \end{subfigure}
\caption{Stability (left, OY axis) and groundedness (right, OY) by method size (OX)}%
\Description{Stability and groundedness by method size in Plots}%
\label{figs:size:Plots}
\end{figure}

\begin{figure}[h]
     \begin{subfigure}[b]{0.49\textwidth}
       \includegraphics[width=\textwidth]{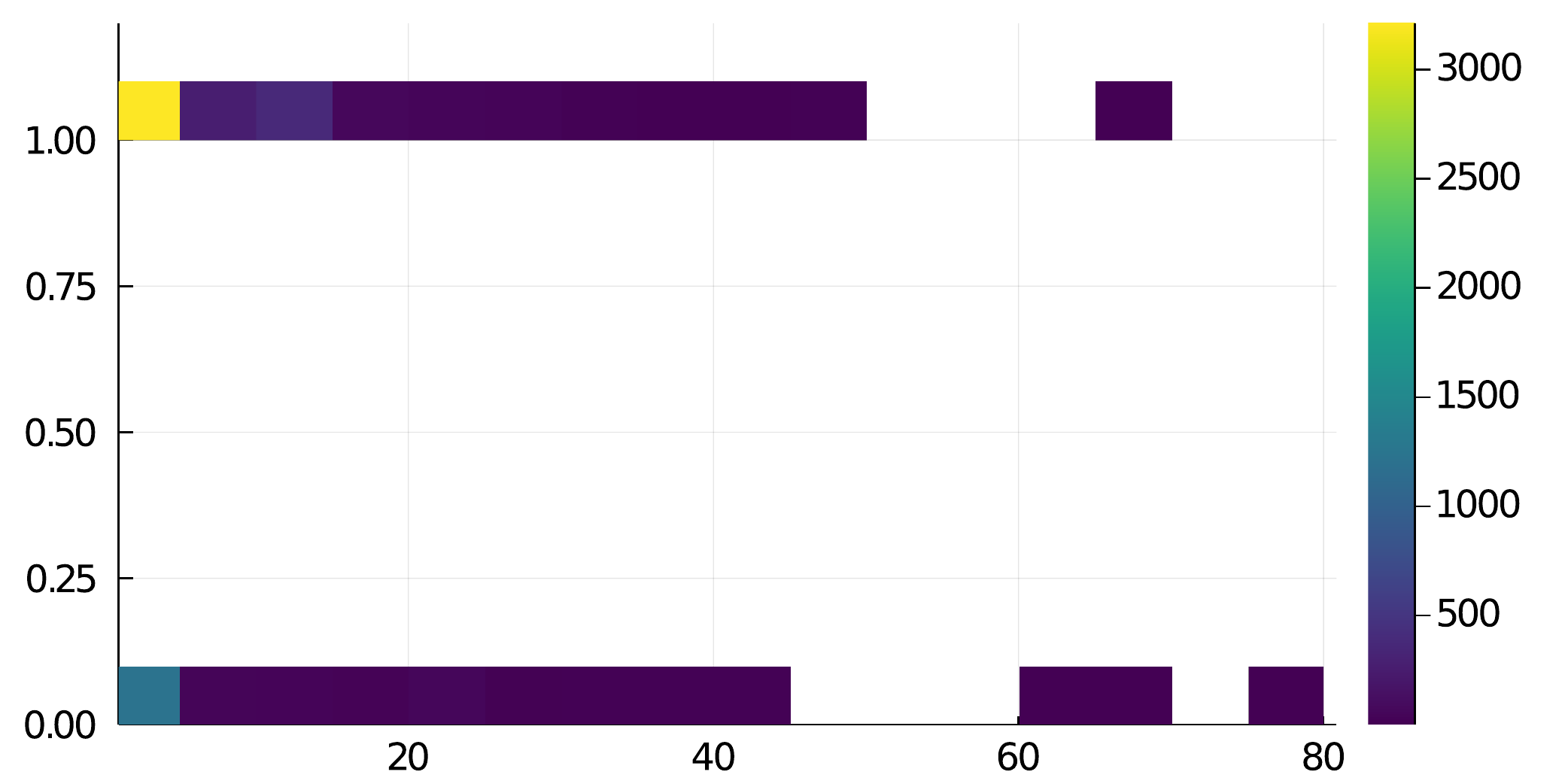}
     \end{subfigure}
     \ \
     \begin{subfigure}[b]{0.49\textwidth}
       \includegraphics[width=\textwidth]{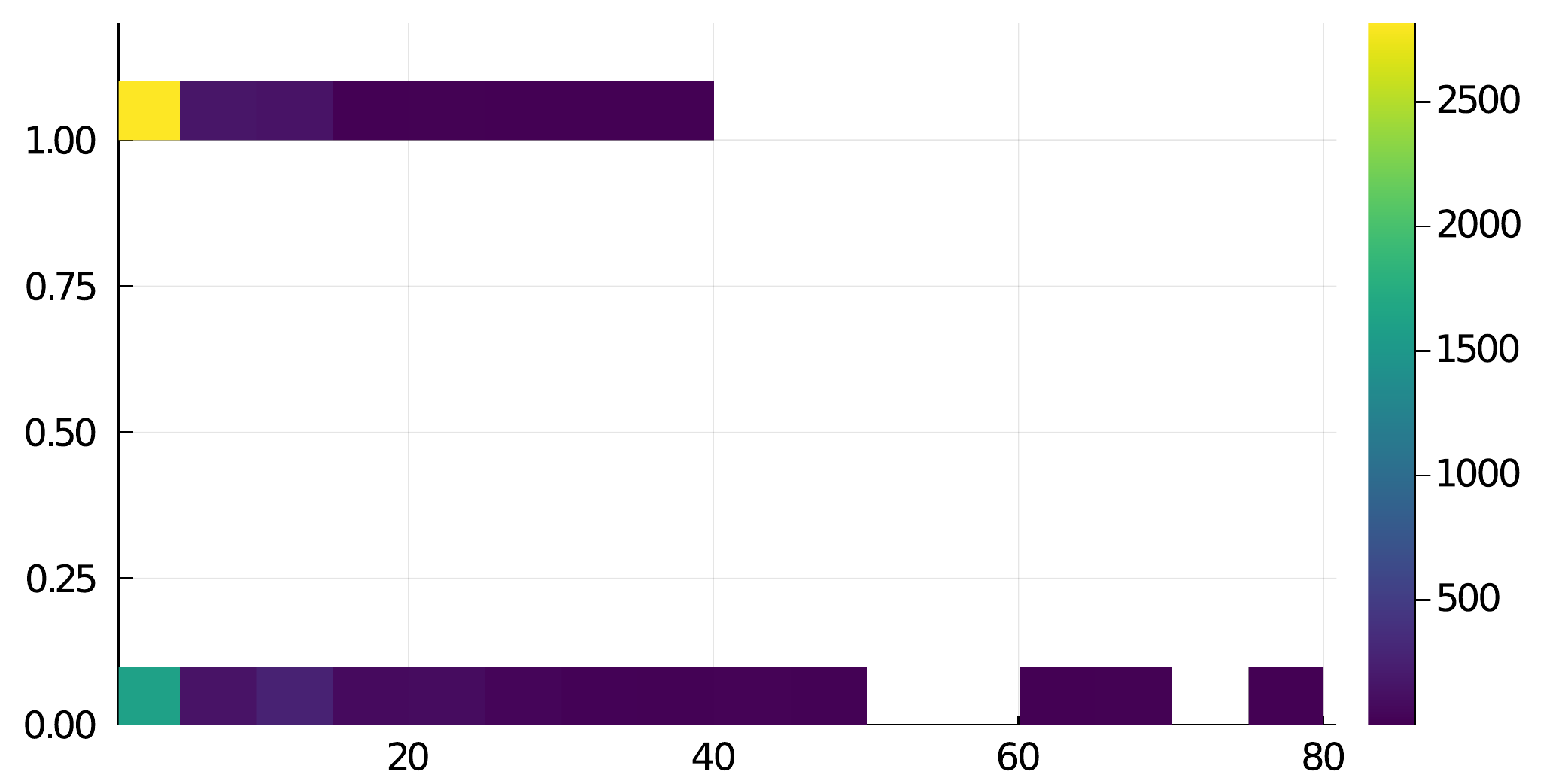}
     \end{subfigure}
\caption{Stability (left, OY axis) and groundedness (right, OY) by number of gotos in method instances (OX)}%
\Description{Stability and groundedness by number of goto's in method instances in Plots}%
\label{figs:gotos:Plots}
\end{figure}

\begin{figure}[h]
     \begin{subfigure}[b]{0.49\textwidth}
       \includegraphics[width=\textwidth]{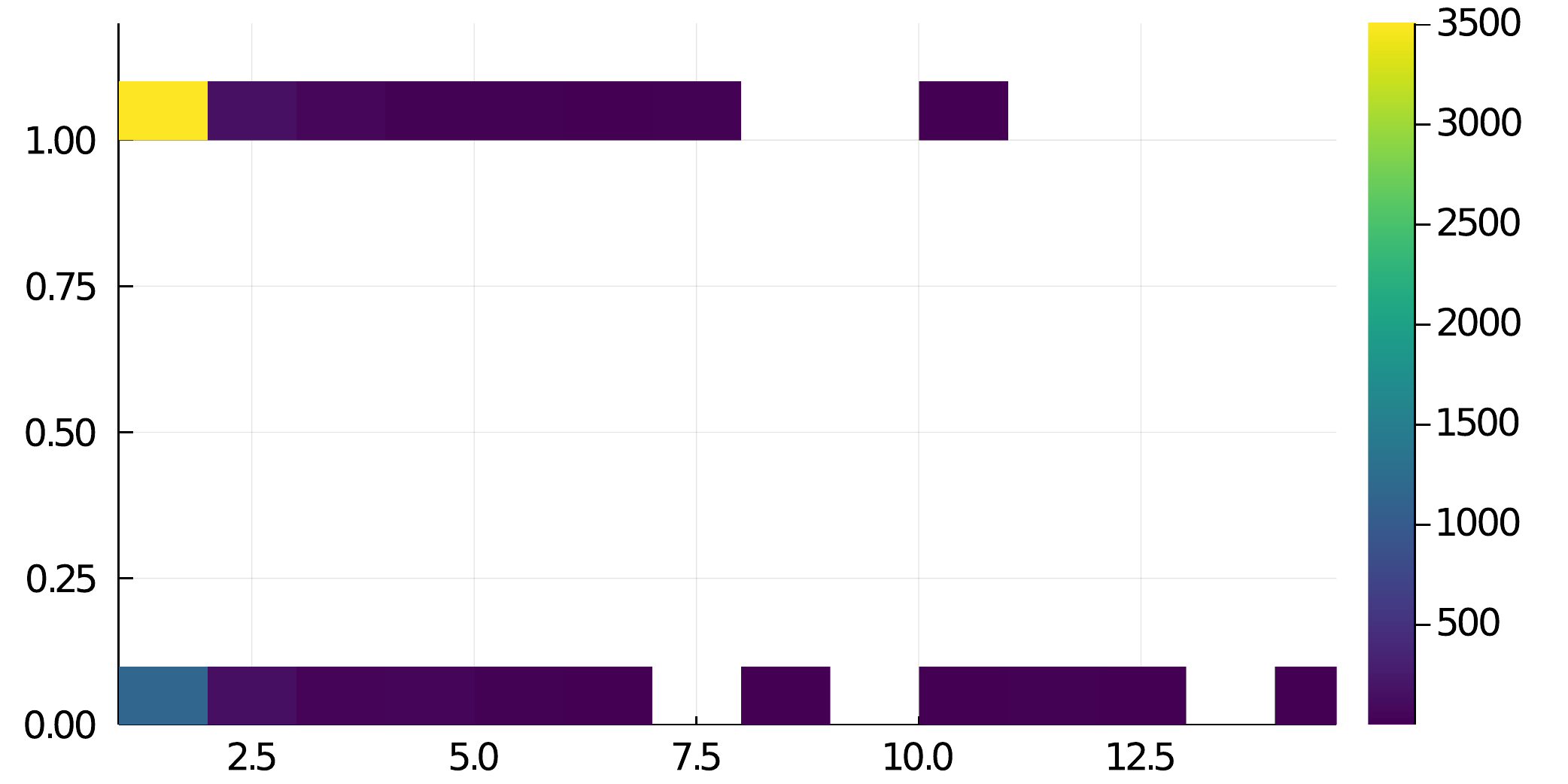}
     \end{subfigure}
     \ \
     \begin{subfigure}[b]{0.49\textwidth}
       \includegraphics[width=\textwidth]{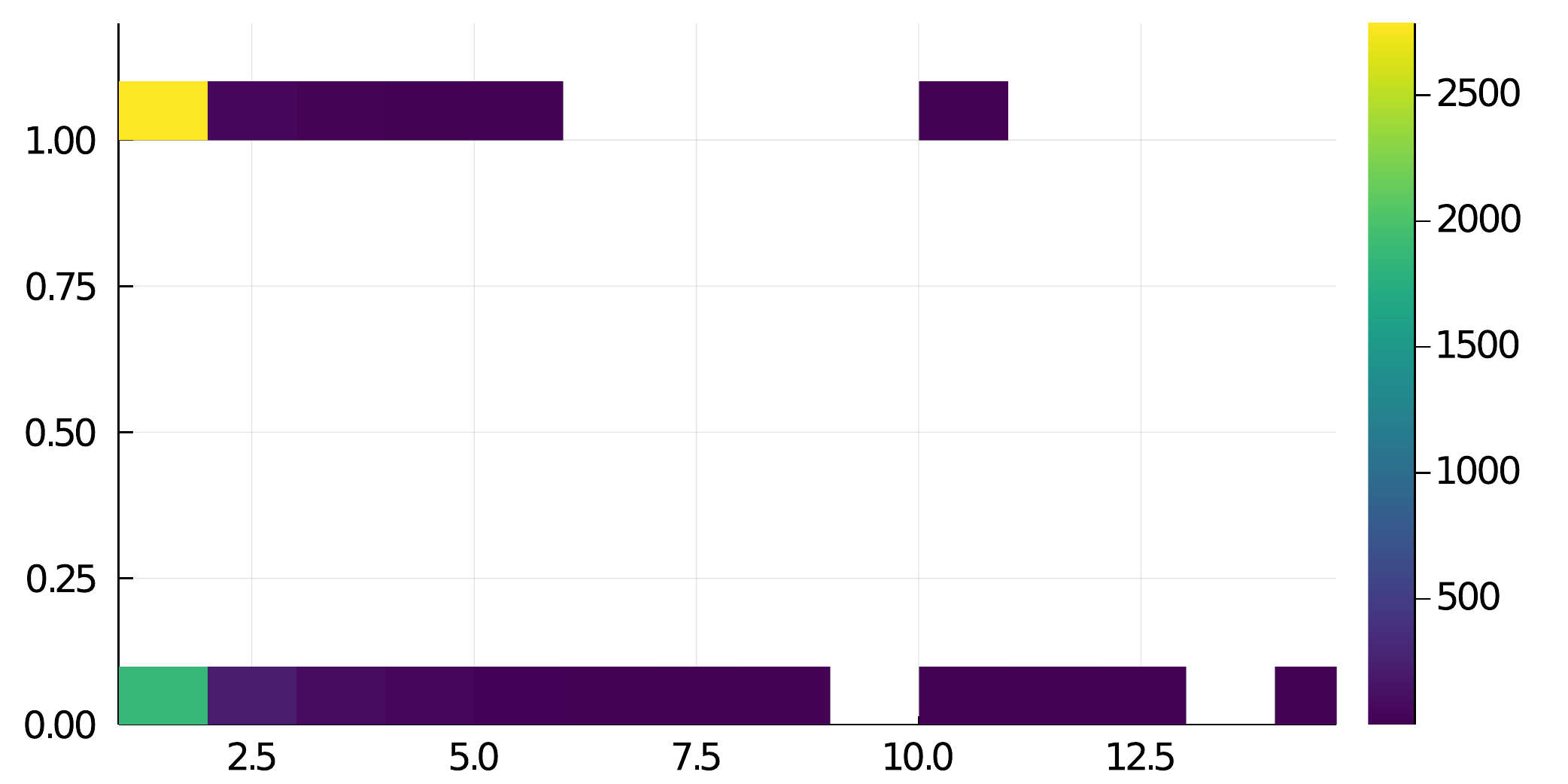}
     \end{subfigure}
\caption{Stability (left, OY axis) and groundedness (right, OY) by number of returns in method instances (OX)}%
\Description{Stability and groundedness by number of returns in method instances in Plots}%
\label{figs:returns:Plots}
\end{figure}
\clearpage
\subsection{Package: Pluto}
\begin{figure}[h]
     \begin{subfigure}[b]{0.49\textwidth}
       \includegraphics[width=\textwidth]{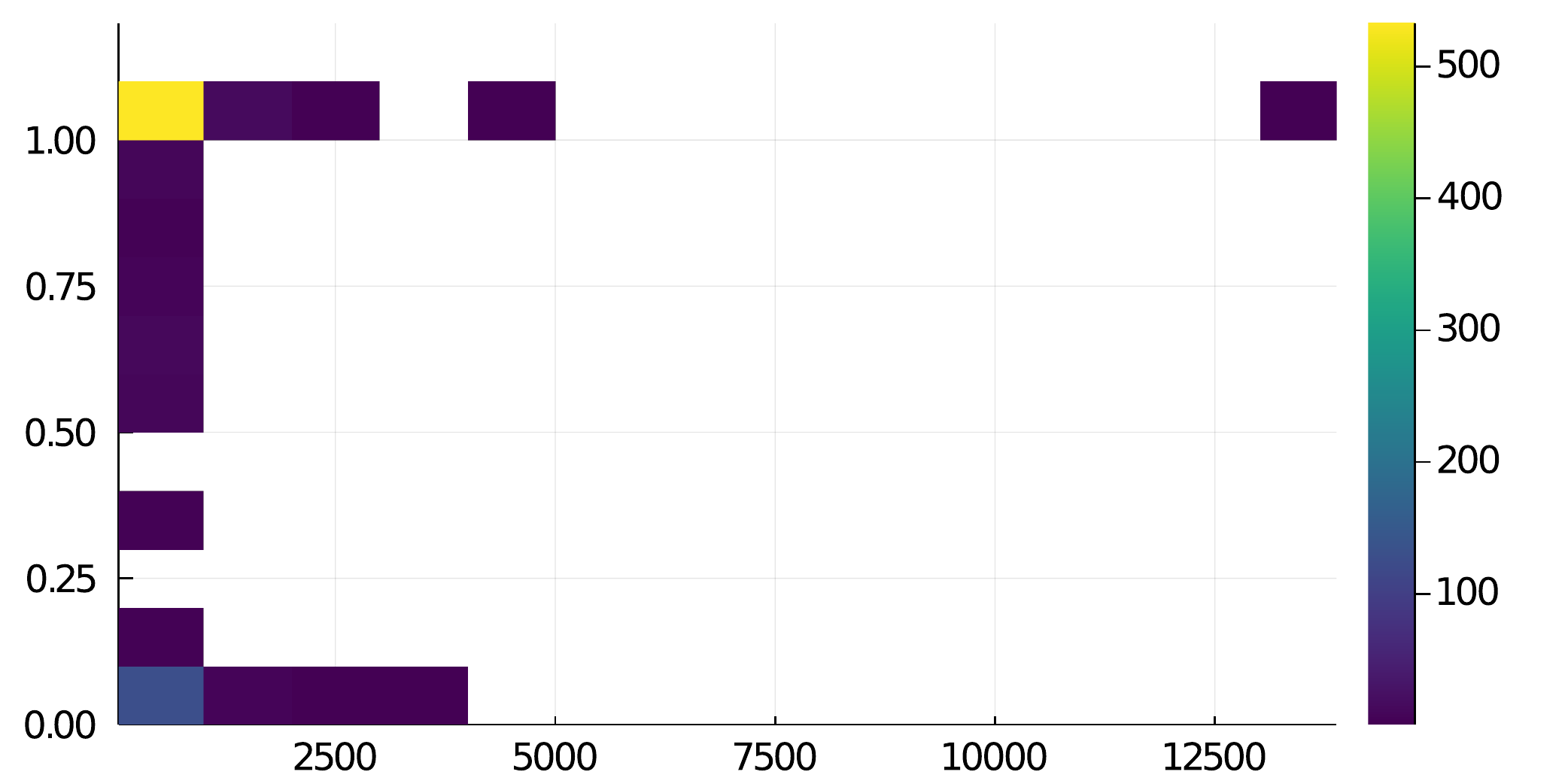}
     \end{subfigure}
     \ \
     \begin{subfigure}[b]{0.49\textwidth}
       \includegraphics[width=\textwidth]{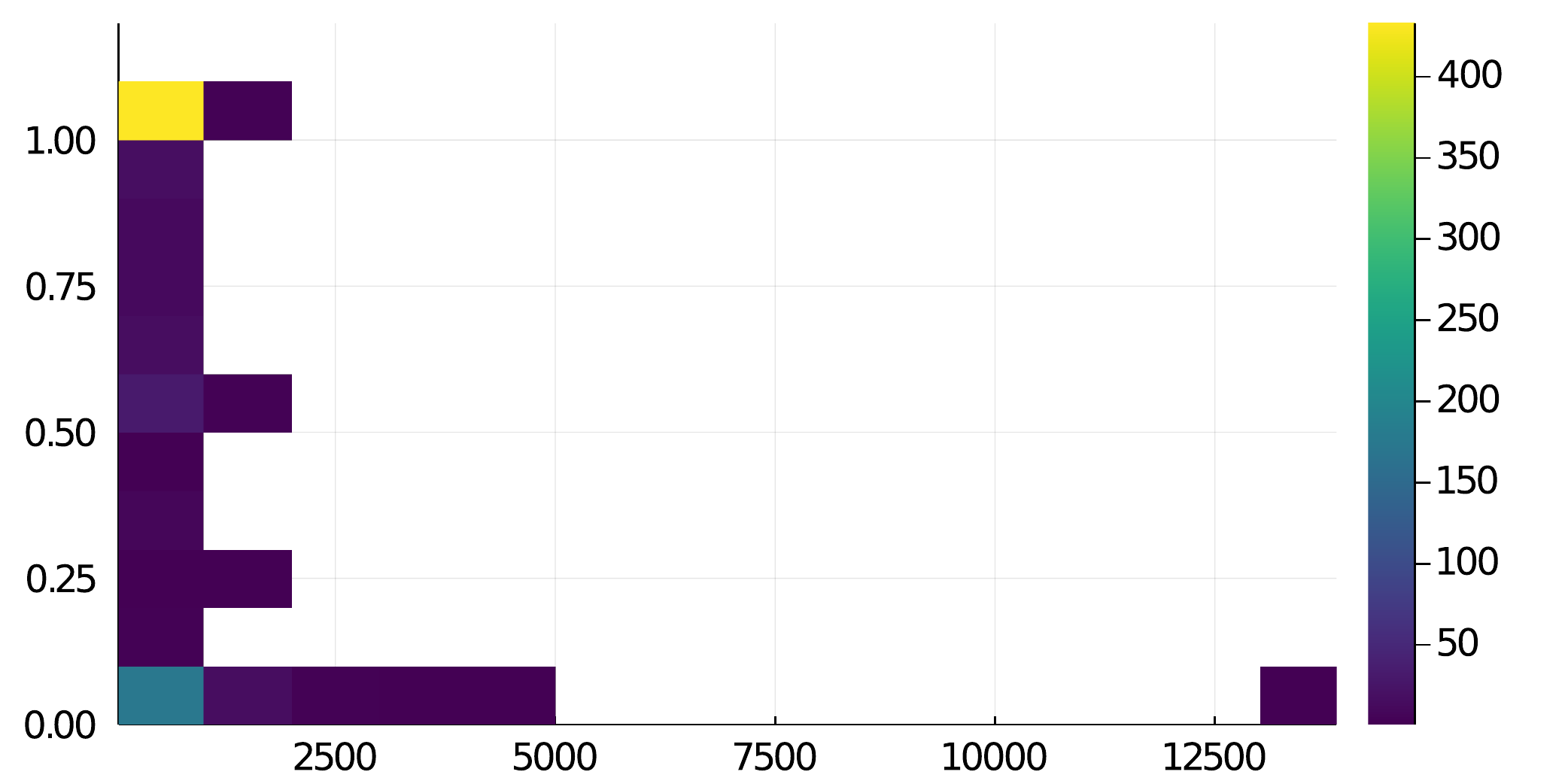}
     \end{subfigure}
\caption{Stability (left, OY axis) and groundedness (right, OY) by method size (OX)}%
\Description{Stability and groundedness by method size in Pluto}%
\label{figs:size:Pluto}
\end{figure}

\begin{figure}[h]
     \begin{subfigure}[b]{0.49\textwidth}
       \includegraphics[width=\textwidth]{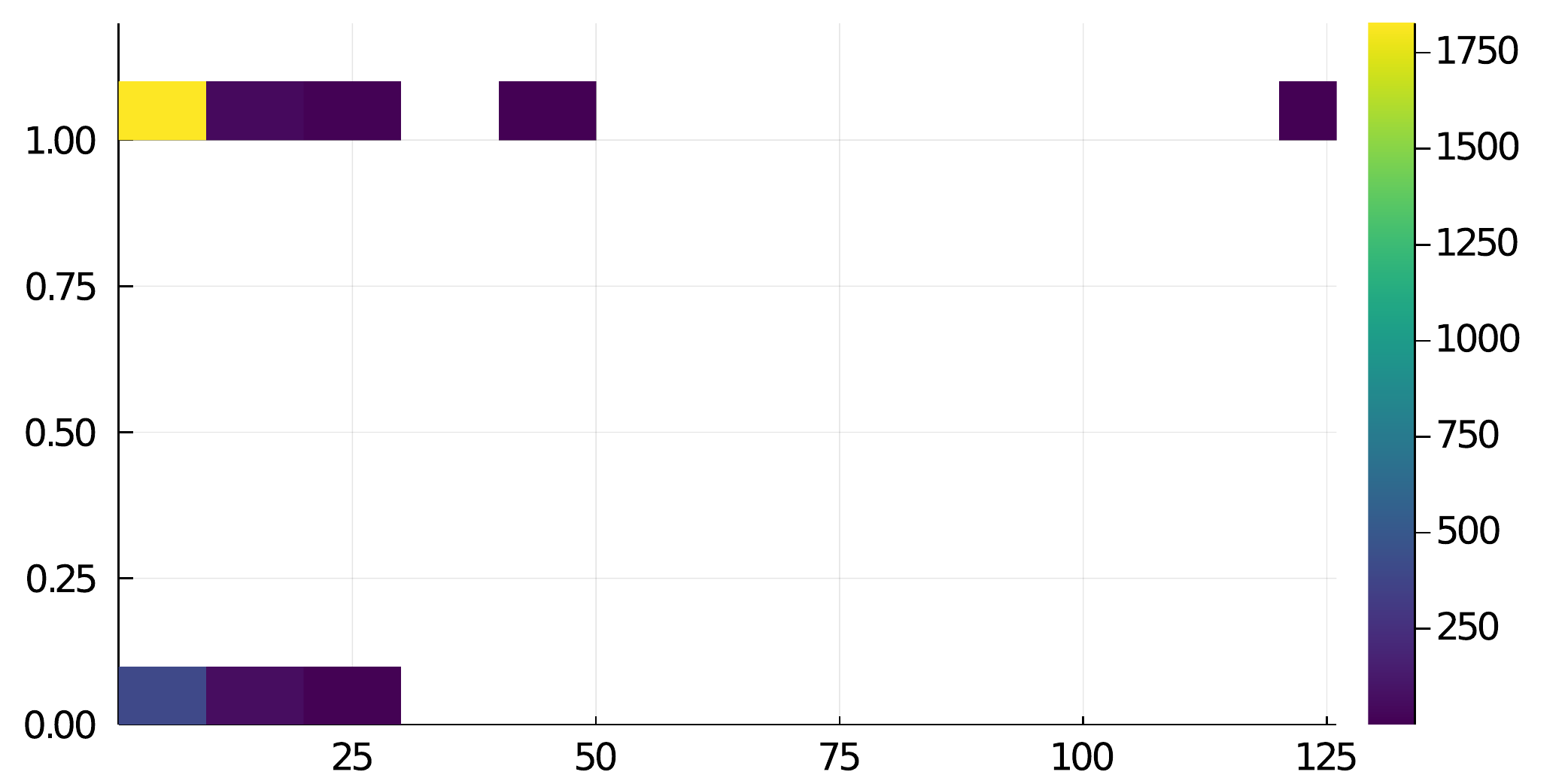}
     \end{subfigure}
     \ \
     \begin{subfigure}[b]{0.49\textwidth}
       \includegraphics[width=\textwidth]{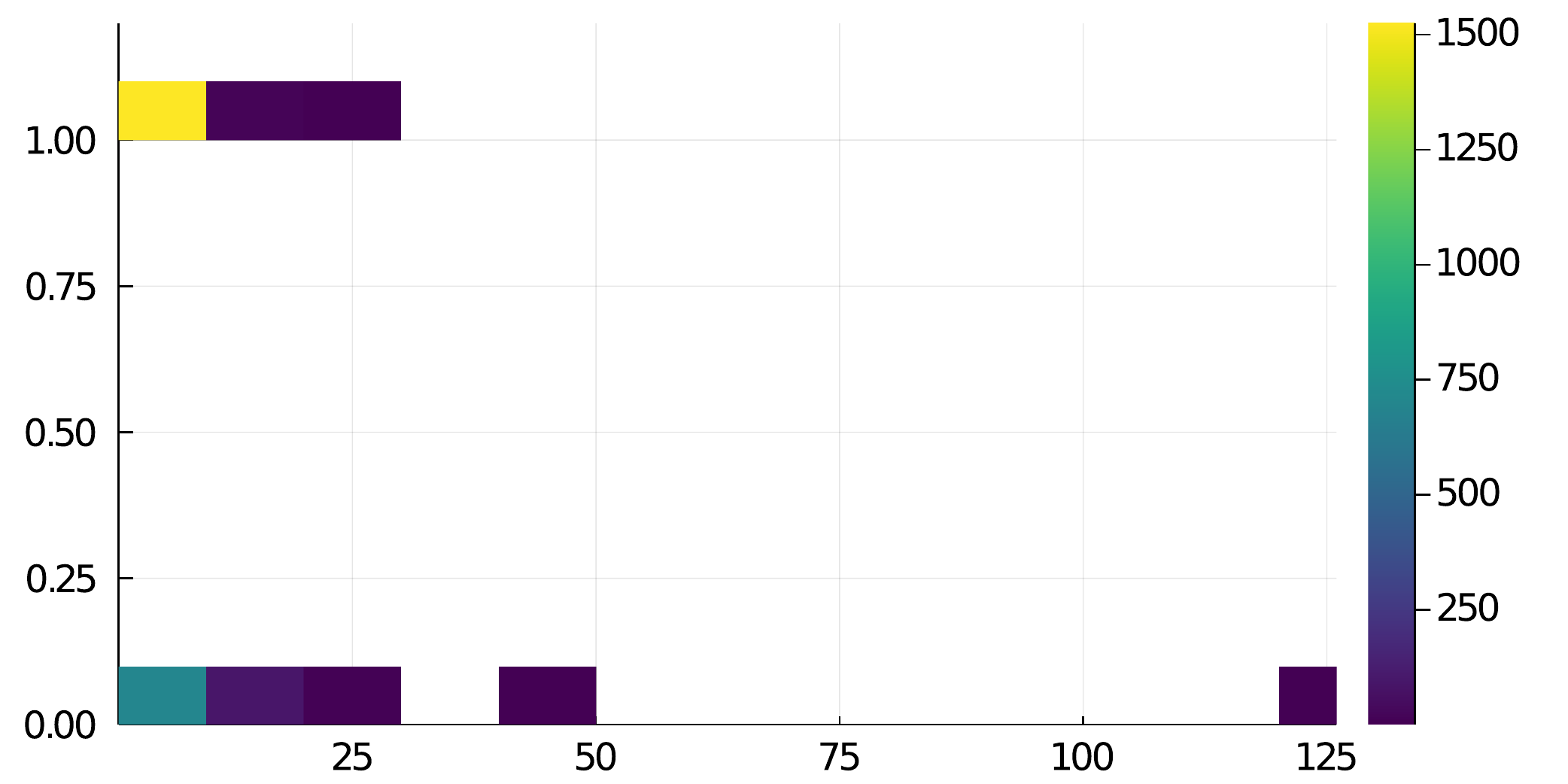}
     \end{subfigure}
\caption{Stability (left, OY axis) and groundedness (right, OY) by number of gotos in method instances (OX)}%
\Description{Stability and groundedness by number of goto's in method instances in Pluto}%
\label{figs:gotos:Pluto}
\end{figure}

\begin{figure}[h]
     \begin{subfigure}[b]{0.49\textwidth}
       \includegraphics[width=\textwidth]{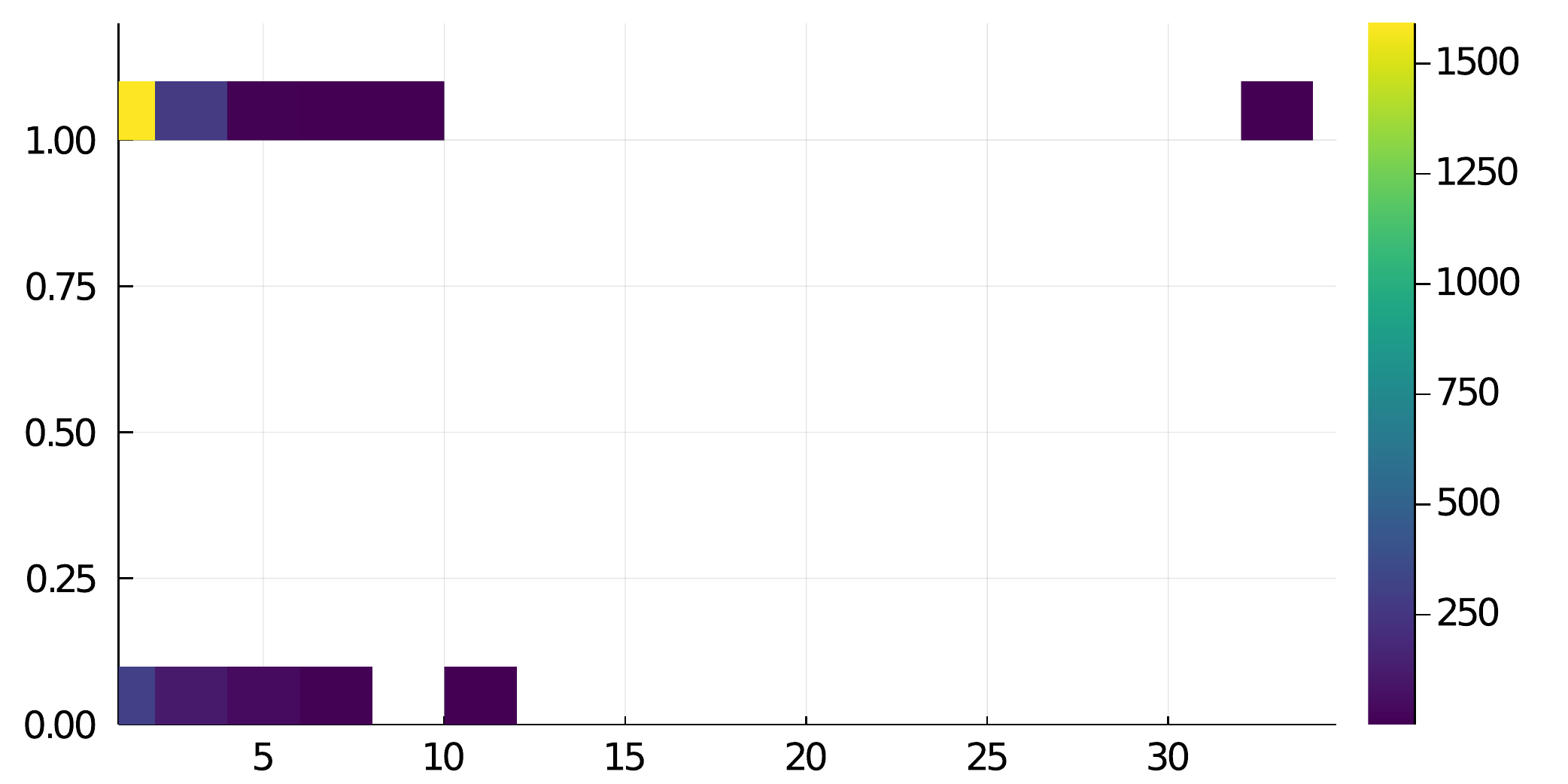}
     \end{subfigure}
     \ \
     \begin{subfigure}[b]{0.49\textwidth}
       \includegraphics[width=\textwidth]{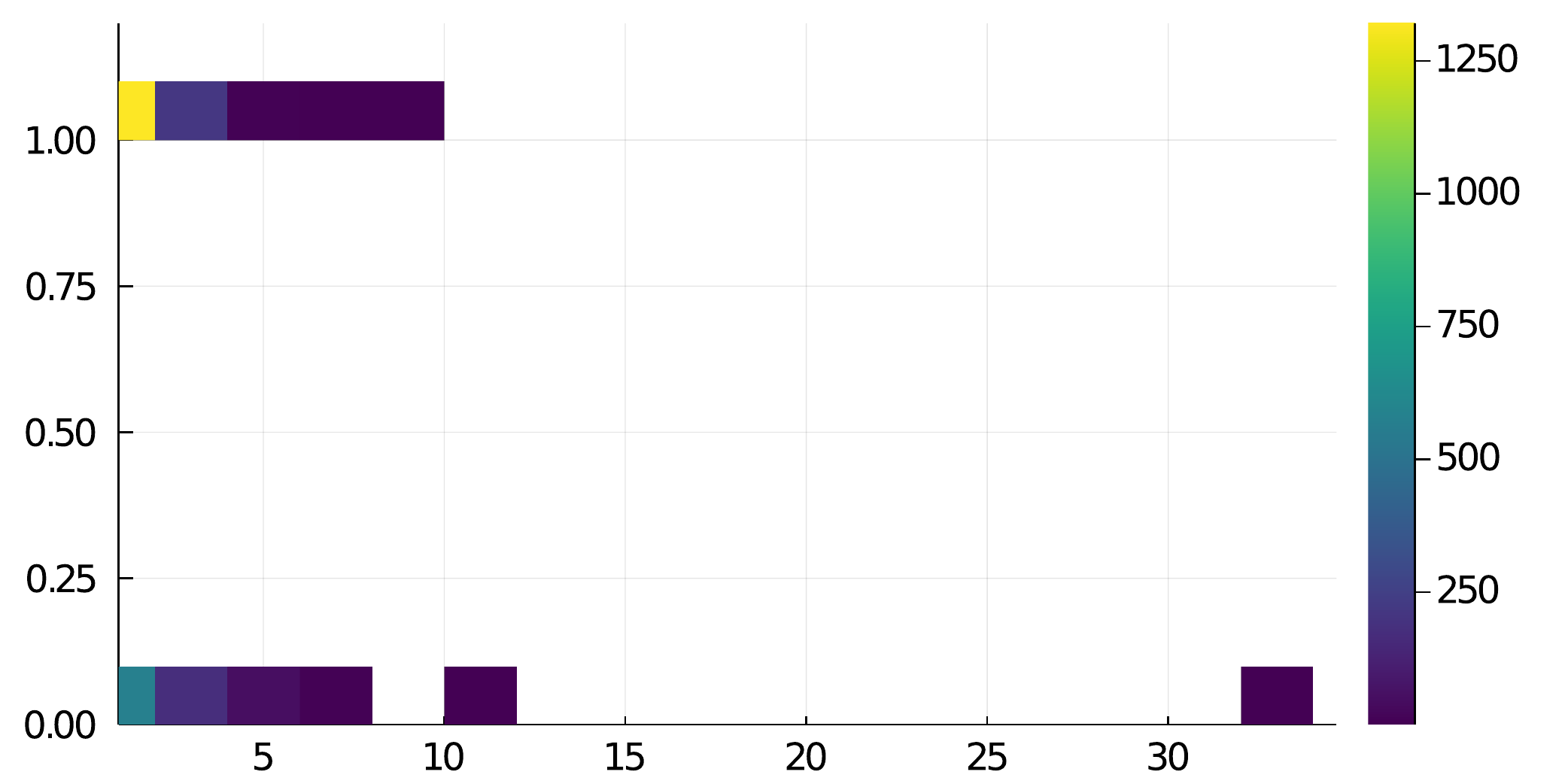}
     \end{subfigure}
\caption{Stability (left, OY axis) and groundedness (right, OY) by number of returns in method instances (OX)}%
\Description{Stability and groundedness by number of returns in method instances in Pluto}%
\label{figs:returns:Pluto}
\end{figure}
\clearpage

}

\end{document}